%% file: main.tex
\documentclass{article}
\usepackage[margin=2cm]{geometry}
\usepackage[T1]{fontenc}
\usepackage[utf8]{inputenc}
\usepackage{amsmath,amssymb,amsthm}
\usepackage[shortlabels]{enumitem}
\usepackage{color,url}
\usepackage{hyperref}
\usepackage[capitalise]{cleveref}
\usepackage{xspace}
\usepackage{todonotes}
\usepackage{thm-restate}
\usepackage{comment}
\usepackage{xcolor}
\usepackage{colortbl}
\usepackage{authblk}
\usepackage{libertine}
\usepackage{cite}
\usetikzlibrary{decorations.pathreplacing,calligraphy}

\theoremstyle{plain}
\newtheorem{theorem}{Theorem}[section]
\newtheorem{lemma}[theorem]{Lemma}

\newtheorem{observation}[theorem]{Observation}
\newtheorem{corollary}[theorem]{Corollary}

\newtheorem{definition}[theorem]{Definition}
\newtheorem{claim}{Claim}[theorem]
\newenvironment{claimproof}{\noindent\textit{Proof of Claim \theclaim:}}{\hfill$\lrcorner$\medskip} 

\usepackage{todonotes}

\input{macros}
\input{icons}

\begin{document}
\title{Maximum Weight Independent Set in Hereditary Classes of Ordered Graphs}
\author[1]{Paweł Rafał Bieliński}
\author[2]{Marta Piecyk}
\author[1]{Pawe{\l} Rz\k{a}\.zewski}

\affil[1]{Warsaw University of Technology}
\affil[2]{CISPA Helmholtz Center for Information Security}
\date{}

\maketitle

\begin{abstract}
    The complexity of classical computational problems in graph classes defined by forbidding induced subgraphs is one of the central topics of algorithmic graph theory.
    Recently, there has been a growing interest in the complexity of such problems in \emph{ordered graphs}, i.e., graphs with a fixed linear ordering of vertices.
    Such an approach allows us to investigate the boundary of tractability more closely.
    However, most results so far concern coloring problems.    
    
    In this paper, we focus on the complexity of the \textsc{Maximum Weight Independent Set} (\textsc{MWIS}) problem in classes of ordered graphs.
    For every ordered graph $H$, we classify the complexity of \textsc{MWIS} in ordered graphs that exclude $H$ as an induced subgraph into one of the following cases:
    \begin{enumerate}[(1)]
        \item solvable in polynomial time,
        \item solvable in quasipolynomial time,
        \item solvable in subexponential time,
        \item \textsf{NP}-hard.
    \end{enumerate}
    Notably, case (3) contains only one well-structured family of $H$ obtained from two nested edges by adding isolated vertices in a specific way.
    Thus, our results yield an almost complete complexity dichotomy for \textsc{MWIS} in classes of ordered graphs defined by a single forbidden induced subgraph into cases solvable in quasipolynomial time and those that are \textsf{NP}-hard.
\end{abstract}

\section{Introduction}
\blfootnote{PRB and PRz were supported by the National Science Centre grant number 2024/54/E/ST6/00094.}
\input{intro}\label{sec:intro}

\section{Preliminaries}\label{sec:prelim}
\input{prelim}

\section{Algorithms}\label{sec:algos}
\input{algos}
\subsection[]{Polynomial-time cases: excluding \pthree, \chord, or \oneedgek}\label{sec:poly}
\input{poly}

\subsection[]{Excluding \aabb}\label{sec:aabb}
\input{aabb}

\subsection[]{Excluding \aakbb}\label{sec:aacbb}
\input{aacbb}

\subsection[]{Excluding \ababk}\label{sec:ababk}
\input{abab}

\subsection[]{Excluding \abbak}\label{sec:abbak}
\input{abba}

\section{Hardness}\label{sec:hardness}
\input{hardness-reduction-3sat}
\input{hardness-2-subdivision}
\input{hardness-long-subdivision}
\input{hardness-chain-reduction}
\input{hardness-summary}

\section{Conclusion}\label{sec:outro}
\input{outro.tex}

\paragraph{Acknowledgement.} Part of this work was done during the workshop Homonolo 2025.
We are grateful to the organizers and other participants for the inspiring atmosphere and fruitful discussions.

\bibliographystyle{plain}
\bibliography{main}
\end{document}

%% file: macros.tex
\usepackage[mathscr]{euscript}

\newcommand\blfootnote[1]{%
  \begingroup
  \renewcommand\thefootnote{}\footnote{#1}%
  \addtocounter{footnote}{-1}%
  \endgroup
}

\renewcommand{\phi}{\varphi}

\renewcommand{\epsilon}{\varepsilon}

\newcommand{\Oh}{\mathcal{O}}

\newcommand{\cC}{\mathcal{C}}

\newcommand{\cI}{\mathcal{I}}
\newcommand{\cJ}{\mathcal{J}}

\newcommand{\cP}{\mathcal{P}}
\newcommand{\cQ}{\mathcal{Q}}

\newcommand{\cS}{\mathcal{S}}

\newcommand{\cX}{\mathcal{X}}
\newcommand{\cY}{\mathcal{Y}}
\newcommand{\cZ}{\mathcal{Z}}

\newcommand{\vphi}{\varphi}

\newcommand{\wei}{\mathfrak{w}}

\newcommand{\MIS}{\textsc{MIS}\xspace}
\newcommand{\MWIS}{\textsc{MWIS}\xspace}

\newcommand{\sat}[1]{{#1}-\textsc{Sat}\xspace}

\newcommand{\Left}{\mathsf{Left}}
\newcommand{\Right}{\mathsf{Right}}

\newcommand{\NP}{\textsf{NP}\xspace}

\newcommand{\tsd}{^{(2)}}

\newcommand{\true}{\mathsf{true}}
\newcommand{\false}{\mathsf{false}}

\newcommand{\mathbi}[1]{\mathcal{#1}}

%% file: icons.tex
\newcommand{\pthree}{
\begin{tikzpicture}[every node/.style={fill=black, circle, inner sep=0.5pt, outer sep = 0}, scale=0.2,baseline= -0.5mm]
    \node (a) at (0,0) {};
    \node (b) at (1,0) {};
    \node (c) at (2,0) {};    
    \path[draw] (a) to[bend left = 60] (b);
    \path[draw] (b) to[bend left = 60] (c);
\end{tikzpicture}
}

\newcommand{\extpthree}{
\begin{tikzpicture}[every node/.style={fill=black, circle, inner sep=0.5pt, outer sep = 0}, scale=0.2,baseline= -0.5mm]
    \node[fill=none] (start) at (-1,0) {\footnotesize{$_{(k)}$}};   
    \node[fill=none] (end) at (3,0) {\footnotesize{$_{(k)}$}};   
    \node (a) at (0,0) {};
    \node (b) at (1,0) {};
    \node (c) at (2,0) {};        
    \path[draw] (a) to[bend left = 60] (b);
    \path[draw] (b) to[bend left = 60] (c);
\end{tikzpicture}
}

\newcommand{\chord}{
\begin{tikzpicture}[every node/.style={fill=black, circle, inner sep=0.5pt, outer sep = 0}, scale=0.2,baseline= -0.5mm]
    \node (a) at (0,0) {};
    \node (b) at (1,0) {};
    \node (c) at (2,0) {};    
    \path[draw] (a) to[bend left = 60] (b);
    \path[draw] (a) to[bend left = 60] (c);
\end{tikzpicture}
}

\newcommand{\chordrev}{
\begin{tikzpicture}[every node/.style={fill=black, circle, inner sep=0.5pt, outer sep = 0}, scale=0.2,baseline= -0.5mm]
    \node (c) at (0,0) {};
    \node (b) at (1,0) {};
    \node (a) at (2,0) {};    
    \path[draw] (a) to[bend left = -60] (b);
    \path[draw] (a) to[bend left = -60] (c);
\end{tikzpicture}
}

\newcommand{\extchord}{
\begin{tikzpicture}[every node/.style={fill=black, circle, inner sep=0.5pt, outer sep = 0}, scale=0.2,baseline= -0.5mm]
    \node[fill=none] (start) at (-1,0) {\footnotesize{$_{(k)}$}};   
    \node[fill=none] (end) at (3,0) {\footnotesize{$_{(k)}$}};   
    \node (a) at (0,0) {};
    \node (b) at (1,0) {};
    \node (c) at (2,0) {};        
    \path[draw] (a) to[bend left = 60] (b);
    \path[draw] (a) to[bend left = 60] (c);
\end{tikzpicture}
}

\newcommand{\extchordrev}{
\begin{tikzpicture}[every node/.style={fill=black, circle, inner sep=0.5pt, outer sep = 0}, scale=0.2,baseline= -0.5mm]
    \node[fill=none] (start) at (-1,0) {\footnotesize{$_{(k)}$}};   
    \node[fill=none] (end) at (3,0) {\footnotesize{$_{(k)}$}};   
    \node (c) at (0,0) {};
    \node (b) at (1,0) {};
    \node (a) at (2,0) {};        
    \path[draw] (a) to[bend left = -60] (b);
    \path[draw] (a) to[bend left = -60] (c);
\end{tikzpicture}
}

\newcommand{\oneedgek}{
\begin{tikzpicture}[every node/.style={fill=black, circle, inner sep=0.5pt, outer sep = 0}, scale=0.2,baseline= -0.5mm]
    \node (a1) at (0,0) {};
    \node[fill=none] (k1) at (1.5,0) {\footnotesize{$_{(k)}$}};   
    \node (a2) at (3,0) {};   
    \path[draw] (a1) to[bend left = 70] (a2);
\end{tikzpicture}
}

\newcommand{\extoneedgek}{
\begin{tikzpicture}[every node/.style={fill=black, circle, inner sep=0.5pt, outer sep = 0}, scale=0.2,baseline= -0.5mm]
    \node[fill=none] (start) at (-1,0) {\footnotesize{$_{(k)}$}};   
    \node[fill=none] (end) at (4,0) {\footnotesize{$_{(k)}$}};   
    \node (a1) at (0,0) {};
    \node[fill=none] (k1) at (1.5,0) {\footnotesize{$_{(k)}$}};   
    \node (a2) at (3,0) {};       
    \path[draw] (a1) to[bend left = 70] (a2);
\end{tikzpicture}
}

\newcommand{\abab}{
\begin{tikzpicture}[every node/.style={fill=black, circle, inner sep=0.5pt, outer sep = 0}, scale=0.2,baseline= -0.5mm]
    \node (a1) at (0,0) {};
    \node (b1) at (1,0) {};
    \node (a2) at (2,0) {};    
    \node (b2) at (3,0) {};    
    \path[draw] (a1) to[bend left = 60] (a2);
    \path[draw] (b1) to[bend left = 60] (b2);
\end{tikzpicture}
}

\newcommand{\ababk}{
\begin{tikzpicture}[every node/.style={fill=black, circle, inner sep=0.5pt, outer sep = 0}, scale=0.2,baseline= -0.5mm]
    \node (a1) at (0,0) {};
    \node[fill=none] (k1) at (1,0) {\footnotesize{$_{(k)}$}};   
    \node (b1) at (2,0) {};
    \node[fill=none] (k2) at (3,0) {\footnotesize{$_{(k)}$}};   
    \node (a2) at (4,0) {};    
    \node[fill=none] (k3) at (5,0) {\footnotesize{$_{(k)}$}};   
    \node (b2) at (6,0) {};    
    \path[draw] (a1) to[bend left = 70] (a2);
    \path[draw] (b1) to[bend left = 70] (b2);
\end{tikzpicture}
}

\newcommand{\abkab}{
\begin{tikzpicture}[every node/.style={fill=black, circle, inner sep=0.5pt, outer sep = 0}, scale=0.2,baseline= -0.5mm]
    \node (a1) at (1,0) {};
    \node (b1) at (2,0) {};
    \node[fill=none] (k2) at (3,0) {\footnotesize{$_{(k)}$}};   
    \node (a2) at (4,0) {};    
    \node (b2) at (5,0) {};    
    \path[draw] (a1) to[bend left = 70] (a2);
    \path[draw] (b1) to[bend left = 70] (b2);
\end{tikzpicture}
}

\newcommand{\akbab}{
\begin{tikzpicture}[every node/.style={fill=black, circle, inner sep=0.5pt, outer sep = 0}, scale=0.2,baseline= -0.5mm]
    \node (a1) at (0,0) {};
    \node[fill=none] (k1) at (1,0) {\footnotesize{$_{(k)}$}};   
    \node (b1) at (2,0) {};
    \node (a2) at (3,0) {};      
    \node (b2) at (5,0) {};    
    \path[draw] (a1) to[bend left = 70] (a2);
    \path[draw] (b1) to[bend left = 70] (b2);
\end{tikzpicture}
}

\newcommand{\extababk}{
\begin{tikzpicture}[every node/.style={fill=black, circle, inner sep=0.5pt, outer sep = 0}, scale=0.2,baseline= -0.5mm]
    \node[fill=none] (start) at (-1,0) {\footnotesize{$_{(k)}$}};   
    \node[fill=none] (end) at (7,0) {\footnotesize{$_{(k)}$}};   
    \node (a1) at (0,0) {};
    \node[fill=none] (k1) at (1,0) {\footnotesize{$_{(k)}$}};   
    \node (b1) at (2,0) {};
    \node[fill=none] (k2) at (3,0) {\footnotesize{$_{(k)}$}};   
    \node (a2) at (4,0) {};    
    \node[fill=none] (k3) at (5,0) {\footnotesize{$_{(k)}$}};   
    \node (b2) at (6,0) {};    
    \path[draw] (a1) to[bend left = 70] (a2);
    \path[draw] (b1) to[bend left = 70] (b2);
\end{tikzpicture}
}

\newcommand{\aabb}{
\begin{tikzpicture}[every node/.style={fill=black, circle, inner sep=0.5pt, outer sep = 0}, scale=0.2,baseline= -0.5mm]
    \node (a1) at (0,0) {};
    \node (a2) at (1,0) {};
    \node (b1) at (2,0) {};    
    \node (b2) at (3,0) {};    
    \path[draw] (a1) to[bend left = 60] (a2);
    \path[draw] (b1) to[bend left = 60] (b2);
\end{tikzpicture}
}

\newcommand{\aakbb}{
\begin{tikzpicture}[every node/.style={fill=black, circle, inner sep=0.5pt, outer sep = 0}, scale=0.2,baseline= -0.5mm]
    \node (a1) at (0,0) {};
    \node (a2) at (1,0) {};
    \node[fill=none] (k) at (2,0) {\footnotesize{$_{(k)}$}};    
    \node (b1) at (3,0) {};    
    \node (b2) at (4,0) {};    
    \path[draw] (a1) to[bend left = 60] (a2);
    \path[draw] (b1) to[bend left = 60] (b2);
\end{tikzpicture}
}

\newcommand{\extaakbb}{
\begin{tikzpicture}[every node/.style={fill=black, circle, inner sep=0.5pt, outer sep = 0}, scale=0.2,baseline= -0.5mm]
    \node[fill=none] (start) at (-1,0) {\footnotesize{$_{(k)}$}};   
    \node[fill=none] (end) at (5,0) {\footnotesize{$_{(k)}$}};   
    \node (a1) at (0,0) {};
    \node (a2) at (1,0) {};
    \node[fill=none] (k) at (2,0) {\footnotesize{$_{(k)}$}};    
    \node (b1) at (3,0) {};    
    \node (b2) at (4,0) {};    
    \path[draw] (a1) to[bend left = 60] (a2);
    \path[draw] (b1) to[bend left = 60] (b2);
\end{tikzpicture}
}

\newcommand{\abba}{
\begin{tikzpicture}[every node/.style={fill=black, circle, inner sep=0.5pt, outer sep = 0}, scale=0.2,baseline= -0.5mm]
    \node (a1) at (0,0) {};
    \node (b1) at (1,0) {};
    \node (b2) at (2,0) {};    
    \node (a2) at (3,0) {};    
    \path[draw] (a1) to[bend left = 70] (a2);
    \path[draw] (b1) to[bend left = 45] (b2);
\end{tikzpicture}
}

\newcommand{\abbak}{
\begin{tikzpicture}[every node/.style={fill=black, circle, inner sep=0.5pt, outer sep = 0}, scale=0.2,baseline= -0.5mm]
    \node (a1) at (0,0) {};
    \node[fill=none] (k1) at (1,0) {\footnotesize{$_{(k)}$}};  
    \node (b1) at (2,0) {};
    \node (b2) at (3,0) {};    
    \node[fill=none] (k2) at (4,0) {\footnotesize{$_{(k)}$}};  
    \node (a2) at (5,0) {};    
    \path[draw] (a1) to[bend left = 70] (a2);
    \path[draw] (b1) to[bend left = 45] (b2);
\end{tikzpicture}
}

\newcommand{\extabbak}{
\begin{tikzpicture}[every node/.style={fill=black, circle, inner sep=0.5pt, outer sep = 0}, scale=0.2,baseline= -0.5mm]
    \node[fill=none] (start) at (-1,0) {\footnotesize{$_{(k)}$}};   
    \node[fill=none] (end) at (6,0) {\footnotesize{$_{(k)}$}};  
    \node (a1) at (0,0) {};
    \node[fill=none] (k1) at (1,0) {\footnotesize{$_{(k)}$}};  
    \node (b1) at (2,0) {};
    \node (b2) at (3,0) {};    
    \node[fill=none] (k2) at (4,0) {\footnotesize{$_{(k)}$}};  
    \node (a2) at (5,0) {};    
    \path[draw] (a1) to[bend left = 70] (a2);
    \path[draw] (b1) to[bend left = 45] (b2);
\end{tikzpicture}
}

\newcommand{\bad}{
\begin{tikzpicture}[every node/.style={fill=black, circle, inner sep=0.5pt, outer sep = 0}, scale=0.2]
    \node (s) at (0,0) {};
    \node (p) at (1,0) {};
    \node (q) at (2,0) {};
    \node (r) at (3,0) {};
    \path[draw] (s) to[bend left = 60] (p);
    \path[draw] (s) to[bend left = 60] (r);
\end{tikzpicture}
}

\newcommand{\cad}{
	\begin{tikzpicture}[every node/.style={fill=black, circle, inner sep=0.5pt, outer sep = 0}, scale=0.2]
		\node (s) at (0,0) {};
		\node (p) at (1,0) {};
		\node (q) at (2,0) {};
		\node (r) at (3,0) {};
		\path[draw] (s) to[bend left = 60] (q);
		\path[draw] (s) to[bend left = 60] (r);
	\end{tikzpicture}
}

\newcommand{\adb}{
\begin{tikzpicture}[every node/.style={fill=black, circle, inner sep=0.5pt, outer sep = 0}, scale=0.2]
    \node (s) at (0,0) {};
    \node (p) at (1,0) {};
    \node (q) at (2,0) {};
    \node (r) at (3,0) {};
    \path[draw] (s) to[bend left = 60] (r);
    \path[draw] (p) to[bend left = 60] (r);
\end{tikzpicture}
}

\newcommand{\adbp}{
\begin{tikzpicture}[every node/.style={fill=black, circle, inner sep=0.5pt, outer sep = 0}, scale=0.2]
    \node (s) at (0,0) {};
    \node (p) at (1,0) {};
    \node (q) at (2,0) {};
    \node (r) at (3,0) {};
    \path[draw] (s) to[bend left = 60] (r);
    \path[draw] (p) to[bend left = 60] (r);
    \path[draw] (q) to[bend left = 60] (r);
\end{tikzpicture}
}

\newcommand{\abd}{
\begin{tikzpicture}[every node/.style={fill=black, circle, inner sep=0.5pt, outer sep = 0}, scale=0.2]
    \node (s) at (0,0) {};
    \node (p) at (1,0) {};
    \node (q) at (2,0) {};
    \node (r) at (3,0) {};
    \path[draw] (s) to[bend left = 60] (p);
    \path[draw] (p) to[bend left = 60] (r);
\end{tikzpicture}
}

\newcommand{\abdp}{
\begin{tikzpicture}[every node/.style={fill=black, circle, inner sep=0.5pt, outer sep = 0}, scale=0.2]
    \node (s) at (0,0) {};
    \node (p) at (1,0) {};
    \node (q) at (2,0) {};
    \node (r) at (3,0) {};
    \path[draw] (s) to[bend left = 60] (p);
    \path[draw] (p) to[bend left = 60] (r);
    \path[draw] (p) to[bend left = 60] (q);
\end{tikzpicture}
}

\newcommand{\abdc}{
\begin{tikzpicture}[every node/.style={fill=black, circle, inner sep=0.5pt, outer sep = 0}, scale=0.2]
    \node (s) at (0,0) {};
    \node (p) at (1,0) {};
    \node (q) at (2,0) {};
    \node (r) at (3,0) {};
    \path[draw] (s) to[bend left = 60] (p);
    \path[draw] (p) to[bend left = 60] (r);
    \path[draw] (q) to[bend left = 60] (r);
\end{tikzpicture}
}

\newcommand{\abcd}{
\begin{tikzpicture}[every node/.style={fill=black, circle, inner sep=0.5pt, outer sep = 0}, scale=0.2]
    \node (s) at (0,0) {};
    \node (p) at (1,0) {};
    \node (q) at (2,0) {};
    \node (r) at (3,0) {};
    \path[draw] (s) to[bend left = 60] (p);
    \path[draw] (p) to[bend left = 60] (q);
    \path[draw] (q) to[bend left = 60] (r);
\end{tikzpicture}
}

\newcommand{\abc}{
\begin{tikzpicture}[every node/.style={fill=black, circle, inner sep=0.5pt, outer sep = 0}, scale=0.2]
    \node (s) at (0,0) {};
    \node (p) at (1,0) {};
    \node (q) at (2,0) {};
    \path[draw] (s) to[bend left = 60] (p);
    \path[draw] (p) to[bend left = 60] (q);
\end{tikzpicture}
}

\newcommand{\bac}{
	\begin{tikzpicture}[every node/.style={fill=black, circle, inner sep=0.5pt, outer sep = 0}, scale=0.2]
		\node (s) at (0,0) {};
		\node (p) at (1,0) {};
		\node (q) at (2,0) {};
		\path[draw] (s) to[bend left = 60] (p);
		\path[draw] (s) to[bend left = 60] (q);
	\end{tikzpicture}
}

\newcommand{\acbd}{
\begin{tikzpicture}[every node/.style={fill=black, circle, inner sep=0.5pt, outer sep = 0}, scale=0.2]
    \node (s) at (0,0) {};
    \node (p) at (1,0) {};
    \node (q) at (2,0) {};
    \node (r) at (3,0) {};
    \path[draw] (s) to[bend left = 60] (q);
    \path[draw] (p) to[bend left = 40] (q);
    \path[draw] (p) to[bend left = 60] (r);
\end{tikzpicture}
}

\newcommand{\acb}{
\begin{tikzpicture}[every node/.style={fill=black, circle, inner sep=0.5pt, outer sep = 0}, scale=0.2]
    \node (s) at (0,0) {};
    \node (p) at (1,0) {};
    \node (q) at (2,0) {};
    \path[draw] (s) to[bend left = 60] (q);
    \path[draw] (p) to[bend left = 40] (q);
\end{tikzpicture}
}

\newcommand{\adcb}{
\begin{tikzpicture}[every node/.style={fill=black, circle, inner sep=0.5pt, outer sep = 0}, scale=0.2]
    \node (s) at (0,0) {};
    \node (p) at (1,0) {};
    \node (q) at (2,0) {};
    \node (r) at (3,0) {};
    \path[draw] (s) to[bend left = 60] (r);
    \path[draw] (p) to[bend left = 60] (q);
    \path[draw] (q) to[bend left = 60] (r);
\end{tikzpicture}
}

\newcommand{\badc}{
\begin{tikzpicture}[every node/.style={fill=black, circle, inner sep=0.5pt, outer sep = 0}, scale=0.2]
    \node (s) at (0,0) {};
    \node (p) at (1,0) {};
    \node (q) at (2,0) {};
    \node (r) at (3,0) {};
    \path[draw] (s) to[bend left = 60] (p);
    \path[draw] (q) to[bend left = 60] (r);
    \path[draw] (s) to[bend left = 60] (r);
\end{tikzpicture}
}

\newcommand{\abnce}{
\begin{tikzpicture}[every node/.style={fill=black, circle, inner sep=0.5pt, outer sep = 0}, scale=0.2]
    \node (s) at (0,0) {};
    \node (p) at (1,0) {};
    \node (q) at (2,0) {};
    \node (r) at (3,0) {};
    \node (t) at (4,0) {};
    \path[draw] (s) to[bend left = 60] (p);
    \path[draw] (q) to[bend left = 60] (t);
\end{tikzpicture}
}

\newcommand{\abncde}{
\begin{tikzpicture}[every node/.style={fill=black, circle, inner sep=0.5pt, outer sep = 0}, scale=0.2]
    \node (s) at (0,0) {};
    \node (p) at (1,0) {};
    \node (q) at (2,0) {};
    \node (r) at (3,0) {};
    \node (t) at (4,0) {};
    \path[draw] (s) to[bend left = 60] (p);
    \path[draw] (q) to[bend left = 60] (r);
    \path[draw] (r) to[bend left = 60] (t);
\end{tikzpicture}
}

\newcommand{\aenbcd}{
\begin{tikzpicture}[every node/.style={fill=black, circle, inner sep=0.5pt, outer sep = 0}, scale=0.2]
    \node (s) at (0,0) {};
    \node (p) at (1,0) {};
    \node (q) at (2,0) {};
    \node (r) at (3,0) {};
    \node (t) at (4,0) {};
    \path[draw] (s) to[bend left = 60] (t);
    \path[draw] (p) to[bend left = 60] (q);
    \path[draw] (q) to[bend left = 60] (r);
\end{tikzpicture}
}

\newcommand{\aenbdc}{
\begin{tikzpicture}[every node/.style={fill=black, circle, inner sep=0.5pt, outer sep = 0}, scale=0.2]
    \node (s) at (0,0) {};
    \node (p) at (1,0) {};
    \node (q) at (2,0) {};
    \node (r) at (3,0) {};
    \node (t) at (4,0) {};
    \path[draw] (s) to[bend left = 60] (t);
    \path[draw] (p) to[bend left = 60] (r);
    \path[draw] (q) to[bend left = 60] (r);
\end{tikzpicture}
}

\newcommand{\acbnde}{
\begin{tikzpicture}[every node/.style={fill=black, circle, inner sep=0.5pt, outer sep = 0}, scale=0.2]
    \node (s) at (0,0) {};
    \node (p) at (1,0) {};
    \node (q) at (2,0) {};
    \node (r) at (3,0) {};
    \node (t) at (4,0) {};
    \path[draw] (s) to[bend left = 60] (q);
    \path[draw] (p) to[bend left = 60] (q);
    \path[draw] (r) to[bend left = 60] (t);
\end{tikzpicture}
}

\newcommand{\aednbc}{
\begin{tikzpicture}[every node/.style={fill=black, circle, inner sep=0.5pt, outer sep = 0}, scale=0.2]
    \node (s) at (0,0) {};
    \node (p) at (1,0) {};
    \node (q) at (2,0) {};
    \node (r) at (3,0) {};
    \node (t) at (4,0) {};
    \path[draw] (s) to[bend left = 60] (t);
    \path[draw] (r) to[bend left = 60] (t);
    \path[draw] (p) to[bend left = 60] (q);
\end{tikzpicture}
}

\newcommand{\abxba}{
\begin{tikzpicture}[every node/.style={fill=black, circle, inner sep=0.5pt, outer sep = 0}, scale=0.2]
    \node (s) at (0,0) {};
    \node (p) at (1,0) {};
    \node (q) at (2,0) {};
    \node (r) at (3,0) {};
    \node (t) at (4,0) {};
    \path[draw] (s) to[bend left = 60] (t);
    \path[draw] (p) to[bend left = 60] (r);
\end{tikzpicture}
}

\newcommand{\abbcca}{
\begin{tikzpicture}[every node/.style={fill=black, circle, inner sep=0.5pt, outer sep = 0}, scale=0.2]
    \node (s) at (0,0) {};
    \node (p) at (1,0) {};
    \node (q) at (2,0) {};
    \node (r) at (3,0) {};
    \node (t) at (4,0) {};
    \node (u) at (5,0) {};
    \path[draw] (s) to[bend left = 60] (u);
    \path[draw] (p) to[bend left = 60] (q);
    \path[draw] (r) to[bend left = 60] (t);
\end{tikzpicture}
}

\newcommand{\abccba}{
\begin{tikzpicture}[every node/.style={fill=black, circle, inner sep=0.5pt, outer sep = 0}, scale=0.2]
    \node (s) at (0,0) {};
    \node (p) at (1,0) {};
    \node (q) at (2,0) {};
    \node (r) at (3,0) {};
    \node (t) at (4,0) {};
    \node (u) at (5,0) {};
    \path[draw] (s) to[bend left = 60] (u);
    \path[draw] (p) to[bend left = 60] (t);
    \path[draw] (q) to[bend left = 60] (r);
\end{tikzpicture}
}

\newcommand{\abxxba}{
\begin{tikzpicture}[every node/.style={fill=black, circle, inner sep=0.5pt, outer sep = 0}, scale=0.2]
    \node (s) at (0,0) {};
    \node (p) at (1,0) {};
    \node (q) at (2,0) {};
    \node (r) at (3,0) {};
    \node (t) at (4,0) {};
    \node (u) at (5,0) {};
    \path[draw] (s) to[bend left = 60] (u);
    \path[draw] (p) to[bend left = 60] (t);
\end{tikzpicture}
}

\newcommand{\abcbca}{
\begin{tikzpicture}[every node/.style={fill=black, circle, inner sep=0.5pt, outer sep = 0}, scale=0.2]
    \node (s) at (0,0) {};
    \node (p) at (1,0) {};
    \node (q) at (2,0) {};
    \node (r) at (3,0) {};
    \node (t) at (4,0) {};
    \node (u) at (5,0) {};
    \path[draw] (s) to[bend left = 60] (u);
    \path[draw] (p) to[bend left = 60] (r);
    \path[draw] (q) to[bend left = 60] (t);
\end{tikzpicture}
}

\newcommand{\abccab}{
\begin{tikzpicture}[every node/.style={fill=black, circle, inner sep=0.5pt, outer sep = 0}, scale=0.2]
    \node (s) at (0,0) {};
    \node (p) at (1,0) {};
    \node (q) at (2,0) {};
    \node (r) at (3,0) {};
    \node (t) at (4,0) {};
    \node (u) at (5,0) {};
    \path[draw] (s) to[bend left = 60] (t);
    \path[draw] (p) to[bend left = 60] (u);
    \path[draw] (q) to[bend left = 60] (r);
\end{tikzpicture}
}

\newcommand{\abcabc}{
\begin{tikzpicture}[every node/.style={fill=black, circle, inner sep=0.5pt, outer sep = 0}, scale=0.2]
    \node (s) at (0,0) {};
    \node (p) at (1,0) {};
    \node (q) at (2,0) {};
    \node (r) at (3,0) {};
    \node (t) at (4,0) {};
    \node (u) at (5,0) {};
    \path[draw] (s) to[bend left = 60] (r);
    \path[draw] (p) to[bend left = 60] (t);
    \path[draw] (q) to[bend left = 60] (u);
\end{tikzpicture}
}

\newcommand{\aabbcc}{
\begin{tikzpicture}[every node/.style={fill=black, circle, inner sep=0.5pt, outer sep = 0}, scale=0.2]
    \node (s) at (0,0) {};
    \node (p) at (1,0) {};
    \node (q) at (2,0) {};
    \node (r) at (3,0) {};
    \node (t) at (4,0) {};
    \node (u) at (5,0) {};
    \path[draw] (s) to[bend left = 60] (p);
    \path[draw] (q) to[bend left = 60] (r);
    \path[draw] (t) to[bend left = 60] (u);
\end{tikzpicture}
}

\newcommand{\cadb}{
	\begin{tikzpicture}[every node/.style={fill=black, circle, inner sep=0.5pt, outer sep = 0}, scale=0.2]
		\node (s) at (0,0) {};
		\node (p) at (1,0) {};
		\node (q) at (2,0) {};
		\node (r) at (3,0) {};
		\path[draw] (s) to[bend left = 60] (q);
		\path[draw] (s) to[bend left = 60] (r);
		\path[draw] (p) to[bend left = 60] (r);
	\end{tikzpicture}
}

\newcommand{\adbc}{
	\begin{tikzpicture}[every node/.style={fill=black, circle, inner sep=0.5pt, outer sep = 0}, scale=0.2]
		\node (s) at (0,0) {};
		\node (p) at (1,0) {};
		\node (q) at (2,0) {};
		\node (r) at (3,0) {};
		\path[draw] (s) to[bend left = 60] (r);
		\path[draw] (p) to[bend left = 60] (r);
		\path[draw] (p) to[bend left = 60] (q);
	\end{tikzpicture}
}

\newcommand{\cabd}{
	\begin{tikzpicture}[every node/.style={fill=black, circle, inner sep=0.5pt, outer sep = 0}, scale=0.2]
		\node (s) at (0,0) {};
		\node (p) at (1,0) {};
		\node (q) at (2,0) {};
		\node (r) at (3,0) {};
		\path[draw] (s) to[bend left = 60] (q);
		\path[draw] (s) to[bend left = 60] (p);
		\path[draw] (p) to[bend left = 60] (r);
	\end{tikzpicture}
}

\newcommand{\abnced}{
	\begin{tikzpicture}[every node/.style={fill=black, circle, inner sep=0.5pt, outer sep = 0}, scale=0.2]
		\node (s) at (0,0) {};
		\node (p) at (1,0) {};
		\node (q) at (2,0) {};
		\node (r) at (3,0) {};
		\node (t) at (4,0) {};
		\path[draw] (s) to[bend left = 60] (p);
		\path[draw] (q) to[bend left = 60] (t);
		\path[draw] (r) to[bend left = 60] (t);
	\end{tikzpicture}
}

\newcommand{\adc}{
	\begin{tikzpicture}[every node/.style={fill=black, circle, inner sep=0.5pt, outer sep = 0}, scale=0.2]
		\node (a) at (0,0) {};
		\node (b) at (1,0) {};
		\node (c) at (2,0) {};
		\node (d) at (3,0) {};
		\path[draw] (a) to[bend left = 60] (d);
		\path[draw] (c) to[bend left = 60] (d);
	\end{tikzpicture}
}

\newcommand{\abcbac}{
	\begin{tikzpicture}[every node/.style={fill=black, circle, inner sep=0.5pt, outer sep = 0}, scale=0.2]
		\node (s) at (0,0) {};
		\node (p) at (1,0) {};
		\node (q) at (2,0) {};
		\node (r) at (3,0) {};
		\node (t) at (4,0) {};
		\node (u) at (5,0) {};
		\path[draw] (s) to[bend left = 60] (t);
		\path[draw] (p) to[bend left = 60] (r);
		\path[draw] (q) to[bend left = 60] (u);
	\end{tikzpicture}
}

\newcommand{\abbcac}{
	\begin{tikzpicture}[every node/.style={fill=black, circle, inner sep=0.5pt, outer sep = 0}, scale=0.2]
		\node (s) at (0,0) {};
		\node (p) at (1,0) {};
		\node (q) at (2,0) {};
		\node (r) at (3,0) {};
		\node (t) at (4,0) {};
		\node (u) at (5,0) {};
		\path[draw] (s) to[bend left = 60] (t);
		\path[draw] (p) to[bend left = 60] (q);
		\path[draw] (r) to[bend left = 60] (u);
	\end{tikzpicture}
}

\newcommand{\abacbc}{
	\begin{tikzpicture}[every node/.style={fill=black, circle, inner sep=0.5pt, outer sep = 0}, scale=0.2]
		\node (s) at (0,0) {};
		\node (p) at (1,0) {};
		\node (q) at (2,0) {};
		\node (r) at (3,0) {};
		\node (t) at (4,0) {};
		\node (u) at (5,0) {};
		\path[draw] (s) to[bend left = 60] (q);
		\path[draw] (p) to[bend left = 60] (t);
		\path[draw] (r) to[bend left = 60] (u);
	\end{tikzpicture}
}

\newcommand{\aabccb}{
	\begin{tikzpicture}[every node/.style={fill=black, circle, inner sep=0.5pt, outer sep = 0}, scale=0.2]
		\node (s) at (0,0) {};
		\node (p) at (1,0) {};
		\node (q) at (2,0) {};
		\node (r) at (3,0) {};
		\node (t) at (4,0) {};
		\node (u) at (5,0) {};
		\path[draw] (s) to[bend left = 60] (p);
		\path[draw] (q) to[bend left = 60] (u);
		\path[draw] (r) to[bend left = 60] (t);
	\end{tikzpicture}
}

\newcommand{\aabcbc}{
	\begin{tikzpicture}[every node/.style={fill=black, circle, inner sep=0.5pt, outer sep = 0}, scale=0.2]
		\node (s) at (0,0) {};
		\node (p) at (1,0) {};
		\node (q) at (2,0) {};
		\node (r) at (3,0) {};
		\node (t) at (4,0) {};
		\node (u) at (5,0) {};
		\path[draw] (s) to[bend left = 60] (p);
		\path[draw] (q) to[bend left = 60] (t);
		\path[draw] (r) to[bend left = 60] (u);
	\end{tikzpicture}
}

\newcommand{\ac}{
	\begin{tikzpicture}[every node/.style={fill=black, circle, inner sep=0.5pt, outer sep = 0}, scale=0.2]
		\node (a) at (0,0) {};
		\node (b) at (1,0) {};
		\node (c) at (2,0) {};
		\path[draw] (a) to[bend left = 60] (c);
	\end{tikzpicture}
}


%% file: intro.tex
In the \textsc{Max Independent Set} (\MIS) problem, we are given a graph $G$ and we ask for a largest \emph{independent set} in $G$, i.e., a largest set of pairwise nonadjacent vertices. The \textsc{Max Weight Independent Set} (\MWIS) problem is a generalization of \MIS, where we are given a graph $G$ together with a weight function $\wei: V(G)\to\mathbb{Q}_{>0}$, and we ask for an independent set of maximum total weight. Both problems are among the most fundamental and well-studied problems in algorithmic graph theory.
They are also known to be notoriously hard: \MIS is among Karp's 21 \NP-complete problems~\cite{DBLP:conf/coco/Karp72}, cannot be solved in subexponential time under the Exponential-Time Hypothesis (ETH)~\cite{DBLP:journals/jcss/ImpagliazzoPZ01}, is a canonical \textsf{W[1]}-hard problem~\cite{platypus}, and is hard to approximate~\cite{Hastad96cliqueis}, even in parameterized setting~\cite{DBLP:journals/siamcomp/ChalermsookCKLM20,DBLP:conf/focs/LinRSW23}.

A typical approach when dealing with such a hard problem is to consider restricted instances in order to understand the boundary of tractability.
Usually, one considers instances coming from certain \emph{hereditary} graph classes, i.e., classes of graphs closed under vertex deletion.
Every such class can be characterized by a family of forbidden induced subgraphs.
Thus, as a starting point, it is natural to ask for which graphs $H$ the \MWIS problem can be solved efficiently on \emph{$H$-free graphs}, i.e., graphs that do not contain $H$ as an induced subgraph.

The complexity study of \MWIS in $H$-free graphs is one of the central topics in algorithmic graph theory.
It turns out that the crucial role is played by the family of subcubic forests whose every component is a path or a \emph{subdivided claw}, i.e., a tree with at most one vertex of degree 3.
Let us denote by $\cS$ the family of such forests.

It was already observed by Alekseev~\cite{alekseev1982effect} that if $H \notin \cS$, then \MIS (and thus, \MWIS) is \NP-hard in $H$-free graphs.
Indeed, this observation follows from combining two known results. First, \MIS is \NP-hard for subcubic graphs~\cite{GAREY1976237}, and they in particular exclude any graph $H$ with a vertex of degree at least 4.
The second ingredient is the so-called `Poljak's subdivision trick'~\cite{Po74}:
if we subdivide one edge of a graph $G$ twice, i.e., we replace it by a four-vertex path, then the size of a largest independent set increases by exactly one.
Thus, if $H$ has a cycle or two vertices of degree 3 in a single component, then subdividing each edge of any instance of \MIS sufficiently many times yields an equivalent instance that is $H$-free.

It turns out that polynomial-time algorithms for \MWIS are only known for small graphs $H \in \cS$,
most notably, if $H$ is a path on at most 6 vertices~\cite{DBLP:journals/talg/GrzesikKPP22}, a graph obtained from a three-leaf star by subdividing one edge once~\cite{DBLP:journals/jda/LozinM08}, or a graph whose every component is a three-leaf star~\cite{DBLP:journals/dam/BrandstadtM18a}.
On the other hand, we know no \NP-hardness results for \MWIS in $H$-free graphs with $H \in \cS$ and the general belief is that the problem is polynomial-time solvable in all such cases.
This belief is supported by the fact that for every $H \in \cS$, the \MWIS problem can be solved in quasipolynomial time when restricted to $H$-free graphs~\cite{DBLP:conf/focs/GartlandL20,DBLP:conf/sosa/PilipczukPR21,DBLP:conf/stoc/GartlandLMPPR24}.
Note that this is a strong indication that the problem is not \NP-hard in any of these cases,
as otherwise every problem in \NP{} would be solvable in quasipolynomial time, which is considered unlikely.

In this paper we follow the line of research described above, but we focus on hereditary classes of \emph{ordered} graphs, i.e., graphs with a fixed linear order of their vertex set.
We aim to understand for which ordered graphs $H$ the \MWIS problem can be solved efficiently in $H$-free ordered graphs (where an ordered induced subgraph is an induced subgraph that preserves the relative order of its vertices).
We emphasize that the ordered setting can be seen as a refinement of the unordered one.
Indeed, forbidding a single \emph{unordered} graph is equivalent to forbidding \emph{all} of its possible orderings as ordered induced subgraphs.
Instead, we can forbid only some of the orderings, which allows us to trace the boundary of tractability more precisely.

Let us remark that the idea of looking at classical graph-theoretic problems in an ordered setting is not new.
For example, there is a rich (and still growing) line of research concerning Tur\'an-type~\cite{DBLP:conf/bcc/Tardos19,DBLP:journals/jct/KorandiTTW19,DBLP:journals/siamdm/BoskovicK23,Pach2006} and Ramsey-type~\cite{DBLP:journals/combinatorics/BalkoC0K20,DBLP:journals/jct/ConlonFLS17,DBLP:journals/combinatorics/DudekGR24} problems in ordered graphs. There are also some results on the chromatic number~\cite{DBLP:journals/combinatorica/AxenovichRU18,Walczakordered} or structural properties~\cite{DBLP:journals/jctb/PachT21,DBLP:journals/siamdm/ScottSS22} of certain classes of ordered graphs.
However, the algorithmic results seem to be much scarcer: we are only aware of recent work on the complexity of (\textsc{List}) $k$-\textsc{Coloring} in classes defined by excluding one induced ordered subgraph~\cite{DBLP:journals/siamdm/HajebiLS24,DBLP:conf/stacs/PiecykR26}.
In particular, to the best of our knowledge, the complexity of \MWIS in this setting was not studied before.
In this paper, we fill this gap by proving  the following theorem.

\begin{theorem}\label{thm:main}
    Let $H$ be an ordered graph with at least two vertices. 
    The \MWIS problem in ordered $H$-free graphs can be solved in:
    \begin{enumerate}[(1)]
        \item polynomial time, if $H$ is an induced subgraph of a graph in $\bigcup_{k \geq 0} \{\extpthree,\extchord,\extchordrev,\extoneedgek \}$, 
        \item quasipolynomial time, if $H$ is an induced subgraph of a graph in  $\bigcup_{k \geq 0} \{\extaakbb, \extababk \}$,
        \item subexponential time, if $H$ is an induced subgraph of a graph in  $\bigcup_{k \geq 0} \{\extabbak \}$.        
    \end{enumerate}
    In all other cases, even the \MIS problem is \NP-hard.
\end{theorem}

The notation used in \cref{thm:main} (and throughout the paper) should be self-explanatory. For example, \pthree is the ordered graph on vertices $1,2,3$ in natural order and edges $12$ and $23$.
The only non-obvious feature is that by $(k)$ we denote $k$ consecutive isolated vertices, so, e.g., 
\extabbak is the graph with vertices $1,\ldots,4k+4$ in natural order, and edges $(k+1)(3k+4)$ and $(2k+1)(2k+2)$.

Let us make a few remarks concerning the statement of \cref{thm:main}. First, notice that the case that $H= \extabbak$ for some $k$ is the only one that we cannot classify either as quasipolynomial-time solvable or \NP-hard.
Thus, our \cref{thm:main} gives an almost complete complexity dichotomy into cases solvable in quasipolynomial time and those that are \NP-hard.
Still, \extabbak-free and, in particular, \abba-free graphs admit some strong structural properties that allow us to solve \MWIS in subexponential time.

Interestingly, \abba-free graphs remain one of few open cases concerning the complexity of \textsc{List} 3-\textsc{Coloring}~\cite{DBLP:journals/siamdm/HajebiLS24,DBLP:conf/stacs/PiecykR26}, which is another indication that graphs from this class might be difficult to understand. Their structure does not seem restricted enough to allow for (simple) efficient algorithms, nor rich enough to allow for (simple) hardness results.

Another indication that \abba-free graphs have complicated structure is the fact that unordered graphs that admit an ordering that avoids \abba as a subgraph are precisely graphs of \emph{queue number} 1~\cite{HeathLeightonRosenberg1992}. Such graphs are \NP-hard to recognize~\cite{DBLP:journals/siamcomp/HeathR92} and thus, there is little hope for their simple structural characterization.
On the other hand, unordered graphs that admit an ordering that avoids \abab as a subgraph are known as graphs of \emph{stack number} 1~\cite{BernhartKainen1979} and they are precisely outerplanar graphs -- a rather simple class~\cite{ChartrandHarary1967}.
These two examples already illustrate that graphs defined by forbidding even a very simple ordered substructure might be quite complicated. Furthermore, the ordering of the vertices in a forbidden subgraph can have a huge impact on the structure of the graphs in the corresponding hereditary class.

Now, let us proceed to a high-level overview of the proof of \cref{thm:main}.

\paragraph{Technical overview.}
It is not hard to observe (see \cref{lem:pacman}) that if $H'$ is obtained from $H$ by adding some isolated vertices at the beginning or at the end of the ordering, then the \MWIS problem in $H'$-free graphs can be reduced to solving the same problem in the more restrictive class of $H$-free graphs.
We note that the above observation does not apply to isolated vertices added in the middle of the ordering, as illustrated by, e.g., cases of $H=\pthree$ and $H=\abd$. Indeed, \MWIS is polynomial-time solvable in \pthree-free graphs, but \NP-hard in \abd-free graphs.

Thus, from now on let us focus on the case that $H$ has no isolated vertices at the beginning or at the end of the ordering. In what follows, the input graph is always denoted by $G$ and its number of vertices is denoted by $n$.

\medskip
The cases listed in \cref{thm:main} (1) are actually quite simple.
One can readily verify that ordered \pthree-free graphs are co-comparability graphs and ordered \chord-free and ordered \chordrev-free graphs are chordal -- in both cases \MWIS can be solved in polynomial time by classical algorithms~\cite{DBLP:journals/siamcomp/Gavril72,Golumbic2004}. 
Interestingly, the above inclusions are in fact equivalences in the sense that every co-comparability graph admits an ordering that avoids \pthree and every chordal graph admits an ordering that avoids \chord (and the reverse ordering avoids \chordrev).
The last case from \cref{thm:main} (1), i.e., $H = \oneedgek$ for some $k$, can be handled by a simple dynamic programming algorithm.

\medskip
So, let us proceed to cases listed in \cref{thm:main} (2).
Let us start with $H = \aabb$.
We remark that later we will show an algorithm for the more general case $H = \aakbb$ for any fixed $k$,
but the algorithm for $H=\aabb$ is much simpler and has better running time.

Note that \aabb-free graphs contain all bipartite graphs: indeed, if we take any bipartite graph and order its vertices so that all vertices from one part come before all vertices from the other part, we get an \aabb-free ordered graph.
\MWIS in bipartite graphs can be solved in polynomial time by a reduction to the maximum matching problem~\cite{LovaszPlummer1986_MatchingTheory}. Therefore, when working with \aabb-free graphs, we should be aware that we are in particular solving the (unordered) bipartite case and thus, we should not hope for a simple dynamic programming or branching algorithm.

In fact, our approach is to start with an \aabb-free instance of \MWIS and reduce it, in quasipolynomial time, to a quasipolynomial number of bipartite instances, each of which can be solved in polynomial time. 
A \emph{seagull} is a triple of vertices $x \prec y \prec z$ such that $y$ is adjacent to both $x$ and $z$.
We observe that graphs that do not contain any seagull are bipartite.
Thus, we employ a branching scheme in order to destroy all seagulls.

First, we observe that every two seagulls in an \aabb-free graph share a vertex or are joined by an edge.
Thus, if $x,y,z$ is a seagull, then the closed neighborhood of one of $x,y,z$ intersects at least one third of all seagulls in the graph.
We branch on such a vertex $v$, i.e., we guess whether it is in the solution or not.
In the former case we remove $v$ and all its neighbors from the graph (which destroys a constant fraction of all seagulls) and in the latter one we remove just $v$ (which destroys at least one seagull -- the one formed by $x,y,z$).
As the number of all seagulls is $\Oh(n^3)$, this branching scheme produces at most $n^{\Oh(\log n)}$ instances, each of which has no seagulls and thus, can be solved in polynomial time.
Thus, the overall running time of the algorithm is $n^{\Oh(\log n)}$.

\medskip
Now let us discuss the more general case $H = \aakbb$ for any fixed $k$.
Here we also use a branching scheme, but it is more involved than in the previous case.
The algorithm maintains a more structured instance, where the vertex set is partitioned into three consecutive parts $X, Y, Z$, such that $X$ and $Z$ are independent, while $Y$ is the ``workspace'' that still requires handling. Initially, we take $X=Z=\emptyset$ and $Y = V(G)$.
If $Y=\emptyset$, then the graph is bipartite and the problem can be solved in polynomial time, so the goal is to shrink $Y$ until it becomes empty.

To this end, we split $Y$ into two halves, $Y_L$ and $Y_R$, and exhaustively guess the first $k$ vertices of the optimum solution that lie in $Y_R$.
If $X\cup Y_L$ or $Y_R \cup Z$, say the former one, is independent, then we can move $Y_L$ to $X$, thereby reducing the size of the workspace by at least half.
Thus, we may assume that both $X \cup Y_L$ and $Y_R \cup Z$ contain an edge; call such edges \emph{internal}.
Now the forbidden pattern \aakbb yields a key structural property: the closed neighborhood of every internal edge on one side of the split intersects every internal edge on the opposite side.
Consequently, if both sides still contain an internal edge, then the neighborhood of one endpoint of any edge on the sparser side (i.e., one with fewer internal edges) intersects a constant fraction of the internal edges on the denser side.

This leads to an internal branching procedure, analogous in spirit to the \aabb-case: we guess whether such a vertex belongs to the solution or not.
If we exclude it, we destroy at least one internal edge; if we include it, we delete its closed neighborhood and thus, destroy a constant fraction of all internal edges.
We continue until one side of the split becomes independent; the number of instances produced this way is $n^{\Oh(\log n)}$.
At that point, we absorb the independent side into $X$ or $Z$, thereby obtaining a new structured instance whose workspace has size at most half of the previous one.

Repeating this halving procedure recursively yields quasipolynomially many subinstances, each eventually reduced to a bipartite graph.
Thus, the recursion has two levels: the outer one halves the workspace, while the inner one destroys a constant fraction of internal edges. Altogether, this gives running time $n^{\Oh(\log^2 n)}$.

\medskip

Now, let us consider the case $H = \ababk$ for any fixed $k$.
For simplicity, let us discuss the case $k=0$, i.e., $H = \abab$ -- we just remark that the general case is slightly more complicated.
Here the general strategy is similar in spirit to the case of \aakbb, but the relevant structure of the considered instances is different (and the arguments are more involved).
We start by partitioning the vertex set into three consecutive blocks $A,B,C$ of roughly equal size and view them as a \emph{chain}, that is, a sequence of consecutive sets (called \emph{links}) in which edges may appear only inside one set, between two consecutive sets, or between the first and the last one.
The algorithm repeatedly \emph{refines} such a chain by splitting one selected link into two consecutive sublinks.

The main subroutine is a refinement lemma:
for any chosen link, in quasipolynomial time we can branch into quasipolynomially many instances in such a way that this link is replaced by two smaller consecutive links.
Again the crucial structural observation is that, since \abab is forbidden, there is always a vertex whose closed neighborhood intersects a constant fraction of the ``unwanted'' edges.
This again shows that in quasipolynomial time we can produce  $n^{\Oh(\log n)}$ branches, in each of which the chosen link can indeed be split into two parts.

Starting from the initial chain $(A,B,C)$, we repeatedly refine the smallest link of \emph{type $A$} (i.e., one that originates from $A$), the smallest link of \emph{type $B$}, and the smallest link of type $C$.
In each refinement step the size of a smallest link of the given type is halved, so after $\Oh(\log n)$ refinements for each type, we arrive at a chain containing an empty link of each of the three types.
These three empty links separate the graph into three pairwise nonadjacent subinstances, each missing one of sets $A,B,C$ and therefore having size at most roughly $2n/3$.
We solve these subinstances recursively and sum their optimum values.

Summarizing the above, the algorithm has three nested recursive layers.
One application of the refinement lemma costs $n^{\Oh(\log n)}$ branches.
Repeating refinements until an empty link of each type appears creates a tree of depth $\Oh(\log n)$ and hence the number of instances produced is $n^{\Oh(\log^2 n)}$.
Finally, each leaf instance splits into three pairwise nonadjacent subinstances of size at most roughly $2n/3$ each, which contributes the third logarithmic factor in the exponent.
Thus, the overall running time is $n^{\Oh(\log^3 n)}$.

\medskip

Now, let us move to \cref{thm:main}~(3) and discuss $H = \abbak$ for any fixed $k$.
Again, for simplicity, let us focus on the case $k=0$, i.e., $H = \abba$.
The first step is a standard branching on high-degree vertices~\cite{DBLP:journals/algorithmica/BacsoLMPTL19}.
Setting a threshold $\tau = n^{1/3}$, we branch on every vertex of degree at least $\tau$ until all remaining instances have maximum degree below $\tau$.
This already costs only subexponential time, so from now on we may assume that every vertex has degree at most $\tau$.

The next step is to partition the ordered vertex set into consecutive segments $Z_1,\ldots,Z_t$ so that each segment $Z_i$ comes with an edge $x_iy_i$ spanning it, and edges between different segments may appear only between consecutive ones.
Such a partition can be obtained by greedily picking, at each step, an edge leaving the current prefix whose right endpoint is as far to the right as possible.
Thus, after this partition, the graph behaves like a path of segments.
The forbidden pattern \abba implies that once the neighborhood of $\{x_i,y_i\}$ is removed, the rest of the segment  $Z_i$ becomes independent.

Now we look for a place where the instance can be split.
If some small segment (i.e., of size at most $n^{2/3}$) appears ``in the middle'' of the graph, then by branching on its intersection with the optimum solution we separate the graph into two independent subinstances, each of size at most roughly $2n/3$, and recurse on both sides.
Otherwise, every segment meeting the middle part is large, so there can be only few of them.
In this case we branch on a set of size $\Oh(n^{2/3})$ containing the two ``boundary'' small segments together with the neighborhoods of all $x_i,y_i$ for the middle segments $Z_i$.
After removing this set, the middle part becomes bipartite, while the remaining left and right parts are again multiplicatively smaller and can be solved recursively.

Thus, every branch either is solved directly or reduces the problem to instances of size at most about $2n/3$, while paying only a subexponential branching cost per level.
The overall running time is $n^{\Oh(n^{2/3}\log^3 n)}$.

\medskip
Finally, let us say a few words about the hardness counterpart of \cref{thm:main}.
We show several hardness reductions, each of which is tailored to a specific family of forbidden ordered graphs $H$.
Typically, we start with an unordered instance of \MIS, then we use the Poljak's subdivision trick to subdivide certain edges (possibly many times), and then we order the resulting graph in a way that ensures that $H$ is not an induced subgraph.
It is possibly worth mentioning that in vast majority of cases we can actually exclude certain graphs $H$ as a \emph{subgraph}, which is a stronger condition than excluding them as an induced subgraph.

\paragraph{Organization of the paper.}
In \cref{sec:prelim} we give present definitions and preliminary results.
In \cref{sec:algos} we present the algorithms for the positive cases of \cref{thm:main}.
In particular, statement (1) is proved in \cref{sec:poly};
statement (2) is proved in \cref{sec:aabb} (special case $H = \aabb$), \cref{sec:aacbb} (case $H = \aakbb$), and \cref{sec:ababk} (case $H = \ababk$);
and statement (3) is proved in \cref{sec:abbak}.
In \cref{sec:hardness} we prove the hardness results for the negative cases of \cref{thm:main}.
Finally, \cref{sec:outro} concludes with open problems and directions for future research.

%% file: prelim.tex
All logarithms in the paper are binary.
Throughout the paper, we consider vertex-weighted graphs.
We always assume that all computations on weights can be performed in constant time.
For a vertex-weighted graph $(G,\wei)$ and a set $X \subseteq V(G)$, we define $\wei(X) = \sum_{v \in X} \wei(v)$.

\paragraph{Ordered graphs.}
Let $G$ be a graph with a fixed linear ordering of its vertex set.
We write $v \prec u$ if $v$ is before $u$ in the ordering, and $v \preceq u$ if $v \prec u$ or $v=u$.
For a set of vertices $X$, we write $X \prec v$ if $x \prec v$ for every $x \in X$, and $v \prec X$ if $v \prec x$ for every $x \in X$. We also write, e.g., $x \prec Y$ for $\{x\} \prec Y$.
We define operators $\succ$ and $\succeq$ analogously.

The ordering of vertices induces the lexicographic ordering of edges: for distinct edges $uv, u'v'$, where $u \prec v$ and $u' \prec v'$, we write that $uv \prec u'v'$ if either $u \prec u'$ or $u=u'$ and $v \prec v'$.

By $N(v)$ we denote the neighborhood of $v$ in $G$, and by $N[v]$ we denote the closed neighborhood of $v$ in $G$, i.e., $N[v] = N(v) \cup \{v\}$.
We say that $u$ is a \emph{left neighbor} (resp., \emph{right neighbor}) of $v$ if $u \in N(v)$ and $u \prec v$ (resp., $v \prec u$).

The \emph{mirror image} of an ordered graph $G$ is the ordered graph obtained from $G$ by reversing the order of its vertices.
By the \emph{underlying graph} of an ordered graph $G$, we mean the unordered graph obtained from $G$ by forgetting the order of its vertices.

A \emph{segment} is a set of consecutive vertices in the ordering. A \emph{prefix} (resp., \emph{suffix}) is a segment that contains the first (resp.,  last) vertex in the ordering.
An edge $uv$ is a \emph{under} and edge $u'v'$ if $u' \leq u \prec v \prec v'$.
An edge $uv$ is \emph{maximal} if there is no other edge $u'v'$ such that $u' \preceq u \preceq v \preceq v'$.
For two edges \emph{cross} whenever they induce \abab.


\paragraph{\MWIS in ordered graphs.}

By $\alpha(G)$ we denote the independence number of a graph $G$, i.e., the maximum size of an independent set in $G$.
For a vertex-weighted graph $(G,\wei)$, by $\alpha(G,\wei)$ we denote the maximum weight of an independent set in $G$.

\begin{lemma}\label{lem:pacman}
    Let $H$ be an ordered graph, and let $H^*$ be the graph obtained from $H$ by adding an isolated vertex before all vertices of $H$, and an isolated vertex after all vertices of $H$. 
    Then, \MWIS in $H$-free graphs and \MWIS in $H^*$-free graphs are polynomially equivalent.
\end{lemma}
\begin{proof}
    Note that every graph that is $H$-free is also $H^*$-free, so a reduction from \MWIS in $H$-free graphs to \MWIS is $H^*$-free graphs is trivial.
    Thus, let us show a Turing reduction from \MWIS  in $H^*$-free graphs to \MWIS in $H$-free graphs.

    Let $(G,\wei)$ be an instance of \MWIS, where $G$ is $H^*$-free and has $n$ vertices.
    We exhaustively guess the first and the last vertex of some fixed optimum solution $I^*$; call these vertices, respectively, $x$ and $y$.    
    Clearly it only makes sense to consider choices when $xy \notin E(G)$. We note that if $|I^*|=1$, it might happen that $x=y$.

    This step results in at most $n^2$ branches; consider one of them.    
    We remove all vertices preceding $x$, all vertices following $y$, and the closed neighborhood of $\{x,y\}$.
    Let $G'$ be the resulting graph. Note that is it $H$-free as any induced copy of $H$, together with $x$ and $y$,
    gives an induced copy of $H^*$ in $G$.
    
    Now we observe that in the branch where $x,y$ were guessed correctly,
    $I^* \setminus \{x,y\}$ is an optimum solution for \MWIS in $(G',\wei)$.
    Thus, we reduced solving \MWIS on $(G,\wei)$ to solving \MWIS on $\Oh(n^2)$ instances, each of which is $H$-free and has at most $n$ vertices.
\end{proof}

Repeated application of \cref{lem:pacman} yields the following.
\begin{corollary}\label{cor:pacman}
    Let $H$ be an ordered graph and and let $k \geq 1$ be a fixed integer.
    Let $H^*$ be the graph obtained from $H$ by adding $k$ isolated vertices before all vertices of $H$, and another $k$ isolated vertices after all vertices of $H$. 
    Then, \MWIS in $H$-free graphs and \MWIS in $H^*$-free graphs are polynomially equivalent.
\end{corollary}

%% file: algos.tex
In this section we present the algorithms for the positive cases of \cref{thm:main}.
By \cref{cor:pacman}    it suffices to present algorithms for the cases when $H$ is exactly \pthree, \chord, \oneedgek (for statement (1) of \cref{thm:main}), \aakbb, \ababk (for statement (2)), and \abbak (for statement (3)).

%% file: poly.tex
Let us start with proving \cref{thm:main} (1).

\begin{theorem}
   \MWIS in can be solved in polynomial time in \pthree-free graphs.
\end{theorem}
\begin{proof}
    Let us observe that every  \pthree-free graph is a comparability graph.
    Indeed, if $G$ is an ordered \pthree-free graph, then (the reflexive closure of) the relation of being a right neighbor
    is a partial order.
    Now, the claim follows by recalling that \MWIS can be solved in polynomial time in comparability graphs~\cite{Golumbic2004}.
\end{proof}

\begin{theorem}
   \MWIS in can be solved in polynomial time in \chord-free and in \chordrev-free graphs.
\end{theorem}
\begin{proof}
    As \chordrev-free graphs are exactly the mirror images of \chord-free graphs, it suffices to prove the claim for \chord-free graphs.
    Let us observe that every \chord-free graph is chordal:
    the ordering of $V(G)$ is a perfect elimination ordering (see also~\cite{DBLP:journals/siamdm/HajebiLS24}).
    As \MWIS is polynomial-time solvable in chordal graphs~\cite{DBLP:journals/siamcomp/Gavril72}, the claim follows.
\end{proof}

\begin{theorem}
   For every fixed $k$, \MWIS can be solved in polynomial time in \oneedgek-free graphs.
\end{theorem}
\begin{proof}
    Let $G$ be a \oneedgek-free ordered graph with consecutive vertices $v_1,\ldots,v_n$.
    Let $\wei$ be the weight function on $V(G)$.
    
    The algorithm is based on dynamic programming, processing the subsequent prefixes of the vertex set.
    We introduce a dynamic programming table $\mathsf{Tab}[i, I]$, where $i \in \{0,\ldots,n\}$ and $I \subseteq \{v_1,\ldots,v_i\}$ is an independent set of size at most $k$. 
    
    We say that an independent set $I^*$ is \emph{compatible} with $(i,I)$ if $I^* \subseteq \{v_1,\ldots,v_i\}$ and
    \begin{description}
        \item[Case A: $|I|=k$:] the last $k$ vertices of $I^*$ are $I$, 
        \item[Case B: $|I|<k$:] $I^* = I$.
    \end{description}
    The intended value in $\mathsf{Tab}[i, I]$ is the maximum weight of an independent set compatible with $(i,I)$.
    Clearly, having computed all entries of $\mathsf{Tab}[i, I]$, we can return $\max_{I} \mathsf{Tab}[n, I]$.

    We initialize $\mathsf{Tab}[0, \emptyset] =0$ and all other values to $-\infty$.
    Suppose now that $i \geq 1$ and we have computed all entries $\mathsf{Tab}[i-1, \cdot]$.
    We aim to compute $\mathsf{Tab}[i, I]$ for some fixed independent set $I$ contained in $\{v_1,\ldots,v_i\}$.

    If $|I| < k$, we simply store the weight of $I$.
    Now suppose that $|I| = k$.
    If $v_i \notin I$, we set $\mathsf{Tab}[i, I] = \mathsf{Tab}[i-1, I]$.
    So consider the case $v_i \in I$.
    For an independent set $I' \subseteq \{v_1,\ldots,v_{i-1}\}$, we write $I' \sqsubset I$ if either $I' = I \setminus \{v_i\}$ or $I'$ is obtained from $I$ by removing $v_i$ and adding a vertex $u$ such that $u \prec I$ and $uv_i \notin E$.
    We set 
    \begin{equation}
        \mathsf{Tab}[i, I] = \max_{ I' \sqsubset I}  \mathsf{Tab}[i-1, I'] + \wei(v_i).         \label{setdynprog}
    \end{equation}
    It is straightforward to verify that the value of $\mathsf{Tab}[i, I]$ is at least its intended value.
    Let us argue that it is not larger, i.e., the value set in step \eqref{setdynprog} corresponds to the weight of some independent set compatible with $(i,I)$.
    Let $I' \sqsubset I$ such that $\mathsf{Tab}[i, I] = \mathsf{Tab}[i-1, I'] + \wei(v_i)$.
    If $|I'| < k$, then the claim is clear.
    So suppose that $|I'|=k$, i.e., $I' = I \setminus \{v_i\} \cup \{u\}$ for some $u \prec I$.

    Let $I^*$ be a maximum-weight independent set in $\{v_1,\ldots,v_{i-1} \}$ compatible with $I'$.
    Clearly, $\mathsf{Tab}[i, I] = \wei(I^* \cup \{v_i\})$. We claim that $I^* \cup \{v_i\}$ is independent.
    As each of $I^*, I$ is independent, suppose that there is $w \in I^*$ which is adjacent to $v_i$.
    As $v_iu \notin E$ and $I^*$ is compatible with $I'$, we must have $w \prec u$.
    However, now the set of vertices $\{w\} \cup I' \cup \{v_{i} \}$ induces a copy of \oneedgek in $G$, a contradiction.

    The number of all entries in $\mathsf{Tab}[\cdot,\cdot]$ is bounded by $n^{k+1}$, i.e., is polynomial in $n$.
    As each entry is computed in polynomial time, the overall running time is also polynomial.
\end{proof}

%% file: aabb.tex
Now let us proceed to quasipolynomial-time algorithms, i.e., the proof of \cref{thm:main} (2).
Before we show the algorithm for the case that $H = \aakbb$ for any fixed $k$, let us present
an algorithm for the case $k=0$, i.e., for $H = \aabb$.
This algorithm is significantly simpler than the one for general $k$ and has a lower running time, so it serves as a good warm-up for the more general case.

\begin{theorem}\label{thm:aabb}
    \MWIS in $n$-vertex \aabb-free graphs can be solved in time $n^{\Oh(\log n)}$.    
\end{theorem}

Let us start with some observations.

Let $G$ be an ordered graph.
A \emph{seagull} is a triple $(x,y,z)$ of distinct vertices of $G$ such that $x \prec y \prec z$
and $xy, yz \in E(G)$. Thus, $G[\{x,y,z\}]$ is either a triangle or \pthree.

\begin{lemma}\label{lem:seagullbipartite}
    An ordered graph with no seagulls is bipartite.
\end{lemma}
\begin{proof}
    Suppose that $G$ is an ordered graph with no seagulls.
    We can assume that $G$ has no isolated vertices, as their existence does not influence bipartiteness.
    Thus, we can partition the vertices of $G$ into the following sets:
    \begin{align*}
        X =  \{ v ~|~ N(v) \prec v\} \quad \text{and} \quad
        Y =  \{ v ~|~ N(v) \succ v\}.
    \end{align*}
    We observe that each of $X$ and $Y$ is an independent set.
    Thus, $G$ is indeed bipartite.
\end{proof}

\begin{lemma}\label{lem:seagulldisjoint}
    Let $G$ be a \aabb-free ordered graph.
    Then, $G$ has no two seagulls that are disjoint and nonadjacent.
\end{lemma}
\begin{proof}
    For contradiction, suppose that $G$ has two disjoint and nonadjacent seagulls $(x,y,z)$ and $(x',y',z')$.
    If $y' \prec y$, then $\{x',y',y,z\}$ is an induced copy of \aabb in $G$.
    If $y' \succ y$, then $\{x,y,y',z'\}$ is an induced copy of \aabb in $G$.
    In both cases we reach a contradiction, which completes the proof.
\end{proof}

Combining \cref{lem:seagullbipartite} and \cref{lem:seagulldisjoint}, we obtain the following.

\begin{corollary}\label{cor:seagull}
    Let $G$ be a \aabb-free ordered graph that is not bipartite.
    Then there is $v \in V(G)$ that belongs to a seagull and such that $N[v]$ intersects one third of all seagulls.
    Furthermore, such $v$ can be found in polynomial time.
\end{corollary}
\begin{proof}
    Since $G$ is not bipartite, by \cref{lem:seagullbipartite} it contains a seagull $(x,y,z)$.
    By \cref{lem:seagulldisjoint} we know that $N[\{x,y,z\}]$ intersects all seagulls in $G$.
    Thus, by averaging,  there is $v \in \{x,y,z\}$ such that $N[v]$ intersects one third of all seagulls.    
    The running time bound follows from the fact that the number of all seagulls is $\Oh(n^3)$,
    so exhaustive enumeration works in polynomial time.
\end{proof}

Now we are ready to present our algorithm.

\begin{proof}[Proof of \cref{thm:aabb}]
    Let $(G,\wei)$ be an instance of \MWIS, where $G$ has $n$ vertices and is \aabb-free.

    Our algorithm is a recursive procedure.
    The base case is that the instance graph is bipartite.
    Then, we solve the problem in polynomial time using the K\H{o}nig-Egerv\'ary theorem and flow techniques~\cite{LovaszPlummer1986_MatchingTheory}. Note that this in particular covers the case that $n \leq 2$.

    Thus, assume that $G$ is not bipartite.
    Let $v$ be a vertex given by applying \cref{cor:seagull} on $G$.
    We branch on $v$, i.e., we call the algorithm recursively for the instance $G-v$ and for the instance $G-N[v]$.
    In the first branch we seek for independent sets that do not contain $v$, and in the second branch we seek for independent sets that contain $v$.
    We return the maximum of the value found in the first branch, and the value found in the second branch plus $\wei(v)$.

    The correctness of the algorithm is obvious. Let us argue about the running time.
    We will measure the potential of the instance by the number $\mu$ of seagulls.
    If $\mu = 0$, the graph is bipartite and thus the instance can be solved in polynomial time.
    The number of instances created by branching is given by the recursive inequality
    \[
        F(\mu) \leq F(\mu-1) + F(2\mu/3),
    \]
    which solves to $\mu^{\Oh(\log \mu)}$ (see, e.g.,~\cite{DBLP:journals/siamdm/DebskiPR22} for similar recursions and how to solve them).
    Recalling that $\mu \leq n^3$, we conclude that overall running time of the algorithm is bounded by $n^{\Oh(\log n)}$. 
\end{proof}

%% file: aacbb.tex
We are ready to show a quasipolynomial-time algorithm for the case $H = \aakbb$ for any fixed $k$.

\begin{theorem}\label{thm:aakbb}
For every fixed $k$, \MWIS in $n$-vertex \aakbb-free graphs can be solved in time $n^{\Oh(\log^2 n)}$.    
\end{theorem}

\begin{proof}
Without loss of generality, let us assume that $k \geq 1$.
Let $(G,\wei)$ be an instance of \MWIS such that $G=(V,E)$ is an $n$-vertex \aakbb-free ordered graph.
Both $G$ and $\wei$ remain unchanged throughout the course of the algorithm, and we will represent certain induced subgraphs of $G$ by their vertex sets.

\paragraph{Structured instances.}
Throughout the algorithm, we will work with instances of \MWIS equipped with some additional structure.
A \emph{structured instance} of \MWIS is a triple $(X,Y,Z)$, where:
\begin{itemize}
\item $X,Y,Z$ are pairwise disjoint subsets of $V$,
\item $X \prec Y \prec Z$, and
\item $X,Z$ are independent sets.
\end{itemize}
For a structured instance $\cI = (X,Y,Z)$, we define $\alpha(\cI):=\alpha(G[X\cup Y\cup Z],\wei)$.

The set $Y$ of a structured instance $(X,Y,Z)$ is called the \emph{workspace}. We will measure the size of a structured instance by the size of its workspace.
We can see the initial instance $(G,\wei)$ of \MWIS as the structured instance $\cI_0 := (\emptyset,V,\emptyset)$ with measure $n$.
On the other hand, a structured instance of measure 0 is bipartite, thus it can be solved in polynomial time.

For structured instances $\cI = (X,Y,Z)$ and $\cI' = (X',Y',Z')$,
we say that $\cI'$ is a \emph{subinstance} of $\cI$ if $X \cap (X' \cup Y' \cup Z') \subseteq X'$ and $Z \cap (X' \cup Y' \cup Z') \subseteq Z'$.
In other words, $\cI'$ was obtained from $\cI$ by deleting some vertices and moving some vertices from the workspace to the independent sets. Note that the measure of $\cI'$ does not exceed the measure of $\cI$.

The main idea of the algorithm is to reduce solving our original structured instance $(\emptyset,V,\emptyset)$ to solving a number of structured subinstances, each of measure 0. We will do this by branching on carefully chosen vertices.

\paragraph{Main branching.} The following claim encapsulates the main step of the algorithm.

\begin{claim}\label{clm:line-outer}
Given a structured instance $\cI = (X,Y,Z)$ of measure $\mu > 1$,
in time $\mu^{\Oh(\log \mu)}$ we can construct a family $\mathscr{I}$ of $\mu^{\Oh(\log \mu)}$ pairs $(\cI',\ell')$ such that:
\begin{enumerate}[(P1)]
	\item for every $(\cI',\ell') \in \mathscr{I}$, $\cI'$ is a structured subinstance of $\cI$ and $\ell' \in \mathbb{Q}_{\geq 0}$, \label{prop:subinst}
    \item for every $(\cI',\ell') \in \mathscr{I}$, the measure of $\cI'$ is at most $\lceil \mu/2 \rceil$, and \label{prop:measuredrops}
    \item $\alpha(\cI)=\max_{(\cI',\ell')\in \mathscr{I}} ( \alpha(\cI')+\ell' )$. \label{prop:solution}
\end{enumerate}
\end{claim}
Before we prove \cref{clm:line-outer}, let us first argue that it is sufficient to prove the theorem.
We construct the recursion tree whose nodes are labeled by pairs $(\cI,\ell)$, where $\cI$ is a subinstance of $\cI_0$ and $\ell \in \mathbb{Q}_{\geq 0}$.
The root of this tree is labeled by $(\cI_0,0)$.
Consider a node $d$ that has not been processed yet; let it be labeled with the pair $(\cI,\ell)$.
If $\cI$ has a measure of at most 1, we do not do anything; $d$ becomes a leaf of the recursion tree.
Otherwise, we apply \cref{clm:line-outer} to $\cI$, obtaining the family $\mathscr{I}$.
For each $(\cI',\ell') \in \mathscr{I}$, we create a new child node of $d$ and label it with $(\cI',\ell')$.
We perform this construction exhaustively.

Note that all leaves of the recursion tree correspond to instances of measure at most 1 and they can be solved in polynomial time.
Indeed, if the measure is 0, the instance is bipartite.
If the measure is 1, we can guess whether the unique vertex of the workspace belongs to the sought-for optimum solution and continue with an instance of measure 0.
By processing the tree from leaves to the root and using \ref{prop:solution} from \cref{clm:line-outer} at each step, we can then solve $\cI_0$ (equivalently, $(G,\wei)$).

Now let us argue about the running time.
The number of nodes in the recursion tree is proportional to the number of leaves,
which is given by the following inequality:
\[
	F(\mu) \leq \mu^{\Oh(\log \mu)} \ F(\lceil \mu/2 \rceil) = \mu^{\Oh(\log^2 \mu)}.
\]
As local computation in each node (of measure $\mu$) requires time $\mu^{\Oh(\log \mu)}$, and the measure of $\cI_0$ is $n$,
we obtain that the overall running time is bounded by $n^{\Oh(\log^2 n)}$, as claimed.

\paragraph{Proof of \cref{clm:line-outer}.} Thus, we are left with proving \cref{clm:line-outer}.
It will be convenient to define a refined variant of a structured instance.
A \emph{refined instance} is a quadruple $(X,Y_L,Y_R,Z)$, where $Y_L,Y_R$ are disjoint sets such that $Y_L \prec Y_R$ and $(X, Y_L \cup Y_R, Z)$ is a structured instance.
In other words, a refined instance is obtained from a structured instance $(X,Y,Z)$ by splitting $Y$ into a prefix $Y_L$ and a suffix $Y_R$. Having this in mind, we say that $(X,Y_L,Y_R,Z)$ is a \emph{split} of $(X,Y,Z)$.
We also define $\alpha((X,Y_L,Y_R,Z)) = \alpha((X,Y,Z))$.
We say that a refined instance $(X',Y'_L,Y'_R,Z')$ is a subinstance of a refined instance $(X,Y_L,Y_R,Z)$ if $X' \subseteq X$, $Y'_L \subseteq Y_L$, $Y'_R \subseteq Y_R$, and $Z' \subseteq Z$.

Let $\cI = (X,Y,Z)$ be a structured instance of measure $\mu > 0$, i.e., with $Y \neq \emptyset$.

\subparagraph{Creating a family $\mathscr{J}$ of splits.} 
Let $Y = \{y_1,\ldots,y_m\}$, where $y_1 \prec y_2 \prec \ldots \prec y_m$, and let $r=\lceil m/2 \rceil$.
Define $Y_L=\{y_1,\ldots,y_r\}$ and $Y_R=\{y_{r+1},\ldots,y_m\}$, and note that $\cJ_0= (X,Y_L,Y_R,Z)$ is a split of $\cI$.
As an intermediate step towards the construction of $\mathscr{I}$, we will construct an auxiliary family $\mathscr{J}$.
Each element of $\mathscr{J}$ is a pair $(\cJ,\ell)$, where $\cJ$ is a refined subinstance of $\cJ_0$, and $\ell \in \mathbb{Q}_{\geq 0}$.

First, we exhaustively guess the first $k$ vertices of the sought-after optimum solution $I^*$ that are in $Y_R$, or the whole set $I^* \cap Y_R$ if $|I^* \cap Y_R| < k$.
More formally, for every independent set $I \subseteq Y_R$ of size at most $k$, we define a refined instance $\cJ^I=(X^I,Y_L^I,Y_R^I,Z^I)$ which is a subinstance of $\cJ$, as follows:
\begin{itemize}
\item $X^I = X \setminus N(I)$, $Y_L^I = Y_L \setminus N(I)$, $Z^I = Z \setminus N(I)$,
\item if $|I| < k$, then $Y^I_R = \emptyset$,
\item if $|I| = k$ and $y^*$ is the last vertex of $I$, then $Y^I_R$ is obtained from $Y_R$ by removing  $N[I]$ and all vertices that precede $y^*$.
\end{itemize}
Furthermore, we set $\ell^I = \wei(I)$ and include the pair $(\cJ^I,\ell^I)$ into $\mathscr{J}$.
This concludes the construction of $\mathscr{J}$.
Clearly, the size of the constructed family is $\Oh(\mu^k)$, i.e., polynomial in $\mu$ as $k$ is a constant.
Furthermore, $\mathscr{J}$ can be constructed in polynomial time.

By the construction, it is clear that
\begin{equation}\label{eq:J}
	\alpha(\cI) = \alpha(\cJ_0)=\max_{ (\cJ,\ell) \in \mathscr{J} } \alpha(\cJ) + \ell.
\end{equation}

Now let us analyze the structure of instances in $\mathscr{J}$.
Fix any element $(\cJ,\ell) \in \mathscr{J}$, where $\cJ = (X',Y'_L,Y'_R, Z')$.
Since $\cJ$ is a subinstance of $\cJ_0$ and $k \geq 1$, we have that
\begin{equation}\label{eq:sizeY}
|Y'_L| \leq \lceil \mu/2 \rceil \quad \text{ and } \quad |Y'_R| \leq \lceil \mu/2 \rceil.
\end{equation}
Now, define $\Left=X' \cup Y'_L$ and $\Right=Y'_R \cup Z'$, and $E_L = E(G[\Left])$ and $E_R = E(G[\Right])$.
Let us consider the case that one of the sets $E_L,E_R$ is empty.
Suppose that $E_L = \emptyset$; the case $E_R = \emptyset$ is analogous.
In such a case $\cI' = (X' \cup Y'_L,Y'_R,Z')$ is a structured instance,
which is a subinstance of $\cI$.
Furthermore, as $|Y'_R| \leq  \lceil |Y|/2 \rceil$ by \eqref{eq:sizeY},
we see that the measure of this structured instance is at most $\lceil |Y|/2 \rceil =\lceil \mu/2 \rceil$.
Thus, the pair $(\cI',\ell')$ can be inserted to the family $\mathscr{I}$ that we construct.
We call $\cI'$ obtained from $\cJ$ that way a \emph{derived structure instance}.

Thus, the overall strategy is to branch on carefully selected vertices in order to make one of the sets $E_L,E_R$ empty.
This justifies measuring the size of a refined subinstance $\cJ$ by the value $\xi(\cJ) = |E_L| + |E_R|$. 

The crucial structural observation is as follows.

\begin{claim}\label{clm:dominating-left-right}
	Let $u,v\in \Left$ (resp.,  $\Right$) be such that $uv\in E$.
    Then, $\Right\setminus N[\{u,v\}]$ (resp.,  $\Left\setminus N[\{u,v\}]$) is an independent set.
\end{claim}
\begin{claimproof}
Let $I$ be the independent set for which the pair $(\cJ,\ell) \in \mathscr{J}$ was created.
Let $uv$ be an edge with both endpoints in $\Left$ (the other case is symmetric).
For contradiction, suppose that there are $u',v'\in \Right\setminus N[\{u,v\}]$ such that $u'v'\in E$.
Since $Z' \subseteq Z$ is independent, at least one of $u',v'$ must belong to $Y'_R$.
In particular, $Y'_R \neq \emptyset$, which means that $|I|=k$.

By the construction of $\cJ$, we have that $\Left \prec I \prec \Right$ and $I$ is nonadjacent to $\Left\cup \Right$.
Therefore, the set $\{u,v\}\cup I \cup \{u',v'\}$ induces a copy of \aakbb in $G$, a contradiction.
\end{claimproof}

\subparagraph{Internal branching.} 
For each $(\cJ,\ell) \in \mathscr{J}$, we perform branching described below.
Again, we can imagine this process as building a recursion tree $\mathsf{T}$, whose nodes are labeled with pairs $(\cJ',\ell')$,
where $\cJ'$ is a subinstance of $\cJ$ and $\ell' \in \mathbb{Q}_{\geq 0}$.
The root of the tree is labeled by $(\cJ,\ell)$.

Now, consider a node $d$ of $\mathsf{T}$ that has not been processed yet.
Let its label be $(\cJ',\ell')$, where $\cJ' = (X',Y'_L,Y'_R,Z')$.
We define $\Left,\Right,E_L,E_R$ for $\cJ'$ as before.
If $E_L = \emptyset$ or $E_R = \emptyset$, we do nothing and $d$ becomes a leaf of the recursion tree.

So suppose otherwise, and by symmetry, assume that $|E_L| \geq |E_R|$. This means that $|E_L| \geq \xi(\cJ')/2$.
Pick any edge $uv \in E_R$.
By \cref{clm:dominating-left-right}, we know that $N[\{u,v\}]$ intersects all edges in $E_L$.
Consequently, there is $x \in \{u,v\}$ such that $N[x]$ intersects at least half of the edges in $E_L$, i.e., at least one quarter of the edges in $E_L \cup E_R$.

We add two children to $d$ in $\mathsf{T}$.
One of them corresponds to \emph{not including} $x$ in the solution and is labeled by the instance $(X',Y'_L,Y'_R \setminus \{x\}, Z' \setminus \{x\})$ and number $\ell'$.
The second child corresponds to \emph{including} $x$ in the solution.
It is labeled by the instance $(X' \setminus N[x], Y'_L \setminus N[x], Y'_R \setminus N[x], Z' \setminus N[x])$, and the number $\ell' + \wei(x)$.

Note that the measure of both created instances drops compared to the measure of $\cJ'$.
On one hand, removing just $x$ decreases the measure by at least 1, as it destroys the edge $uv$.
On the other hand, removing $N[x]$ from the graph decreases the measure to at most $|E_L|/2 + |E_R| \leq 3\xi(\cJ')/4$.

We proceed in this manner exhaustively.
Now, having built $\mathsf{T}$, consider a leaf node and let it be labeled with $(\cJ',\ell')$, where $\cJ' = (X',Y'_L,Y'_R,Z')$. As either $X' \cup Y'_L$ or $Y'_R \cup Z'$ is independent, there is a structured instance  $\cI'$ derived from $\cJ'$. We include $(\cI',\ell')$ into $\mathscr{I}$.
This completes the definition of this family.

\subparagraph{Correctness.} 
Let us verify that $\mathscr{I}$ satisfies all properties listed in \cref{clm:line-outer}.
The size of $\mathscr{I}$ is equal to the number of leaves in all the recursion trees created for all members of $\mathscr{J}$.
The number of leaves of one such tree is given by a recursion
\[
	R(\xi) \leq R(\xi-1) + R(3\xi/4) = \xi^{\Oh(\log \xi)}.
\]
Since $\xi = \Oh(\mu^2)$, we obtain that the number of leaves on one recursion tree is $\mu^{\Oh(\log \mu)}$.
As the number of all trees is bounded by $|\mathscr{J}| = \Oh(\mu^k)$, we conclude that the size of $\mathcal{I}$ is $\mu^{\Oh(\log \mu)}$, as required. Furthermore, the family can clearly be enumerated in time $\mu^{\Oh(\log \mu)}$.

Property \ref{prop:subinst} from \cref{clm:line-outer} follows directly from the construction.
Property \ref{prop:measuredrops} follows from \eqref{eq:sizeY} and the definition of a derived structured instance.
Finally, \ref{prop:solution} follows easily from combining \eqref{eq:J} with a standard bottom-up induction on a recursion tree. This completes the proof.
\end{proof}

%% file: abab.tex
In this section, we show a quasipolynomial-time algorithm for the case $H = \ababk$ for any fixed $k$.
The general approach is somewhat similar to the one for the case $H = \aakbb$, but the details are significantly more involved.

\begin{theorem}\label{thm:abab}
   For every fixed $k$, \MWIS in $n$-vertex \ababk-free graphs can be solved in time $n^{\Oh(\log^3 n)}$.    
\end{theorem}

\paragraph{Chains.} 
Similarly to \cref{sec:aacbb}, we will work with instances equipped with some additional structure.
Let $(G,\wei)$ be an instance of \MWIS, where $G=(V,E)$ is an ordered \ababk-free graph.
A \emph{chain} (over $(G,\wei)$) is a sequence $\cX = (X_1,\ldots,X_r)$ of $r \geq 3$ pairwise disjoint subsets of $V$, called \emph{links},
such that
\begin{itemize}
\item $X_1 \prec X_2 \prec \ldots \prec X_r$,
\item if $uv$ is an edge of $G$ where $u \in X_i$ and $v \in X_j$ for $i < j$,
then either $j = i+1$ or $i=1$ and $j=r$.
\end{itemize}
We often identify chains with the unions of all links,
i.e., we write $G[\cX]$ for $G[\bigcup \cX]$, and for a set $Y \subseteq V(G)$,
we write $\cX - Y$ for $(X_1 \setminus Y,\ldots,X_r \setminus Y)$.
Furthermore, we define $\alpha(\cX) = \alpha(G[\cX],\wei)$.
Note that if $\cX$ is a chain, then every edge of $G[\cX]$ is either local (i.e., contained in one link),
or joins two cyclically consecutive links.

We say that a chain $\cX' = (X'_1,\ldots,X'_{r+1})$ is a \emph{refinement} of a chain $\cX = (X_1,\ldots,X_r)$ \emph{at} $j \in [r]$ if
\begin{itemize}
\item for all $i < j$, we have $X'_i \subseteq X_i$,
\item $X'_j,X'_{j+1} \subseteq X_j$, and
\item for all $i > j$, we have $X'_{i+1} \subseteq X_i$.
\end{itemize}
In other words, $\cX'$ was obtained by splitting the link $X_j$ into two sets and possibly removing some vertices.

We emphasize that we never require the links of a chain to be non-empty.
Actually, our approach is to start with a partition of $V$ into three segments (note that this is a chain)
and keep refining this chain until we obtain some empty links, and thus disconnect the graph.

The following  lemma encapsulates the main technical step of the algorithm:
that in quasipolynomial time, one can obtain a refinement of a given chain at any given link.

\begin{lemma}\label{lem:refine}
Let $k \geq 1$ be an integer and let $(G,\wei)$ be an instance of \MWIS such that $G$ is an $n$-vertex \ababk-free ordered graph.
Let $\cX = (X_1,\ldots,X_r)$ be a chain over $(G,\wei)$ and let $j \in [r]$.
In time $n^{\Oh(\log n)}$, we can construct a family $\mathscr{Z}$ of $n^{\Oh(\log n)}$ pairs $(\cZ,\ell)$, where $\cZ$ is a refinement of $\cX$ at $j$ and $\ell \in \mathbb{Q}_{\geq 0}$, such that
\[
\alpha(\cX)=\max_{(\cZ,\ell) \in \mathscr{Z}} \alpha(\cZ) + \ell.
\]
\end{lemma}

In the proof of \cref{lem:refine} we distinguish two cases: (1) $1<j<r$, (2) $j \in \{1,r\}$.
The general outline of the argument in both cases is the same, but the details differ.
Thus, we present these constructions separately.

\begin{proof}[Proof of \cref{lem:refine} for $1 < j < r$.]
The proof has two main steps.

\paragraph{Construction of $\mathscr{Y}$.}
Let us first construct an auxiliary family $\mathscr{Y}$,
which corresponds to guessing the first $k$ and the last $k$ vertices of the solution in $X_j$.
More precisely, for each independent set $I \subseteq X_j$ of size at most $2k$, we construct a chain $\cX^I$ as follows.
Let $I_1$ be the set of the first $k$ vertices of $I$, or $I_1 = I$ if $|I| < k$.
Let $I_2 = I \setminus I_1$.
Let $v_1^*$ be the last vertex of $I_1$ (if it exists), and let $v_2^*$ be the first vertex of $I_2$ (if it exists).
Then,  $\cX^I$ is obtained from $\cX$ by
\begin{enumerate}[(i)]
\item removing the vertices of $N[I]$, and
\item if $|I|=2k$, then removing $\{ x\in X_j \ | \ x\prec v_1^* \lor x\succ v_2^*\}$,
and otherwise removing all vertices of $X_j$.
\end{enumerate}
That completes the construction of $\cX^I$.
We emphasize that the links of $\cX^I$ are in one-to-one correspondence with the links of $\cX$.

We add to $\mathscr{Y}$ the pair $(\cX^I,\wei(I))$.
Note that $|\mathscr{Y}| = \Oh(n^{2k})$, i.e., is polynomial in $n$.
Furthermore, as each pair in $\mathscr{Y}$ corresponds to exhaustively guessing (at most) $k$ first and (at most) $k$ last vertices of the solution that are in $X_j$, it clearly holds that
\[
\alpha(\cX)=\max_{(\cY,\ell)\in \mathscr{Y}} \alpha(\cY)+\ell.
\]

The following claim shows the crucial structural property of elements in $\mathscr{Y}$.

\begin{figure}[t]
    \centering
   
    \begin{tikzpicture}[scale=0.9,every node/.style={draw,circle,fill=white,inner sep=0pt,minimum size=8pt},every loop/.style={}]
    
\foreach \k in {1,2,3,4}
{
\node (i1\k) at (\k*0.5-0.5,0) {};
\node (i2\k) at (5.5+\k*0.5,0) {};
}

\node[draw=orange,label=below:\footnotesize{$a$}] (a) at (-1.3,0) {};
\node[draw=orange,label=below:\footnotesize{$b$}] (b) at (2.8,0) {};
\node[draw=orange,label=below:\footnotesize{$c$}] (c) at (4.7,0) {};
\node[draw=orange,label=below:\footnotesize{$d$}] (d) at (8.8,0) {};

\node[draw=orange] (e) at (3.3,0) {};
\node[draw=orange] (e) at (3.75,0) {};
\node[draw=orange] (e) at (4.2,0) {};
\draw [decorate,decoration = {calligraphic brace,mirror}] (3.2,-0.2) --  (4.3,-0.2);
\node[draw=none,fill=none] (e) at (3.75,-0.45) {\footnotesize{$k$}};

\draw (2.5,0.6)--(5,0.6)--(5,-0.6)--(2.5,-0.6)--(2.5,0.6);
\draw (-3.5,0.6)--(-1,0.6)--(-1,-0.6)--(-3.5,-0.6)--(-3.5,0.6);
\draw (8.5,0.6)--(11,0.6)--(11,-0.6)--(8.5,-0.6)--(8.5,0.6);
\draw (-2.25,0.6) to [bend left] (3.75,0.6);
\draw (3.75,0.6) to [bend left] (9.75,0.6);

\draw[color=orange] (a) to [bend left] (c);
\draw[color=orange] (b) to [bend left] (d);

\node[draw=none,fill=none] (x) at (-2.25,0) {$X^I_{j-1}$};
\node[draw=none,fill=none] (x) at (3.75,-1) {$X^I_{j}$};
\node[draw=none,fill=none] (x) at (9.75,0) {$X^I_{j+1}$};

\node[draw=none,fill=none] (i) at (0.75,-1.5) {$I_1$};
\node[draw=none,fill=none] (i) at (6.75,-1.5) {$I_2$};
\draw [decorate,decoration = {calligraphic brace,mirror}] (-0.2,-1) --  (1.7,-1);
\draw [decorate,decoration = {calligraphic brace,mirror}] (5.8,-1) --  (7.7,-1);
\end{tikzpicture}

     \caption{Sets $X^I_{j-1}$, $X^I_j$, $X^I_{j+1}$, and independent sets $I_1$, $I_2$ (each of size $k$) in \cref{claim:abab1}.
     The orange subgraph cannot be induced, as otherwise, together with $I_1$ and $I_2$, it creates the forbidden structure.}
    \label{fig:claim-abab1}
\end{figure}
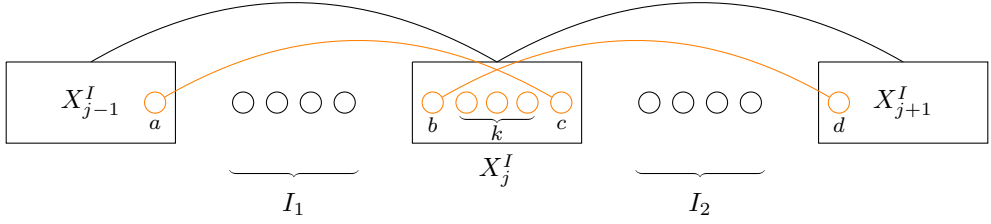

\begin{claim}\label{claim:abab1}
Let $(\cX^I,\ell^I)$, where $\cX^I= (X^I_1,\ldots,X^I_r)$, be a pair in $\mathscr{Y}$ constructed for an independent set $I$.
There is no induced copy of \abkab with edges $ac, bd$ such that $a \in X^I_{j-1}$ and $b,c \in X^I_j$, and $d \in X^I_{j+1}$.
\end{claim}
\begin{claimproof}
For contradiction, suppose that such a copy of \abkab with distinguished vertices $a,b,c,d$ exists.
Let $I'$, where $b \prec I' \prec c$, be the $k$-element independent set in this copy.

As $b,c \in X^I_j$, we have $|I| = 2k$.
Let $I_1$ (resp., $I_2$) be the first (resp., last) $k$ vertices of $I$.
Then we have $a\prec I_1 \prec b\prec I'\prec  c \prec I_2 \prec d$ and there are no edges between $I_1\cup I_2$ and $\{a,b,c,d\}\cup I'$.
Consequently, $\{a,b,c,d\} \cup I_1 \cup I' \cup I_2$ induces a copy of \ababk in $G$ (see \cref{fig:claim-abab1}), a contradiction.
\end{claimproof}

\paragraph{Construction of $\mathscr{Z}_{\cY}$.}
Fix some $(\cY,\ell) \in \mathscr{Y}$.
Now in time $n^{\Oh(\log n)}$ we will construct a family $\mathscr{Z}_{\cY}$ of pairs $(\cZ,\ell')$, where 
 $\cZ$ is a refinement of $\cY$ (and thus of $\cX$) at $j$ and $\ell' \in \mathbb{Q}_{\geq 0}$, such that
\[
\alpha(\cY)=\max_{(\cZ,\ell') \in \mathscr{Z}_{\cY}} \alpha(\cZ)+\ell'.
\]
Let us point out that proving this would conclude the proof of the lemma (in the considered case).
Indeed, the family 
\[
	\mathscr{Z} := \bigcup_{(\cY,\ell) \in \mathscr{Y}} \bigcup_{(\cZ,\ell') \in \mathscr{Z}_{\cY}} (\cZ,\ell + \ell')
\]
satisfies all properties listed in the lemma, and the overall running time and the size of $\mathscr{Z}$ are both bounded by $n^{\Oh(\log n)}$.

Thus, let us show how to construct  $\mathscr{Z}_{\cY}$ for $(\cY,\ell) \in \mathscr{Y}$.
The process can be again described as a recursion tree $\mathsf{T}$.
We start with the root node which is labeled with the pair $(\cY,0)$.
Next, we proceed exhaustively as follows.
Let $d$ be a node of $\mathsf{T}$ that was not processed yet,
and let it be labeled with $(\cY',\ell')$, where $\cY' = (Y'_1,\ldots,Y'_r)$.
If there are no edges from $Y'_{j-1}$ to $Y'_{j}$ or from $Y'_j$ to $Y'_{j+1}$, then we stop the recursion,
i.e., $d$ will be a leaf of $\mathsf{T}$.

Otherwise, let $u^*$ be the first vertex of $Y'_j$ that is adjacent to some vertex in $Y'_{j+1}$,
and let $v^*$ be the last vertex of $Y'_j$ that is adjacent to some vertex in $Y'_{j-1}$.
If $v^*\prec u^*$, then we again stop the recursion and $d$ becomes a leaf.

So now assume that $u^*\preceq v^*$, we emphasize that it is possible that $u^* = v^*$.
Let $S$ be the set of vertices $y\in Y'_j$ such that $u^*\preceq y \preceq v^*$, and let $E_1$ (resp., $E_2$) be the set of edges with one endpoint in $Y'_{j-1}$ (resp., $Y'_{j+1}$), and the other in $S$.
We also define $\widetilde{S}\subseteq S$ to be the set of the vertices adjacent to at least one edge of $E_1\cup E_2$ (see \cref{fig:claim-q1}).

\begin{figure}[t]
    \centering
   
    \begin{tikzpicture}[every node/.style={draw,circle,fill=white,inner sep=0pt,minimum size=8pt},every loop/.style={}]

\foreach \k in {0,1,2,3,4}
{
\node (a\k) at (\k*0.7,0) {};
}

\foreach \k in {0,1,2,3,4,5,6,7,8}
{
\node (b\k) at (5+\k*0.7,0) {};
}

\foreach \k in {0,1,2,3}
{
\node (c\k) at (13+\k*0.7,0) {};
}

\draw[color=orange] (a1) to [bend left] (b1);
\draw[color=orange] (a3) to [bend left] (b7);
\draw[color=orange] (a3) to [bend left] (b4);
\draw[color=orange] (a4) to [bend left] (b2);
\draw[color=orange] (a2) to [bend left] (b6);
\draw[color=green] (b1) to [bend left] (c2);
\draw[color=green] (b4) to [bend left] (c3);
\draw[color=green] (b5) to [bend left] (c0);
\draw[color=green] (b2) to [bend left] (c1);
\draw (a0) to [bend left] (b0);
\draw (a1) to [bend left] (b0);
\draw (b8) to [bend left] (c2);
\draw (b8) to [bend left] (c0);
\draw (b0) to [bend left] (b2);
\draw (b2) to [bend left] (b4);
\draw (b3) to [bend left] (b6);
\draw (b6) to [bend left] (b8);

\foreach \k in {2,5,6}
{
\node[draw=none,fill=blue,opacity=0.5] (u) at (5+\k*0.7,0) {};
}

\foreach \k in {1,2,4,5,6,7}
{
    \node[very thick,fill=none] (s) at (5+\k*0.7,0) {};
}

\node[draw=none,fill=none,label=below:\footnotesize{$u^*$}] (u) at (5.7,0) {};
\node[draw=none,fill=none,label=below:\footnotesize{$v^*$}] (v) at (9.9,0) {};

\draw [decorate,decoration = {calligraphic brace,mirror}] (-0.2,-0.5) --  (3,-0.5);
\draw [decorate,decoration = {calligraphic brace,mirror}] (4.8,-0.5) --  (10.6,-0.5);
\draw [decorate,decoration = {calligraphic brace,mirror}] (12.8,-0.5) --  (15.3,-0.5);
\node[draw=none,fill=none] (x) at (1.4,-1) {$Y'_{j-1}$};
\node[draw=none,fill=none] (x) at (7.9,-1) {$Y'_{j}$};
\node[draw=none,fill=none] (x) at (14.25,-1) {$Y'_{j+1}$};
\node[draw=none,fill=none] (e) at (6,1.6) {\textcolor{orange}{$E_1$}};
\node[draw=none,fill=none] (e) at (10,1.6) {\textcolor{green}{$E_2$}};

\end{tikzpicture}
     \caption{Consecutive sets $Y'_{j-1}$, $Y'_j$, and $Y'_{j+1}$ in the proof of \cref{claim:abab-q1}. Orange, green, and black edges denote, respectively, the edges of $E_1$, the edges of $E_2$, and the remaining edges. Thick nodes denote vertices of $\widetilde{S}$, and blue nodes denote vertices taken to $U$ for $k=3$.}
    \label{fig:claim-q1}
\end{figure}
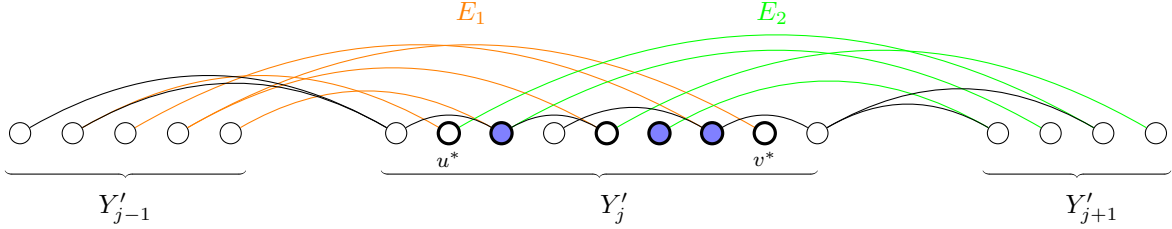

\begin{claim}\label{claim:abab-q1}
There exists a vertex $w$ which is an endpoint of some edge in $E_1\cup E_2$ and such that the set $N[w]$ intersects at least $\frac{1}{2k+4}$-fraction of the edges of $E_1\cup E_2$.
\end{claim}
\begin{claimproof}
By symmetry, assume that $|E_1|\geq |E_2|$.
Let $U$ be an independent set defined as follows.
We initialize $U=\emptyset$.
As long as $|U|< k$, we add (if exists) to $U$ the first vertex of $\widetilde{S} \setminus N[U\cup\{u^*\}]$.
Note that if at the end of the procedure $|U|<k$,
then the set $N[U\cup \{u^*\}]$ contains all vertices of $\widetilde{S}$, and thus intersects all the edges of $E_1$.

So now suppose that $|U|=k$.
Let us fix a neighbor $u'$ of $u^*$ in $Y'_{j+1}$ -- by the choice of $u^*$ such a vertex exists.
Observe that the set $N[U\cup \{u^*,u'\}]$ intersects all edges of $E_1$.
Indeed, suppose otherwise, i.e., there is an edge $xy\in E_1$ with $x\in Y'_{j-1}$, $y\in Y'_{j}$ and such that $x,y$ are nonadjacent to $U\cup \{u^*,u'\}$.
Note that $x\prec u^*\prec U \prec y \prec u'$,
and thus the set $U\cup \{u^*,u',x,y\}$ induces in $G[\cY]$ a copy of \abkab with $x\in Y_{j-1}$, $u^*,y\in Y_j$,
and $u'\in Y_{j+1}$ (see \cref{fig:claim-q1}), which contradicts \cref{claim:abab1}.

Therefore, in any case, the set $N[U\cup \{u^*,u'\}]$ intersects all the edges of $E_1$,
and thus there is a vertex $w \in U\cup \{u^*,u'\}$
such that $N[w]$ intersects at least $\frac{1}{k+2}$-fraction of edges in $E_1$.
Since $|E_1|\geq |E_2|$, we obtain that the set $N[w]$ intersects at least $\frac{1}{2k+4}$-fraction of edges in $E_1\cup E_2$. Moreover, every vertex in $U\cup \{u^*,u'\}$ is an endpoint of some edge in $E_1\cup E_2$, which completes the proof of the claim.
\end{claimproof}

Let us remark that the vertex $w$ from \cref{claim:abab-q1} can be found in polynomial time.
Now we are ready to define child nodes for $d$; the nodes will correspond to taking $w$ to the solution or not, i.e., we add two child nodes of $d$ that are labeled respectively with:
\[
(\cY' - N[w],\ell'+\wei(w)) \quad \text{and} \quad (\cY' - \{w\},\ell').
\]

Clearly, the introduced pairs satisfy the following:
\[
\alpha(\cY')=\max \bigl(  \alpha(\cY'- N[w]) +\wei(w),\alpha(\cY' - \{w\}) \bigr).
\]
This concludes the construction of $\mathsf{T}$.

Now let us bound the number of leaves.
Let $\mu=|E_1|+|E_2|$.
Let us point out that the vertices $u^*$ and $v^*$ might change, but always the sets $S,E_1,E_2$ are subsets of the corresponding sets for the parent node.
By \cref{claim:abab-q1}, in the branch corresponding to taking $w$ to the solution, the number of edges in $E_1\cup E_2$ decreases at least by $\frac{\mu}{2k+4}$, and in the branch corresponding to not taking $w$ to the solution, we remove $w$ which is an endpoint of at least one edge of $E_1\cup E_2$.
Therefore, the bound on the number of leaves in $\mathsf{T}$ is given by the recursive inequality:
\[
F(\mu)\leq F\left( \frac{2k+3}{2k+4} \; \mu \right) + F(\mu-1),
\]
which solves to $\mu^{\Oh(\log \mu)}$ (recall that $k$ is a constant).
Applying $\mu\leq n^2$, the number of leaves in $\mathsf{T}$ is bounded by $(n^2)^{\Oh(\log n^2)}=n^{\Oh(\log n)}$.

Now we are ready to define the family $\mathscr{Z}_{\cY}$.
Consider a leaf node of $\mathsf{T}$, let it be labeled with $(\cY',\ell')$, and let $\cY' = (Y'_1,\ldots,Y'_r)$.
Then one of the following holds:
\begin{enumerate}[(i)]
\item there are no edges between $Y'_{j-1}$ and $Y'_j$,
\item there are no edges between $Y'_j$ and $Y'_{j+1}$,
\item for the first vertex $u^*\in Y'_j$ that is adjacent to some vertex in $Y'_{j+1}$,
and for the last vertex $v^*\in Y'_j$ that is adjacent to some vertex in $Y'_{j-1}$,
it holds that $v^*\prec u^*$.
\end{enumerate}
In each of these cases, we add to $\mathscr{Z}_{\cY}$ the pair $(\cZ,\ell')$, where $\cZ$ is defined as follows, depending on the case:
\begin{enumerate}[(i)]
\item $\cZ = (Y'_1,\ldots,Y'_{j-1},\emptyset,Y'_j,Y'_{j+1},\ldots,Y'_r)$,
\item $\cZ = (Y'_1,\ldots,Y'_{j-1},Y'_j,\emptyset,Y'_{j+1},\ldots,Y'_r)$,
\item $\cZ = (Y'_1,\ldots,Y'_{j-1},Y'_{j,\leftarrow},Y'_{j,\rightarrow},Y'_{j+1},\ldots,Y'_r)$,
where $Y'_{j,\leftarrow}=\{x\in Y'_j \ | \ x\preceq v^*\}$ and $Y'_{j,\rightarrow}=Y'_j\setminus Y'_{j,\leftarrow}$.
\end{enumerate}
We repeat this for every leaf node of $\mathsf{T}$. 
It is straightforward to verify that  $\mathscr{Z}_{\cY}$  satisfies all required properties.

This completes the proof of the lemma in the case if $1 < j < r$.
\end{proof}

\begin{proof}[Proof of \cref{lem:refine} for $j \in \{1,r\}$.]
Let us assume that $j=1$, as the case $j=r$ is symmetric.
The proof follows the same outline as the one for the previous case.

\paragraph{Construction of $\mathscr{Y}$.}
First, let us construct an auxiliary family $\mathscr{Y}$ which this time corresponds to guessing the last (at most) $k$ vertices in $X_1$ and the first (at most) $k$ vertices in $X_r$ that belong to the solution.

Formally, for every pair of sets $I_1\subseteq X_1$, $I_r\subseteq X_r$, each of size at most $k$,
such that $I:=I_1\cup I_2$ is independent, we create a pair $(\cX^I,\ell^I)$ as follows.
Let $v_1^*$ be the first vertex of $I_1$ and let $v_r^*$ be the last vertex of $I_r$.
Then, $\cX^I$ is obtained from $\cX$ by:
\begin{enumerate}[(i)]
\item removing the vertices of $N[I]$,
\item if $|I_1|=k$, then removing $\{ x \in X_1 \ | \ x \succ v_1^*\}$, and otherwise removing all vertices from $X_1$,
\item if $|I_r|=k$, then removing $\{ x \in X_r \ | \ x \prec v_r^*\}$, and otherwise removing all vertices from $X_r$.
\end{enumerate}
We add the pair $(\cX^I,\wei(I))$ to $\mathscr{Y}$.
Similarly as in the previous case, we have $|\mathscr{Y}|= \Oh(n^{2k})$, and 
\[
\alpha(\cX)=\max_{(\cX^I,\ell^I)\in \mathscr{Y}} \alpha(\cX^I)+\ell^I.
\]

Now, we have the following analogue of \cref{claim:abab1}.

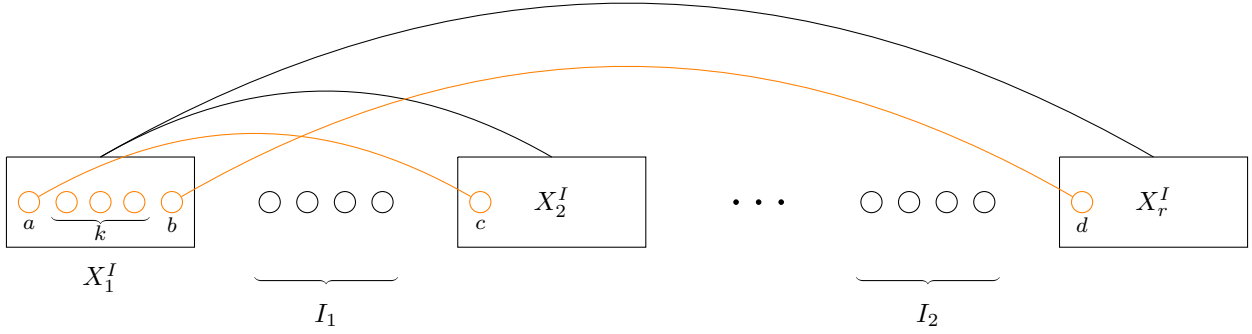
\begin{figure}[t]
    \centering
   
    \begin{tikzpicture}[every node/.style={draw,circle,fill=white,inner sep=0pt,minimum size=8pt},every loop/.style={}]
    
\foreach \k in {1,2,3,4}
{
\node (i1\k) at (\k*0.5-0.5,0) {};
\node (i2\k) at (7.5+\k*0.5,0) {};
}

\draw (-3.5,0.6)--(-1,0.6)--(-1,-0.6)--(-3.5,-0.6)--(-3.5,0.6);

\draw (2.5,0.6)--(5,0.6)--(5,-0.6)--(2.5,-0.6)--(2.5,0.6);

\draw (10.5,0.6)--(13,0.6)--(13,-0.6)--(10.5,-0.6)--(10.5,0.6);

\draw (-2.25,0.6) to [bend left] (3.75,0.6);
\draw (-2.25,0.6) to [bend left] (11.75,0.6);

\node[draw=orange,label=below:\footnotesize{$a$}] (a) at (-3.2,0) {};
\node[draw=orange,label=below:\footnotesize{$b$}] (b) at (-1.3,0) {};
\node[draw=orange,label=below:\footnotesize{$c$}] (c) at (2.8,0) {};
\node[draw=orange,label=below:\footnotesize{$d$}] (d) at (10.8,0) {};

\draw[color=orange] (a) to [bend left] (c);
\draw[color=orange] (b) to [bend left] (d);

\node[draw=orange] (e) at (-2.7,0) {};
\node[draw=orange] (e) at (-2.25,0) {};
\node[draw=orange] (e) at (-1.8,0) {};

\draw [decorate,decoration = {calligraphic brace,mirror}] (-2.9,-0.2) --  (-1.6,-0.2);
\node[draw=none,fill=none] (x) at (-2.25,-0.4) {\footnotesize{$k$}};

\node[draw=none,fill=none] (x) at (-2.25,-1) {$X^I_{1}$};
\node[draw=none,fill=none] (x) at (3.75,0) {$X^I_{2}$};
\node[draw=none,fill=none] (x) at (11.75,0) {$X^I_{r}$};

\node[draw=none,fill=none] (i) at (0.75,-1.5) {$I_1$};
\node[draw=none,fill=none] (i) at (8.75,-1.5) {$I_2$};
\draw [decorate,decoration = {calligraphic brace,mirror}] (-0.2,-1) --  (1.7,-1);
\draw [decorate,decoration = {calligraphic brace,mirror}] (7.8,-1) --  (9.7,-1);

\foreach \k in {6.2,6.5,6.8}
{
\draw[fill=black] (\k,0) circle (0.03);
}

\end{tikzpicture}
     \caption{Sets $X^I_1$, $X^I_2$, $X^I_r$, and independent sets $I_1$, $I_2$ in the proof of \cref{claim:abab2}. The orange subgraph cannot be induced, as otherwise, together with $I_1$ and $I_2$, it creates the forbidden structure.}
    \label{fig:claim-abab2}
\end{figure}

\begin{claim}\label{claim:abab2}
Let $(\cX^I,\ell^I)$, where $\cX^I= (X^I_1,\ldots,X^I_r)$, be a pair in $\mathscr{Y}$ constructed for an independent set $I$.
There is no induced copy of \akbab with edges $ac, bd$ such that $a,b \in X^I_1$ and $c \in X^I_2$, and $d \in X^I_{r}$.
\end{claim}
\begin{claimproof}
For contradiction, suppose that such a copy of \akbab exists.
Let $I'$ be the $k$-element independent set in that copy, where $a \prec I' \prec b$.
Let $I_1 = I \cap X_1$ and $I_r = I \cap X_r$.
As $a \in X^I_1$ and $d \in X^I_r$, we have that $|I_1|=|I_r|=k$.
Thus, the vertices $\{a,b,c,d\} \cup I' \cup I_1 \cup I_r$ induce a copy of \ababk in $G$ (see \cref{fig:claim-abab2}), a contradiction.
\end{claimproof}

\paragraph{Construction of $\mathscr{Z}_{\cY}$.}
Fix some $(\cY,\ell) \in \mathscr{Y}$, where $\cY = (Y_1,\ldots,Y_r)$.
Similarly as in the previous case, in time $n^{\Oh(\log n)}$ we will construct a family $\mathscr{Z}_{\cY}$ of pairs $(\cZ,\ell')$, where  $\cZ$ is a refinement of $\cY$ (and thus of $\cX$) at $1$ and $\ell' \in \mathbb{Q}_{\geq 0}$,
such that
\[
\alpha(\cY)=\max_{(\cZ,\ell') \in \mathscr{Z}_{\cY}} \alpha(\cZ)+\ell'.
\]
Recall that this will complete the proof, as the family
\[
	\mathscr{Z} := \bigcup_{(\cY,\ell) \in \mathscr{Y}} \bigcup_{(\cZ,\ell') \in \mathscr{Z}_{\cY}} (\cZ,\ell + \ell')
\]
satisfies all properties listed in the lemma.

Thus, let us show how to construct  $\mathscr{Z}_{\cY}$ for $(\cY,\ell) \in \mathscr{Y}$.
Again, we will build a recursion tree $\mathsf{T}$ whose root it labeled by $(\cY,0)$.
We proceed exhaustively as follows.
Let $d$ be a node of $\mathsf{T}$ that was not processed yet,
and let it be labeled with $(\cY',\ell')$, where $\cY' = (Y'_1,\ldots,Y'_r)$.

If there are no edges from $Y'_1$ to $Y'_{r}$ or from $Y'_1$ to $Y'_2$, then we stop the recursion,
i.e., $d$ will be a leaf of $\mathsf{T}$.
Otherwise, let $u^*$ be the first vertex of $Y'_1$ that is adjacent to some vertex in $Y'_2$,
and let $v^*$ be the last vertex of $Y'_1$ that is adjacent to some vertex in $Y'_r$.
If $v^*\prec u^*$, then we again stop the recursion and $d$ becomes a leaf.

So now assume that $u^*\preceq v^*$.
As before, let $S$ be the set of vertices $y\in Y'_1$ such that $u^*\preceq y \preceq v^*$, and let $E_1$ (resp., $E_2$) be the set of edges with one endpoint in $Y'_{2}$ (resp., $Y'_{r}$), and the other in $S$.
Finally, let $\widetilde{S}$ be the subset of $S$ that contains all the vertices adjacent to some edge in $E_1\cup E_2$ (see \cref{fig:claim-q2}).

\begin{figure}[t]
    \centering
   
    \begin{tikzpicture}[every node/.style={draw,circle,fill=white,inner sep=0pt,minimum size=8pt},every loop/.style={}]

\foreach \k in {0,1,2,3,4,5,6}
{
\node (a\k) at (\k*0.7,0) {};
}

\foreach \k in {0,1,2,3,4}
{
\node (b\k) at (6.5+\k*0.7,0) {};
}

\foreach \k in {0,1,2,3,4}
{
\node (c\k) at (13+\k*0.7,0) {};
}

\draw (a0) to [bend left] (c3);
\draw (a1) to [bend left] (c2);
\draw (a1) to [bend left] (c4);

\draw[color=orange] (a2) to [bend left] (c4);
\draw[color=orange] (a3) to [bend left] (c2);
\draw[color=orange] (a5) to [bend left] (c4);
\draw[color=orange] (a5) to [bend left] (c0);
\draw[color=orange] (a3) to [bend left] (c1);


\draw[color=green] (a2) to [bend left] (b3);
\draw[color=green] (a2) to [bend left] (b2);
\draw[color=green] (a3) to [bend left] (b4);
\draw[color=green] (a3) to [bend left] (b0);
\draw[color=green] (a5) to [bend left] (b1);
\draw[color=green] (a5) to [bend left] (b3);

\draw (a6) to [bend left] (b0);
\draw (a6) to [bend left] (b2);

\draw (a0) to [bend left] (a2);
\draw (a1) to [bend left] (a2);
\draw (a2) to [bend left] (a4);
\draw (a4) to [bend left] (a5);
\draw (a3) to [bend left] (a5);
\draw (a4) to [bend left] (a6);

\node[draw=none,fill=blue,opacity=0.5] (s) at (2*0.7,0) {};


\foreach \k in {2,3,5}
{
\node[very thick,fill=none] (s) at (\k*0.7,0) {};
}

\node[draw=none,fill=none,label=below:\footnotesize{$u^*$}] (u) at (1.4,0) {};
\node[draw=none,fill=none,label=below:\footnotesize{$v^*$}] (v) at (3.5,0) {};

\draw [decorate,decoration = {calligraphic brace,mirror}] (-0.2,-0.5) --  (4.4,-0.5);
\draw [decorate,decoration = {calligraphic brace,mirror}] (6.3,-0.5) --  (9.5,-0.5);
\draw [decorate,decoration = {calligraphic brace,mirror}] (12.8,-0.5) --  (16,-0.5);
\node[draw=none,fill=none] (x) at (2.1,-1) {$Y'_{1}$};
\node[draw=none,fill=none] (x) at (7.9,-1) {$Y'_{2}$};
\node[draw=none,fill=none] (x) at (14.4,-1) {$Y'_{r}$};

\foreach \k in {10.9,11.15,11.4}
{
\draw[fill=black] (\k,0) circle (0.03);
}

\end{tikzpicture}
     \caption{Sets $Y'_{1}$, $Y'_2$, and $Y'_{r}$ in the proof of \cref{claim:abab-q2}. Orange, green, and black edges denote, respectively, the edges of $E_1$, the edges of $E_2$, and the remaining edges. Thick nodes denote the vertices of $\widetilde{S}$, and blue ones denote the vertices of $U$ (in this example, for $k\geq 1$, we have $U=\{u^*\}$).}
    \label{fig:claim-q2}
\end{figure}
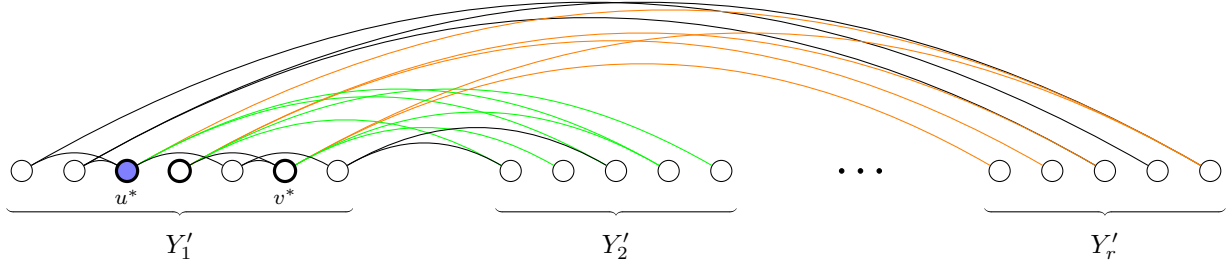

\begin{claim}\label{claim:abab-q2}
There exists a vertex $w$ which is an endpoint of some edge in $E_1\cup E_2$ and such that the set $N[w]$ intersects at least $\frac{1}{2k+4}$-fraction of the edges in $E_1\cup E_2$.
\end{claim}
\begin{claimproof}
Suppose first that $|E_1| \geq |E_2|$.
Let us define an independent set $U$ as follows. We start with $U = \emptyset$.
As long as $|U|<k$, we add (if exists) to $U$ the last vertex of $\widetilde{S} \setminus N[U\cup \{v^*\}]$.
Note that if at the end of the procedure $|U|<k$, then $N[U\cup \{v^*\}]$ contains all vertices of $\widetilde{S}$.
So assume now that $|U|=k$, and let $v'\in Y'_r$ be the neighbor of $v^*$ -- it exists by the choice of $v^*$.
We observe that the set $N[U\cup\{v^*,v'\}]$ intersects all the edges of $E_1$.
Indeed, suppose otherwise, and let $xy$ be an edge such that $x\in Y'_1$, $y\in Y'_2$ and $x,y$ are nonadjacent to $U\cup\{v^*,v'\}$.
Note that we have $x\prec U \prec v^*\prec y\prec v'$.
Then the set $U \cup\{v^*,v',x,y\}$ induces in $G[\cY]$ a copy of \akbab with $x,v^*\in Y_1$, $y\in Y_2$, and $v'\in Y_r$, which contradicts \cref{claim:abab2}.
Therefore, in both cases, the set $N[U\cup\{v^*,v'\}]$ intersects all the edges of $E_1$,
and by the pigeonhole principle, there is $w \in U\cup\{v^*,v'\}$ such that $N[w]$ intersects at least $\frac{1}{k+2}$-fraction of the edges in $E_1$.
Since $|E_1|\geq |E_2|$, the set $N[w]$ intersects at least $\frac{1}{2k+4}$-fraction of the edges in $E_1\cup E_2$, which completes the proof of the claim in this case.

So now suppose that $|E_1|<|E_2|$.
Again, we define an independent set $U$ as follows.
As long as $|U|<k$, we add (if exists) to $U$ the first vertex of $\widetilde{S}\setminus N[U\cup \{u^*\}]$.
If at the end of the procedure $|U|<k$, then $N[U\cup \{u^*\}]$ contains all vertices of $\widetilde{S}$, and thus, intersects all edges in $E_2$.
So we can assume that $|U|=k$.
Similarly as before, using \cref{claim:abab2}, we observe that the set $N[U \cup \{u^*,u'\}]$, where $u'$ is a neighbor of $u^*$ in $Y'_2$, intersects all edges in $E_2$.

Thus, again using the pigeonhole principle and the fact that $|E_2| > |E_1|$,
we conclude that some vertex in $U \cup \{u^*,u'\}$ satisfies the statement of the claim.
\end{claimproof}

Let us go back to the definition of the children of a node $d$ of the recursion tree $\mathsf{T}$.
For a vertex $w$ given by \cref{claim:abab-q2}, we create child nodes labeled, respectively, by
\[
(\cY' - N[w],\ell'+\wei(w)) \quad \text{and} \quad (\cY' - \{w\},\ell').
\]
Similarly as in the previous case we have
\[
\alpha(\cY')=\max \bigl(  \alpha(\cY'- N[w]) +\wei(w),\alpha(\cY' - \{w\}) \bigr).
\]
This concludes the construction of $\mathsf{T}$.
Again, we can bound the number of leaves $\mathsf{T}$ of this tree by $\mu^{\Oh(\log \mu)}$ which is bounded by $n^{\Oh(\log n)}$ as $\mu \leq n^2$.

Finally, we are ready to define $\mathscr{Z}_{\cY}$ in this case.
Consider a leaf node of $\mathsf{T}$, let it be labeled with $(\cY',\ell')$, and let $\cY' = (Y'_1,\ldots,Y'_r)$.
Then one of the following holds:
\begin{enumerate}[(i)]
\item there are no edges between $Y'_{1}$ and $Y'_r$,
\item there are no edges between $Y'_1$ and $Y'_{2}$,
\item for the first vertex $u^*\in Y'_1$ that is adjacent to some vertex in $Y'_{2}$, and for the last vertex $v^*\in Y'_1$ that is adjacent to some vertex in $Y'_{r}$, it holds that $v^*\prec u^*$.
\end{enumerate}
In each of these cases, we add to $\mathscr{Z}_{\cY}$ the pair $(\cZ,\ell')$, where $\cZ$ is defined as follows, depending on the case:
\begin{enumerate}[(i)]
\item $\cZ = (\emptyset,Y'_1,\ldots,Y'_r)$,
\item $\cZ = (Y'_1,\emptyset,Y'_2,\ldots,Y'_r)$,
\item $\cZ = (Y'_{1,\leftarrow},Y'_{1,\rightarrow},Y'_2,\ldots,Y'_r)$,
where $Y'_{1,\leftarrow}=\{x\in Y'_1 \ | \ x\preceq v^*\}$ and $Y'_{1,\rightarrow}=Y'_1\setminus Y'_{1,\leftarrow}$.
\end{enumerate}
We repeat this for every leaf node of $\mathsf{T}$. 
Again, it is straightforward to verify that  $\mathscr{Z}_{\cY}$  satisfies all required properties.
This completes the proof.
\end{proof}

Equipped with \cref{lem:refine}, we are ready to prove \cref{thm:abab}.

\begin{proof}[Proof of \cref{thm:abab}]
Without loss of generality assume that $k \geq 1$.
Let $(G,\wei)$ be an $n$-vertex instance of \MWIS such that $G=(V,E)$ is a \ababk-free ordered graph.
The algorithm is recursive, with the trivial base case that $n \leq 2$.
So assume that $n \geq 3$ and let $V = \{v_1,\ldots,v_n\}$ where $v_1 \prec v_2 \prec \ldots \prec v_n$.
Define:
\begin{align*}
 A= & \{v_1,\ldots,v_{\lceil\frac{n}{3}\rceil}\}, \\
 B= & \{v_{\lceil\frac{n}{3}\rceil+1},\ldots, v_{\lceil\frac{2n}{3}\rceil}\}, \\
 C= & \{v_{\lceil\frac{2n}{3}\rceil+1},\ldots,v_n\}.
\end{align*}
Note that $\cQ = (A,B,C)$ is a chain over $(G,\wei)$.
Furthermore, each its link is of size at most $\lceil\frac{n}{3}\rceil$.

We will keep refining $\cQ$, obtaining new chains.
Note that each link of such a chain originates either from $A$ or from $B$ or from $C$.
We will call such links, respectively, of type $A$, $B$, or $C$.

We recursively construct a labeled rooted tree $\mathsf{T}$ as follows.
We start with a root labeled with $(\cQ,0)$.
Now, consider a node $d$ of $\mathsf{T}$ that was not processed yet.
Let $d$ be labeled with $(\cQ',\ell')$, where $\cQ'=(Q'_1,\ldots,Q'_r)$ is a chain and $\ell \in \mathbb{Q}_{\geq 0}$.

For $\lambda \in \{A,B,C\}$, let $i_{\lambda}$ be an index of a smallest link of type $\lambda$,
where ties are resolved arbitrarily.
If $|Q'_{i_A}| = |Q'_{i_B}| = |Q'_{i_C}| = 0$, we stop the recursion and $d$ becomes a leaf of $\mathsf{T}$.
So suppose there is some $\lambda \in \{A,B,C\}$ such that $Q'_{i_\lambda} \neq \emptyset$.
We call \cref{lem:refine} for $\cQ'$ and $j=i_{\lambda}$,
obtaining the family $\mathscr{Z}$.
For each member $(\cZ,\ell)$ of this family, we create a child node for $d$ and label it with  $(\cZ,\ell)$.
This concludes the construction of $\mathsf{T}$.

Let us bound the number of leaves in the tree.
First note that each node of $\mathsf{T}$ has at most $n^{\Oh(\log n)}$ children.
Furthermore, note that if $d'$ is a child node of $d$, obtained by refining the chain at $i_{\lambda}$ for some $\lambda \in \{A,B,C\}$, then the smallest link of type $\lambda$ in the instance corresponding to $d'$
has size at most half of the size of the smallest link of type $\lambda$ in the instance corresponding to $d$.
This means that on each root-to-leaf path in $\mathsf{T}$, the same value of $\lambda$ is chosen at most $\log n$ times, and thus, the depth of $\mathsf{T}$ is at most $3\log n$.
Consequently, the number of leaves of $\mathsf{T}$ is at most $n^{\Oh(\log^2 n)}$.

Finally, by \cref{lem:refine},
we can compute $\alpha(G,\wei) = \alpha(\cQ)$ by solving the instances corresponding to leaves of $\mathsf{T}$.

So let us focus on analyzing the structure of such an instance $\cQ' = (Q'_1,\ldots,Q'_{r'})$.
Recall that since we are at a leaf of $\mathsf{T}$,
for each $\lambda \in \{A,B,C\}$ there is $i_\lambda$ such that $Q'_{i_\lambda}$ is of type $\lambda$ and $Q'_{i_\lambda}= \emptyset$. Clearly $i_A < i_B < i_C$.
Let us partition $\bigcup \cQ'$ into three sets (see \cref{fig:3-sets-abab}):
\begin{align*}
	V_1 = & \ Q'_1 \cup \ldots \cup Q'_{i_A-1} \cup Q'_{i_C+1} \ldots \cup Q'_{r'}\\
	V_2 = & \ Q'_{i_A+1} \cup \ldots \cup Q'_{i_B-1}\\	
	V_3 = & \ Q'_{i_B+1} \cup \ldots \cup Q'_{i_C-1}.
\end{align*}

Note that since $\cQ'$ is a chain, these three sets are pairwise nonadjacent.
Thus, we obtain that
\[
	\alpha(\cQ') =\alpha(G[V_1],\wei) + \alpha(G[V_2],\wei) + \alpha(G[V_3],\wei),
\]
and each summand can be computed independently by calling the algorithm recursively the corresponding induced subgraph of $G$.
Furthermore, note that $V_1$ is disjoint with $B$, $V_2$ is disjoint with $C$, and $V_3$ is disjoint with $A$.
Consequently, each of these sets is of size at most $2\lceil n/3 \rceil$.
This means that the depth of the recursion is bounded by $\Oh(\log n)$.
Summing up, the overall running time is upper-bounded by $n^{\Oh(\log^3 n)}$.
This completes the proof.
\end{proof}

\begin{figure}[t]
    \centering
   
    \input{figabab}
     \caption{Sets $Q'_1,\ldots,Q'_r$ forming a chain $Q'$ in the proof of \cref{thm:abab}. Dotted blocks denote empty sets. Blocks in orange, blue, and green represent, respectively, the sets that belong to $V_1$, $V_2$, and $V_3$.}
    \label{fig:3-sets-abab}
\end{figure}
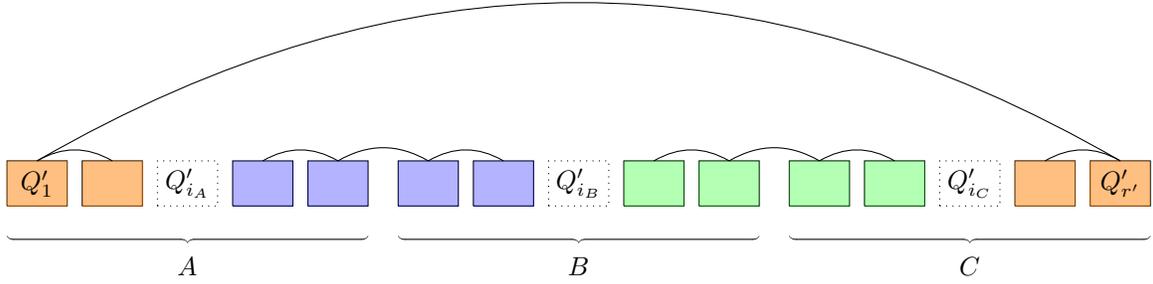

%% file: figabab.tex
\begin{tikzpicture}[every node/.style={draw,circle,fill=white,inner sep=0pt,minimum size=8pt},every loop/.style={}]

    \foreach \k in {0,1,3,4,5.2,6.2,8.2,9.2,10.4,11.4,13.4,14.4}
    {
    \draw (\k,0)--(\k+0.8,0)--(\k+0.8,0.6)--(\k,0.6)--(\k,0);
    }


    \foreach \k in {2,7.2,12.4}
    {
    \draw[dotted] (\k,0)--(\k+0.8,0)--(\k+0.8,0.6)--(\k,0.6)--(\k,0);
    }

    \foreach \k in {0,1,13.4,14.4}
    {
    \draw [draw=none,fill=orange, opacity=0.5] (\k,0) rectangle (\k+0.8,0.6);
    }

    \foreach \k in {3,4,5.2,6.2}
    {
    \draw [draw=none,fill=blue, opacity=0.3] (\k,0) rectangle (\k+0.8,0.6);
    }

    \foreach \k in {8.2,9.2,10.4,11.4}
    {
    \draw [draw=none,fill=green, opacity=0.3] (\k,0) rectangle (\k+0.8,0.6);
    }

    \node[draw=none,fill=none] (q) at (0.4,0.3) {$Q'_1$};
    \node[draw=none,fill=none] (q) at (2.4,0.3) {$Q'_{i_A}$};
    \node[draw=none,fill=none] (q) at (7.6,0.3) {$Q'_{i_B}$};
    \node[draw=none,fill=none] (q) at (12.8,0.3) {$Q'_{i_C}$};
    \node[draw=none,fill=none] (q) at (14.8,0.3) {$Q'_{r'}$};

    \foreach \k in {0,3,5.2,8.2,10.4,13.4}
    {
    \draw (\k+0.4,0.6) to [bend left] (\k+1.4,0.6);
    }
    
    \foreach \k in {4,9.2}
    {
    \draw (\k+0.4,0.6) to [bend left] (\k+1.6,0.6);
    }

    \draw (0.4,0.6) to [bend left] (14.8,0.6);


   \draw [decorate,decoration = {calligraphic brace,mirror}] (0,-0.4) --  (4.8,-0.4);
   \draw [decorate,decoration = {calligraphic brace,mirror}] (5.2,-0.4) --  (10,-0.4);
   \draw [decorate,decoration = {calligraphic brace,mirror}] (10.4,-0.4) --  (15.2,-0.4);

    \node[draw=none,fill=none] (a) at (2.4,-.8) {$A$};
    \node[draw=none,fill=none] (b) at (7.6,-.8) {$B$};
    \node[draw=none,fill=none] (c) at (12.8,-.8) {$C$};

\end{tikzpicture}

%% file: abba.tex
In this section, we prove \cref{thm:main} (3): for every fixed $k$, there is a \emph{subexponential-time} algorithm for \MWIS in \abbak-free graphs.
More precisely, we show the following result.

\begin{theorem}
   For every fixed $k$, \MWIS in $n$-vertex \abbak-free graphs can be solved in time $n^{\Oh(n^{2/3} \cdot \log^3n)}$.    
\end{theorem}
\begin{proof}
    Let $k$ be a fixed integer and let $(G,\wei)$ be an instance of \MWIS, where $G$ is \abbak-free and has $n$ vertices.  
    Note that we do not update the value of $n$ throughout the run of the algorithm.
    Let $\tau = n^{1/3}$.

\paragraph{Step 1. Branching on high-degree vertices.}
    First, we perform a standard branching on high-degree vertices.
    If there is a vertex $v$ of degree at least $\tau$, we branch on including or not including $v$ in the solution.
    The number of instances created in this phase is given by the recursion
    \[
        F(n) \leq F(n-1) + F(n - \tau),    
    \]
    which is solved by $n^{\Oh(n/\tau)}=2^{\Oh(n^{2/3} \log n)}$.
    Each instance $(G',\wei)$ created in this phase has maximum degree at most $\tau$.
    Furthermore, it is an induced subgraph of $G$, so in particular $G'$ is \abbak-free.

\paragraph{Step 2. Partitioning $V(G')$ into segments.}    
    Consider one instance $(G',\wei)$ created in Step~1 and let $n' = |V(G')|$.
    If $n'=1$, the problem is obvious, so let us assume that $n' \geq 2$.
    We can also assume that $G'$ is connected; otherwise, we can call the algorithm for each component of $G'$ separately.

    In the next Claim, we will define some sequence $Z_1,\ldots,Z_t$ of subsets of $V(G')$.
    For $i \in [t]$, we introduce a notation $Z_{\leq i} := \bigcup_{i' \leq i} Z_{i'}$.
    We define $Z_{< i} $ and $Z_{> i}$ analogously.

    \begin{claim}\label{clm:partition}
    There is a sequence $(x_i,y_i,Z_i)_{i=1}^t$, for some $t \geq 1$, such that each $x_i,y_i \in V(G')$ and $Z_i \subseteq V(G')$,
    that satisfies the following properties.
    \begin{enumerate}[(P1)]              
        \item for each $i\in [t]$, the set $Z_{\leq i}$ is a prefix of $V(G')$;
        \item for each $i\in [t]$, we have $x_iy_i \in E(G')$,
        \item for each $i\in [t]$, it holds that $x_i \preceq Z_i \preceq y_i$ and $y_i \in Z_i$,                
        \item for each $i\in [t]$, the only vertices of $\bigcup_{i' \leq i} Z_{i'}$ that might have edges to $V(G') \setminus Z_{\leq i}$ are in $Z_i$,
        \item sets $Z_1,Z_2,\ldots,Z_t$ form a partition of $V(G')$.  
    \end{enumerate}    
    \end{claim}

    \begin{claimproof}
        We construct the sequence iteratively.
        Let $x_1$ be the first vertex of $G'$ and let $y_1$ be the last neighbor of $x_1$;
        it exists since $G'$ is connected and has at least two vertices.
        Let $Z_1$ denote the set of vertices $v$ such that $x_1 \preceq v \preceq y_1$.
        Clearly, properties (P1)-(P4) are met.
    
        Now, suppose we have defined $(x_1,y_1,Z_1),\ldots,(x_{i},y_{i},Z_{i})$ for some $i \geq 1$.
        If $y_{i}$ is the last vertex of $G'$, we finish the construction (i.e., $t=i$);
        note that property (P5) is met in this case.
        So, suppose otherwise. We aim to define $(x_{i+1},y_{i+1},Z_{i+1})$.
    
        Let $U = V(G') \setminus Z_{\leq i}$. By property (P1), $U$ consists of all vertices that are after $Z_i$.
        Since $G'$ is connected, there must be an edge from $Z_{\leq i}$ to $U$.
        Let $x_{i+1}y_{i+1} \in E(G')$ be an edge, so that:
        \begin{itemize}
            \item $x_{i+1} \in Z_{\leq i}$ and $y_{i+1} \in U$,
            \item $y_{i+1}$ is furthest possible.
        \end{itemize}
        By (P4), we have $x_{i+1} \in Z_i$.
        We define $Z_{i+1}$ to be the set of $v$ such that $Z_i \prec v \preceq y_{i+1}$.
        Properties (P1), (P2), (P3) are clearly satisfied.
        Property (P4) follows from the choice of $x_{i+1}y_{i+1}$.
        This completes the $(i+1)$-th step of the construction.
    \end{claimproof}

    Let $(x_i,y_i,Z_i)_{i=1}^t$ be as in \cref{clm:partition}.
    Note that property (P4) implies that
    \begin{itemize}
        \item[$(\star)$] edges between sets $Z_i,Z_j$, for $i<j$, might exist only if $j = i+1$. 
    \end{itemize}    
    We also denote $\mathsf{Roof} = \bigcup_{i=1}^t \{x_i,y_i\}$.

 \paragraph{Step 3. Picking independent sets inside $Z_i$.}            

    For each $i \in [t]$, we proceed as follows.
    Processing vertices from left to right, we greedily pick pairwise nonadjacent vertices in $Z_i \setminus N[\mathsf{Roof}]$.
    We finish once we pick $k$ vertices or when there are no more vertices in  $Z_i \setminus N[\mathsf{Roof}]$ that are nonadjacent to the selected ones.
    Let $\mathsf{First}_i$ denote the set of at most $k$ vertices picked in this step.
    
    Next, processing vertices from right to left, we greedily pick pairwise nonadjacent vertices in $Z_i \setminus(N[\mathsf{Roof}] \cup N[\mathsf{First}_i])$.
    Again, we stop once we pick $k$ vertices or once there are no more possible vertices to choose.
    Let $\mathsf{Last}_i$ denote the set of at most $k$ vertices picked in this step.
    Summing up, $\mathsf{First}_i \cup \mathsf{Last}_i$ is an independent set of size at most $2k$, and there are no edges between this set and $\{x_i,y_i\}$. 

    We also observe the following.

    \begin{claim}\label{clm:inZiindependent}
        For each $i \in [t]$, the set $Z_i \setminus N[ \{x_i,y_i\}  \cup \mathsf{First}_i \cup \mathsf{Last}_i]$ is independent.
    \end{claim}       
    \begin{claimproof}
         For contradiction, suppose there an edge $uv$ in $Z_i \setminus N[ \{x_i,y_i\}  \cup \mathsf{First}_i \cup \mathsf{Last}_i]$.
         Since  $Z_i \setminus N[ \{x_i,y_i\}  \cup \mathsf{First}_i \cup \mathsf{Last}_i]$ is non-empty, we have $|\mathsf{First}_i| = |\mathsf{Last}_i| = k$.
         Moreover, by the choice of $\mathsf{First}_i$ and $ \mathsf{Last}_i$, for every vertex $w\in Z_i \setminus N[ \{x_i,y_i\}  \cup \mathsf{First}_i \cup \mathsf{Last}_i]$, it holds $ \mathsf{First}_i \prec w \prec \mathsf{Last}_i$.
         Consequently, vertices $\{x_i,u,v,y_i\} \cup \mathsf{First}_i \cup \mathsf{Last}_i$ induce a copy of \abbak in $G'$, a contradiction.
    \end{claimproof}   

    \paragraph{Step 4. Some more branching to split the instance.}
    For $i \in [t]$, we say that $Z_i$ is \emph{large} if it has at least $n'/\tau \leq n^{2/3}$ vertices;
    otherwise $Z_i$ is \emph{small}.
    Let $\mathsf{Left}$ (resp.,  $\mathsf{Right}$) consist of the first (resp.,  last) $\lfloor n'/3 \rfloor$ vertices of $G'$,
    and let $\mathsf{Middle} = V(G') \setminus (\mathsf{Left} \cup \mathsf{Right})$.

    We consider two cases.

    
    
    

    \paragraph{Case 1: There is a small set intersecting $\mathsf{Middle}$.}    
        Let $i \in [t]$ be such that $Z_i$ is small and $Z_i \cap \mathsf{Middle}\neq\emptyset$.
        By $(\star)$, there are no edges between $Z_{<i}$ and $Z_{>i}$.
        Furthermore, $|Z_{<i}| \leq \frac{2}{3}n'$ and $|Z_{>i}| \leq \frac{2}{3}n'$.

        We exhaustively guess the intersection of a fixed optimum solution with $Z_i$;
        in each branch we remove $Z_i$ and all neighbors of vertices that are chosen to be in the solution.

        This breaks the instance into two independent subinstances, one contained in $Z_{<i}$ and another contained in $Z_{>i}$,
        each of size at most $\frac{2}{3}n'$. Thus, on each of these subinstances, we can call the algorithm recursively.

        The running time bound in this case is given by the recursive inequality
        \[
            T(n') \leq 2^{\Oh(n'/\tau)} T(2n'/3),        
        \]
        which is solved by 
        \[
            T(n') \leq 2^{\Oh(n' \log n' / \tau)} = 2^{\Oh( n^{2/3} \log n)}.
        \]
        
    \paragraph{Case 2: All small sets (if any) are contained in $\mathsf{Left} \cup \mathsf{Right}$.}
        Let $\ell$ be the largest value of $i$ such that $Z_i$ is small and $Z_i \subseteq \mathsf{Left}$, if such a value exists.
        Similarly, let $r$ be the smallest value of $i$ such that $Z_i$ is small and $Z_i \subseteq \mathsf{Right}$, if such a value exists.
        In what follows, let us assume that both $\ell$ and $r$ are as defined; the other cases are analogous (and slightly simpler).
        We proceed similarly as in Case 1.
        
        Note that removing $Z_\ell$ and $Z_r$ splits the graph into three parts: $Z_{< \ell}$, $\bigcup_{i = \ell+1}^{r-1} Z_{i}$, and $Z_{>r}$.
        Again, there are no edges between distinct parts.      
                
        The first and the last part are multiplicatively smaller than the original instance -- each has at most $n'/3$ vertices.
        On the other hand, the middle part can be arbitrarily large. In fact, it might contain almost all vertices of $G'$.
        However, by the choice of $\ell$ and $r$, we know that all sets $Z_i$ that end up in the middle part are large.
        Let $t' = (r-1) - (\ell+1) +1$, i.e., the number of sets $Z_i$ that belong to the middle part.
        Since $n' \geq \sum_{i=\ell+1}^{r-1} |Z_i| \geq t'n'/\tau$, we conclude that $t' \leq \tau$.

        Let $S = Z_\ell \cup Z_r \cup N_{G'}\left[\bigcup_{i=\ell+1}^{r-1} \left(\{x_i,y_i\} \cup \mathsf{First}_i \cup \mathsf{Last}_i \right) \right]$.
        Note that $|S| \leq 2 \cdot n/\tau + (2+2k)\tau \cdot \tau = \Oh( n^{2/3} )$.
        Again, we exhaustively guess the intersection of a fixed optimum solution with $S$.
        This results in at most $2^{|S|} = 2^{\Oh( n^{2/3} )}$ branches, in each of which we remove $S$ and the neighbors of vertices chosen to be included in the solution.

        Now, we obtain three instances that can be solved independently: two of them --- one contained in $\mathsf{Left}$ and the other contained in $\mathsf{Right}$ --- are small, i.e., have at most $2n'/3$ vertices.
        Let $G''$ be the middle instance.
        By \cref{clm:inZiindependent} and $(\star)$, $G''$ consists of a number of independent sets,
        where only possible edges between distinct sets might appear if sets are consecutive.
        This means that $G''$ is bipartite, and thus \MWIS on $G''$ can be solved in polynomial time.

        Summing up, the running time bound in this case is given by recursive inequality
        \[
            T(n') \leq 2^{\Oh(n^{2/3})} \left ( 2 \; T(2n'/3) + n'^{\Oh(1)} \right) \leq 2^{\Oh(n^{2/3})}  T(2/3 \cdot n'),        
        \]
        which is solved by 
        \[
            T(n') \leq 2^{\Oh(n^{2/3} \log n')} = 2^{\Oh( n^{2/3} \log n)}.
        \]
        This completes the proof.
\end{proof}

%% file: hardness-reduction-3sat.tex
In this section, we explore the \NP-hardness of \MIS in hereditary classes of ordered graphs. 
We utilize four types of reductions, each presented in a dedicated section.
The last section summarizes the results and shows how they combine to prove the hardness part of \cref{thm:main}.

\subsection{Direct reduction from \sat{3}}

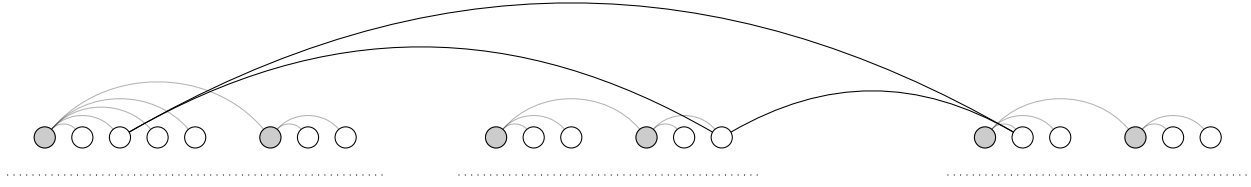
\begin{figure}[t]	
		\centering
		\input{fig3sat}
		\caption{The graph constructed in the proof of \cref{3sat}. Underscored segments indicate blocks, with the two literal vertices indicated by shading. Gray edges, contained within blocks, encode logical dependencies. Meanwhile the black ones join all occurrence vertices corresponding to some clause.}
		\label{fig:3sat}
\end{figure}

\begin{theorem}\label{3sat}
	\MIS is \NP-hard when restricted to \bad-free ordered graphs.
\end{theorem}
\begin{proof}
	We reduce from the \sat{3} problem.
	Consider an instance $\Phi$ of \sat{3} with $n$ variables and $m$ clauses.
	Denote by $\cX$ and $\cC$, respectively, the set of variables and the set of clauses of $\Phi$.
	We can assume that each clause contains three distinct literals.
	We create an ordered graph $G$ in the following way -- see also \cref{fig:3sat}.
	
	For each variable $x$, we create \emph{literal} vertices $v(x)$ and $v(\neg x)$ representing, respectively, literals $x$ and $\neg x$.
	Next, for each clause $c$ and each literal $\ell$ occurring in $c$, we introduce an \emph{occurrence} vertex $v(\ell,c)$.
	
	Now let us define the edges of $G$ -- we split them into three types.
	First, for every variable $x \in \cX$, we add the edge $v(x)v(\neg x)$.
	Next, for every literal $\ell$ and every clause $c$ containing the literal $\neg \ell$, we add the edge $v(\ell)v(\neg\ell,c)$.
	Finally, for any clause $c$ and any two literals $\ell,\ell'$ in $c$, we add the edge $v(\ell,c)v(\ell',c)$.
	
	We order the vertices of $G$ as follows.
	For each variable $x\in \cX$, we create a block consisting of $v(x)$, $v(\neg x)$, and the vertices $v(x,c)$ or $v(\neg x,c)$ for every $c\in \cC$ that contains $x$.
	The vertices of the block for $x$ are ordered so that the first vertex is $v(x)$, followed by all the vertices $v(\neg x,c)$ for $c \in \cC$, followed by $v(\neg x)$, which is in turn followed by all $v(x,c)$'s for $c\in \cC$.
	The relative order of the vertices $v(\neg x,c)$ (resp., $v(x,c)$) is arbitrary.
	Finally, the blocks corresponding to distinct variables are put one after another in arbitrary order.
	This completes the construction of $G$, and it is clearly computable in polynomial time.
	
	Let us verify that $G$ has an independent set of size $n+m$ if and only if $\Phi$ is satisfiable.
	
	\begin{claim}
		If $\Phi$ is satisfiable, then $\alpha(G) \geq n+m$.
	\end{claim}
	
	\begin{claimproof}
		Suppose $\Phi$ is satisfiable and let $\vphi$ be a satisfying truth assignment.		
		We construct an independent set $I$ of size $n+m$ as follows.
		Whenever $\vphi(x)=\true$, we include $v(x)$ in the set $I$; otherwise, we include $v(\neg x)$.
		Since $\vphi$ satisfies $\Phi$, each clause contains at least one literal set by $\vphi$ to $\true$; we fix one such literal $\ell$ arbitrarily for each clause $c\in \cC$, and include the corresponding occurrence vertex $v(\ell,c)$ in $I$.
		Since such vertices are not adjacent to each other, nor to the literal vertices included previously, $I$ is indeed an independent set.
		Moreover, as $I$ includes one vertex for each variable and one vertex for every clause, $|I|=n+m$, as desired.
	\end{claimproof}
	
	\begin{claim}
		If $\alpha(G) \geq n+m$, then $\Phi$ is satisfiable.
	\end{claim}
	
	\begin{claimproof}
		First, let us observe that the vertices of $G$ can be partitioned into $n+m$ cliques, namely $n$ copies of $K_2$ of the form $\{v(x), v(\neg x)\}$ corresponding to the variables of $\Phi$, and $m$ triangles of the form $\{v(\ell_1, c), v(\ell_2,c), v(\ell_3, c)\}$, corresponding to the clauses of $\Phi$.
		Therefore, if $\alpha(G) \geq n+m$ and $I$ is an independent set of size $n+m$ in $G$, then $I$ has to contain precisely one vertex from each of these cliques.
		We define a truth assignment $\vphi$ satisfying $\Phi$ as follows.
		For each variable $x\in \cX$, we set $\vphi(x)=\true$ if $v(x) \in S$, otherwise we set $\vphi(x)=\false$. 
		Note that for every clause $c\in \cC$ and every literal $\ell$ in $c$, the occurrence vertex $v(\ell,c)$ is adjacent to the vertex $v(\neg \ell)$.
		Furthermore, $v(\neg \ell)\notin I$ if and only if $\vphi(\ell)=\true$,
		and thus the vertex $v(\ell,c)$ can be included in $I$ only if $\vphi(\ell)=\true$.
		Therefore, since for every clause $c\in \cC$, there must be a literal $\ell$ such that $v(\ell,c)\in I$, the assignment $\vphi$ satisfies $\Phi$.
        This completes the proof of the claim.
	\end{claimproof}

It only remains to verify that $G$ is \bad-free.
	\begin{claim}
		$G$ does not contain \bad as an induced subgraph.
	\end{claim}
	\begin{claimproof}
		Suppose $G$ contains an induced \bad; let $a, b, c, d$ be its consecutive vertices. 
		Consider the vertex $a$. Since $a$ is the leftmost vertex of the copy of \bad, only its right neighbors are relevant.
		First suppose that $a$ is a literal vertex.
		Then, all right neighbors of $a$ form a segment.
		In particular, there is no vertex between $b$ and $d$ which is nonadjacent to $a$, which contradicts the fact that $c$ is between $b$ and $d$, and is nonadjacent to $a$.
		So now suppose that $a$ is an occurrence vertex, say $a=v(\ell_1,c)$, for a literal $\ell_1$ and a clause $c$.
		Note that the only right neighbors of $a$ can be the vertices $v(\ell_2,c)$ and $v(\ell_3,c)$, where $c=\ell_1 \lor \ell_2 \lor \ell_3$.
		But there is an edge $v(\ell_2,c)v(\ell_3,c)$ in $G$, and $a$ has two right neighbors $b,d$ that are nonadjacent to each other, which is a contradiction.
	\end{claimproof}
	
	This completes the proof of the theorem.
\end{proof}

%% file: fig3sat.tex
\begin{tikzpicture}[
        every node/.style={draw, circle, inner sep=0, outer sep = 0, minimum size=8pt}, baseline= -0.5mm]

            \draw[dotted] (0,0) -- (5,0);
            \draw[dotted] (6,0) -- (10,0);
            \draw[dotted] (12.5, 0) -- (16.5,0);
            
            \node[fill, opacity=.2] (s') at (.5, .5) {};
            \node (s) at (.5, .5) {};
            \node (s1) at (1, .5) {};
            \node (s2) at (1.5, .5) {};
            \node (s3) at (2, .5) {};
            \node (s4) at (2.5, .5) {};

            \node[fill, opacity=.2] (-s') at (3.5, .5) {};
            \node (-s) at (3.5, .5) {};
            \node (-s1) at (4, .5) {};
            \node (-s2) at (4.5, .5) {};

            \node[fill, opacity=.2] (p') at (6.5, .5) {};
            \node (p) at (6.5, .5) {};
            \node (p1) at (7, .5) {};
            \node (p2) at (7.5, .5) {};

            \node[fill, opacity=.2] (-p') at (8.5, .5) {};
            \node (-p) at (8.5, .5) {};
            \node (-p1) at (9, .5) {};
            \node (-p2) at (9.5, .5) {};

            \node[fill, opacity=.2] (q') at (13, .5) {};
            \node (q) at (13, .5) {};
            \node (q1) at (13.5, .5) {};
            \node (q2) at (14, .5) {};

            \node[fill, opacity=.2] (-q') at (15, .5) {};
            \node (-q) at (15, .5) {};
            \node (-q1) at (15.5, .5) {};
            \node (-q2) at (16, .5) {};

            \path[draw, opacity=.3] (s) to[bend left = 50] (s1);
            \path[draw, opacity=.3] (s) to[bend left = 50] (s2);
            \path[draw, opacity=.3] (s) to[bend left = 50] (s3);
            \path[draw, opacity=.3] (s) to[bend left = 50] (s4);
            \path[draw, opacity=.3] (s) to[bend left = 50] (-s);
            \path[draw, opacity=.3] (-s) to[bend left = 50] (-s1);
            \path[draw, opacity=.3] (-s) to[bend left = 50] (-s2);

            \path[draw, opacity=.3] (p) to[bend left = 50] (p1);
            \path[draw, opacity=.3] (p) to[bend left = 50] (p2);
            \path[draw, opacity=.3] (p) to[bend left = 50] (-p);
            \path[draw, opacity=.3] (-p) to[bend left = 50] (-p1);
            \path[draw, opacity=.3] (-p) to[bend left = 50] (-p2);

            \path[draw, opacity=.3] (q) to[bend left = 50] (q1);
            \path[draw, opacity=.3] (q) to[bend left = 50] (q2);
            \path[draw, opacity=.3] (q) to[bend left = 50] (-q);
            \path[draw, opacity=.3] (-q) to[bend left = 50] (-q1);
            \path[draw, opacity=.3] (-q) to[bend left = 50] (-q2);

            \path[draw] (s2) to[bend left=30] (-p2);
            \path[draw] (s2) to[bend left=30] (q1);
            \path[draw] (-p2) to[bend left=30] (q1);

            
		\end{tikzpicture}

%% file: hardness-2-subdivision.tex
\subsection{2-subdivision of every edge}
In this section, we utilize Poljak's subdivision trick, mentioned already in the Introduction.
A \emph{2-subdivision of an edge $xy$} in an unordered graph $G$ is the operation of replacing $xy$ with a four-vertex path with endpoints $x$ and $y$ (the internal vertices are new).
If $G'$ is a graph obtained from $G$ by performing a 2-subdivision on a single edge, then $\alpha(G')=\alpha(G)+1$~\cite{Po74}.
Furthermore, if $G\tsd$ is obtained from $G$ by performing a 2-subdivision on each edge of $G$,
then $\alpha(G\tsd)=\alpha(G)+|E(G)|$.
We call $G\tsd$ a \emph{2-subdivision of $G$}.

\paragraph{Common outline of all reductions.}
All reductions in this section have the same outline.
We reduce from the \MIS problem in general unordered graphs.
We start with an arbitrary instance $(G, k)$ and fix an arbitrary ordering of $G$, making it an ordered graph.
Then, we consider the 2-subdivision $G\tsd$ of $G$ and equip it with an order defined in a specific way, depending on the reduction.
By the observation above, we have that $\alpha(G) \geq k$ if and only if $\alpha(G\tsd) \geq k+|E(G)|$.

In order to complete the reduction, it suffices to define an ordering of $G\tsd$ so that it belongs to the particular class.
In all reductions, the relative order of the vertices of $V(G\tsd) \cap V(G)$ (called \emph{core} vertices) is the same as in $G$.
Thus, in the proofs, we will only discuss the placement of the remaining vertices (called \emph{dummy} vertices).

Before we proceed to the reductions, let us introduce some terminology that will allow us to distinguish between the two dummy vertices corresponding to the same edge of $G$.
Let $e=xy$ be an edge of $G$ where $x \prec y$; note that their relative order is the same in $G$ and in $G\tsd$.
In the construction of $G\tsd$, we replace $e$ with the path $x-\ell_e-r_e-y$.
Now, the dummy vertex $\ell_e$ is called the \emph{left dummy} corresponding to $e$, while $r_e$ is called the \emph{right dummy} corresponding to $e$.

Let us also remark that in all reductions in this section we actually exclude certain graphs $H$ as \emph{subgraphs}, not only as induced subgraphs. Note that this is a stronger property, and thus the reductions are stronger as well. The graphs that exclude a graph $H$ as a subgraph is referred to as \emph{$H$-subgraph-free}.
We also write \emph{($H_1,H_2,\ldots,H_k$)-subgraph-free} to denote the class of ordered graphs that are $H_i$-subgraph-free for every $i\in \{1,2,\ldots,k\}$.

\begin{theorem}\label{xxxlrlrlr}
	\MIS is \NP-hard when restricted to (\adb, \abd, \abnce, \abncde, \abcd, \acbd, \adcb, \aenbcd)-subgraph-free ordered graphs.
\end{theorem}
\begin{proof}
	Let us define an appropriate ordering of the vertices of $G\tsd$, see also \cref{fig:xxxlrlrlr}.
	All core vertices precede all dummy vertices, whereas dummy vertices follow the lexicographic order of their respective edges, with left dummy being immediately followed by the right dummy corresponding to the same edge.
	More formally, we stipulate that $\ell_e\prec r_e$ for every edge $e$ of $G$ and $r_e\prec \ell_f$ whenever $e\prec f$.
	
		We now demonstrate that $G\tsd$ does not contain any of the listed graphs as an ordered subgraph.
	
	\begin{claim}
		$G\tsd$ is (\adb, \acbd, \adcb)-subgraph-free.
	\end{claim}
	\begin{claimproof}
Suppose that $G\tsd$ contains a copy of one of the graphs listed above as a subgraph, and let $a, b, c$ be the consecutive vertices of \acb in that copy.
		
Since $c$ has at least two left neighbors, it must be a right dummy $r_{uv}$ of some edge $uv$ of $G$, with $u\prec v$.
		This means that its two (left) neighbors are precisely $v$ and $\ell_{uv}$, with $v\prec \ell_{uv}$.
		Therefore the vertices of \acb are $a=v, b=\ell_{uv}, c=r_{uv}$.
		In $G\tsd$, however, there is no vertex in between $\ell_{uv}$ and $r_{uv}$, so it cannot contain \adb as a subgraph.
		Moreover, the only neighbor of $\ell_{uv}$ other than $r_{uv}$ is the vertex $u\prec v$ -- thus, there can be no subgraph \acbd or \adcb in $G\tsd$.
	\end{claimproof}
	
	\begin{claim}
		$G\tsd$ is (\abd, \abnce, \abcd, \abncde)-subgraph-free.
	\end{claim}
	\begin{claimproof}
		Suppose $G\tsd$ contains a copy of one of the listed subgraphs, and let $a\prec b\preceq c \prec d \prec e$ be the consecutive vertices of that copy.
		Note that $b=c$ precisely when the subgraph has four vertices.
		
		Since $b$ has a left neighbor, it must be a dummy vertex.
		Therefore any vertex that follows $b$, i.e., $c, d, e$, is a dummy as well.
		However, the only edges between two dummy vertices in $G\tsd$ connect vertices that immediately follow one another.
		Thus $G\tsd$ cannot contain \abd nor \abnce as a subgraph.
		
		Moreover, if $G\tsd$ contains either \abcd or \abncde, then $d$ is adjacent to a dummy vertex on either side.
		But $G\tsd$ has no vertices with that property, a contradiction.
	\end{claimproof}
	
	\begin{claim}
		$G\tsd$ is \aenbcd-subgraph-free.
	\end{claim}
	\begin{claimproof}
		Suppose that $G\tsd$ contains a copy of \aenbcd, and let $a, b, c$ be the consecutive vertices of \abc in that copy.
		
		The vertex $b$, having both a left and a right neighbor, must be a left dummy of some edge $uv$ of $G$ with $u\prec v$.
		Consequently, its left neighbor is $u$.
		Consider any vertex $u' \prec u$.
		As it precedes a core vertex, it must itself be a core vertex.
		For this reason, the only neighbors of $u'$ in $G\tsd$ are dummy vertices corresponding to the edges of $G$ adjacent to $u'$.
		As $u' \prec u$, all of these edges precede $uv$ in lexicographic order.
		Thus all the neighbors of $u'$ in $G\tsd$ precede $\ell_{uv}$, meaning a \aenbcd subgraph cannot occur.
	\end{claimproof}
	
	This completes the proof.	
\end{proof}

	\begin{figure}[t]
		\centering
		\input{figxxxlrlrlr}
		\caption{The ordering of $G\tsd$ in the proof of \cref{xxxlrlrlr}. In all figures in this section we use the following convention. Gray box represents segment occupied by the core vertices. The dashed edge is present in $G$, but is replaced by a four-vertex path in $G\tsd$. This path is depicted in black, with some other such paths in grey. The dummy vertices fill segments represented by white boxes. Within any white box, they are ordered according to the lexicographic order of their corresponding edges of $G$. Flat edges represent connections between pairs of consecutive vertices.}
        \label{fig:xxxlrlrlr}
	\end{figure}
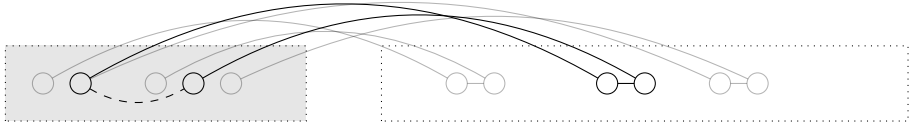

\begin{theorem}\label{xxxrlrlrl}
	\MIS is \NP-hard when restricted to (\adb, \acbd, \aenbdc, \abd, \abnce, \abcd, \abncde)-subgraph-free ordered graphs.
\end{theorem}
\begin{proof}
	Let us define an appropriate ordering of the vertices of $G\tsd$, see also \cref{fig:xxxrlrlrl}.
	The order of vertices is as follows: all core vertices in unchanged order precede all dummy vertices, whereas dummy vertices follow the lexicographic order of their respective edges, with right dummy being immediately followed by the left dummy of the same edge.
	More formally, we stipulate that $r_e\prec \ell_e$ for each edge $e$ and $\ell_e \prec r_f$ whenever $e\prec f$.
	
	Thus we have obtained an ordered graph $G\tsd$.
		We now prove it does not contain any of the subgraphs listed in the theorem.
	
	\begin{claim}
		$G\tsd$ is (\adb, \acbd, \aenbdc)-subgraph-free.
	\end{claim}
	\begin{claimproof}
		Suppose $G\tsd$ contains a copy of one of the graphs listed in the claim and let $a, b, c$ be the consecutive vertices of \acb in that copy.
		
		Since $c$ has at least two left neighbors, it must be a left dummy $\ell_{uv}$ of some edge $uv$ of $G$, with $u\prec v$.
		This means that its two (left) neighbors are precisely $u$ and $r_{uv}$, with $u \prec r_{uv}$.
		Therefore, the vertices of \acb are $a=u, b=r_{uv}, c=\ell_{uv}$.
		
		But in $G\tsd$ there is no vertex in between $r_{uv}$ and $\ell_{uv}$, so it cannot contain \adb as a subgraph.
		Moreover, the only neighbor of $r_{uv}$ other than $\ell_{uv}$ is the vertex $v\prec r_{uv}$ -- thus, there can be no subgraph \acbd.
		
		Lastly, take any vertex $u' \prec u$.
		As it precedes a core vertex, it must itself be a core vertex.
		The only neighbors of $u'$ in $G\tsd$ are, therefore, dummy vertices corresponding to its edges in $G$.
		As $u' \prec u$, all of these edges precede $uv$ in lexicographic order.
		Thus all the neighbors of $u'$ in $G\tsd$ precede $\ell_{uv}$, meaning a \aenbdc subgraph cannot occur.
	\end{claimproof}
	
	\begin{claim}
		$G\tsd$ is (\abd, \abnce, \abcd, \abncde)-subgraph-free.
	\end{claim}
	\begin{claimproof}
		Suppose $G\tsd$ contains a copy of one of the listed subgraphs, and let $a\prec b\preceq c \prec d \prec e$ be the consecutive vertices of that copy.
		Note that $b=c$ precisely when the subgraph has four vertices.
		
		Since $b$ has a left neighbor, it must be a dummy vertex.
		Therefore any vertex that follows $b$, i.e., $c, d, e$, is a dummy as well.
		However, the only edges between two dummy vertices in $G\tsd$ connect vertices that immediately follow one another.
		Thus $G\tsd$ cannot contain \abd nor \abnce as a subgraph.
		
		Moreover, if $G\tsd$ contains either \abcd or \abncde, then $d$ is adjacent to a dummy vertex on either side.
		But $G\tsd$ has no vertices with that property, a contradiction.
	\end{claimproof}

	This completes the proof.
\end{proof}

	\begin{figure}[t]	
		\centering		
	\input{figxxxrlrlrl}
		\caption{The ordering of $G\tsd$ in the proof of \cref{xxxrlrlrl}.}
		\label{fig:xxxrlrlrl}
	\end{figure}
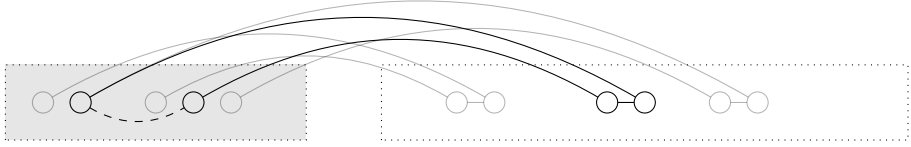

\begin{theorem}\label{rrrxxxlll}
	\MIS is \NP-hard when restricted to (\aabbcc, \badc)-subgraph-free ordered graphs.
\end{theorem}
\begin{proof}
	Let us define an appropriate ordering of the vertices of $G\tsd$, see also \cref{fig:rrrxxxlll}.
	The order of vertices is as follows: all right dummies precede all core vertices, which in turn precede all left dummies. 
	The relative order of the dummy vertices can be arbitrary.
	
	Thus, we have obtained an ordered graph $G\tsd$. 	
	We proceed by showing it does not contain subgraphs \aabbcc or \badc.
	
	\begin{claim}
		$G\tsd$ is \aabbcc-subgraph-free.
	\end{claim}
	\begin{claimproof}
		Suppose that $G\tsd$ contains a copy of \aabbcc as a subgraph, and let $a, b, c, d, e, f$ be its consecutive vertices.
		
		The vertex $b$ has a left neighbor, so it must be either a core vertex or a left dummy.
		Thus, $c$, must also be a core vertex or a left dummy, because $b \prec c$.
		Then $d$ has either a core vertex or a left dummy as a left neighbor, meaning it could only be a left dummy.
		But in that case, both $e, f$ are left dummies, precluding an edge between them.
		This is a contradiction.
	\end{claimproof}
	
	\begin{claim}
		$G\tsd$ is \badc-subgraph-free.
	\end{claim}
	\begin{claimproof}
		Suppose that $G\tsd$ contains a copy of \badc as a subgraph, and let $a, b, c, d$ be its consecutive vertices.
		
		The vertex $a$, having at least two right neighbors, is a right dummy $r_{uv}$ for some $u\prec v$.
		Its neighbors are, therefore, only $v$ and $\ell_{uv}$, with $v \prec \ell_{uv}$, so $b=v$ and $d=\ell_{uv}$.
		Furthermore, the neighbors of $d=\ell_{uv}$ are precisely $r_{uv}$ and $u$, and thus $c=u$.
		This, however, is a contradiction, as $u\prec v$ by definition, whereas $b=v \prec u=c$ in \badc.
	\end{claimproof}

	This completes the proof.
\end{proof}

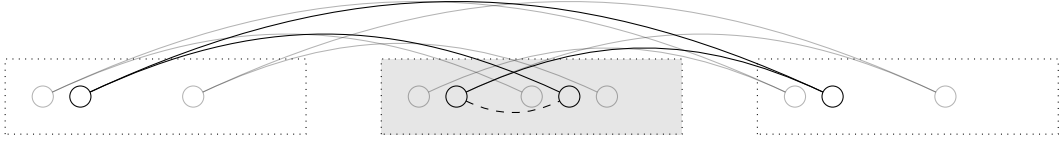
\begin{figure}[t]	
		\centering
		\input{figrrrxxxlll}
        \caption{The ordering of $G\tsd$ in the proof of \cref{rrrxxxlll}.}
		\label{fig:rrrxxxlll}
\end{figure}

\begin{theorem}\label{xxxlllrrr}
	\MIS is \NP-hard when restricted to (\abbcca, \acbnde, \aednbc, \aabbcc) -subgraph-free ordered graphs.
\end{theorem}
\begin{proof}
	Let us define an appropriate ordering of the vertices of $G\tsd$, see also \cref{fig:xxxlllrrr}.
	The order of vertices is as follows: all core vertices, precede all the left dummies, in the lexicographic order of their corresponding edges, which are in turn followed by all the right dummies, again in the lexicographic order of their respective edges.
	Thus we have obtained an ordered graph $G\tsd$.
	
	\begin{claim}
		$G\tsd$ is \abbcca-subgraph-free.
	\end{claim}
	\begin{claimproof}
		Suppose that $G\tsd$ contains a copy of \abbcca as a subgraph, and let $a, b, c, d$ be the consecutive vertices of \aabb in that copy.
		
		Since $b$ has a left neighbor, it is not a core vertex.
		Neither is $c$, as it follows $b$.
		Moreover, $c$ has a right neighbor, so it must be the left dummy $\ell_e$ of some edge $e$ of $G$.
		Therefore $d$ is the right dummy $r_e$.
		Observe that $b$ is followed by a left dummy, so it is not itself a right dummy, i.e., $b=\ell_f$ for some edge $f=uv$ with $u\prec v$.
		Furthermore, as $a$ is a left neighbor of $b$, it holds that $a=u$.
		Therefore, the vertices of \aabb are precisely $a=u, b=\ell_f, c=\ell_e, d=r_e$.
		Moreover, since $\ell_f \prec \ell_e$, it holds $f\prec e$.
		
		Now consider any vertex $u' \prec u$ and let $g$ be an edge incident to $u'$.
		Clearly, the left endpoint of $g$ precedes $u$, and thus $g \prec f \prec e$.
		Therefore $\ell_g \prec r_g \prec r_e$, so all the neighbors of $u'$ in $G\tsd$ precede $r_e$.
		Consequently, $G\tsd$ cannot contain \abbcca as a subgraph, which completes the proof of the claim.
	\end{claimproof}

    \begin{claim}
    $G\tsd$ is (\acbnde, \aednbc)-subgraph-free.
    \end{claim}
    \begin{claimproof}
       Suppose that $G\tsd$ contains a copy of one of the graphs listed in the claim as a subgraph, and let $a, b, c$ be the consecutive vertices of \acb in that copy.
        
       The vertex $c$, having two left neighbors, is a right dummy $r_{uv}$ of some edge $uv$ with $u\prec v$.
       Moreover, the remaining vertices of \acb are the only neighbors of $r_{uv}$, i.e., $a=v$ and $b=\ell_{uv}$.
       Because $r_{uv}$ can only be followed by other right dummies, and those form an independent set, a subgraph \acbnde cannot occur in $G\tsd$.
        
       Moreover, suppose that $G\tsd$ contains an edge $xy$ with $v\prec x\prec y \prec \ell_{uv}$.
       As there are no edges between core vertices or between left dummies, we have $y=\ell_{xv'}$, where $x \prec v'$, and $x$ and $v'$ are two core vertices.
       However, as the left dummies follow the lexicographic order of their corresponding edges and $y=\ell_{xv'} \prec \ell_{uv}$, we have $x \preceq u \prec v$, which is a contradiction with the choice of $xy$.
    \end{claimproof}
	
	\begin{claim}
		$G\tsd$ is \aabbcc-free.
	\end{claim}
	\begin{claimproof}
		Suppose that $G\tsd$ contains a copy of \aabbcc as a subgraph, and let $a, b, c, d, e, f$ be its consecutive vertices.
		
		The vertex $b$ has a left neighbor, so it must be a dummy.
		Thus $c$, must also be a dummy, because $b \prec c$.
		Then $d$ has a dummy as a left neighbor, meaning it could only be a right dummy.
		But in that case, both $e, f$ are right dummies, precluding an edge between them.
		This is a contradiction.
	\end{claimproof}

	This completes the proof.
\end{proof}

\begin{figure}[t]	
		\centering
		\input{figxxxlllrrr}
        \caption{The ordering of $G\tsd$ in the proof of \cref{xxxlllrrr}.}
		\label{fig:xxxlllrrr}
\end{figure}
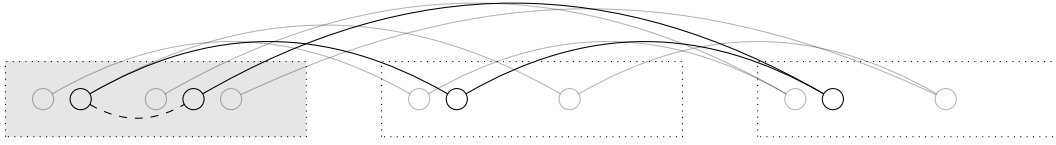

%% file: figxxxlrlrlr.tex
		\begin{tikzpicture}[
        every node/.style={draw, circle, inner sep=0, outer sep = 0, minimum size=8pt}, baseline= -0.5mm]
			\fill[opacity=.1] (0,0) rectangle (4,1);
            \draw[dotted] (0,0) rectangle (4,1);
            \draw[dotted] (5,0) rectangle (12,1);
            
            \node (u) at (1, .5) {};
            \node (v) at (2.5, .5) {};
            \node (l) at (8, .5) {};
            \node (r) at (8.5, .5) {};
            
            \node[opacity=.3] (u') at (.5, .5) {};
            \node[opacity=.3] (v') at (2, .5) {};
            \node[opacity=.3] (v'') at (3, .5) {};
            \node[opacity=.3] (l') at (6, .5) {};
            \node[opacity=.3] (r') at (6.5, .5) {};
            \node[opacity=.3] (l'') at (9.5, .5) {};
            \node[opacity=.3] (r'') at (10, .5) {};
            
            \draw (l)--(r);
            \draw[opacity=.3] (l')--(r');
            \draw[opacity=.3] (l'')--(r'');
            
            \path[draw] (u) to[bend left = 30] (l);
            \path[draw] (v) to[bend left = 30] (r);
            \path[draw, dashed] (u) to[bend right = 30] (v);
            \path[draw, opacity=.3] (u') to[bend left = 30] (l');
            \path[draw, opacity=.3] (v') to[bend left = 30] (r');
            \path[draw, opacity=.3] (u) to[bend left = 25] (l'');
            \path[draw, opacity=.3] (v'') to[bend left = 25] (r'');
            
		\end{tikzpicture}

%% file: figxxxrlrlrl.tex
		\begin{tikzpicture}[
        every node/.style={draw, circle, inner sep=0, outer sep = 0, minimum size=8pt}, baseline= -0.5mm]
			\fill[opacity=.1] (0,0) rectangle (4,1);
            \draw[dotted] (0,0) rectangle (4,1);
            \draw[dotted] (5,0) rectangle (12,1);
            
            \node (u) at (1, .5) {};
            \node (v) at (2.5, .5) {};
            \node (r) at (8, .5) {};
            \node (l) at (8.5, .5) {};
            
            \node[opacity=.3] (u') at (.5, .5) {};
            \node[opacity=.3] (v') at (2, .5) {};
            \node[opacity=.3] (v'') at (3, .5) {};
            \node[opacity=.3] (r') at (6, .5) {};
            \node[opacity=.3] (l') at (6.5, .5) {};
            \node[opacity=.3] (r'') at (9.5, .5) {};
            \node[opacity=.3] (l'') at (10, .5) {};
            
            \draw (l)--(r);
            \draw[opacity=.3] (l')--(r');
            \draw[opacity=.3] (l'')--(r'');
            
            \path[draw] (u) to[bend left = 30] (l);
            \path[draw] (v) to[bend left = 30] (r);
            \path[draw, dashed] (u) to[bend right = 30] (v);
            \path[draw, opacity=.3] (u') to[bend left = 30] (l');
            \path[draw, opacity=.3] (v') to[bend left = 30] (r');
            \path[draw, opacity=.3] (u) to[bend left = 30] (l'');
            \path[draw, opacity=.3] (v'') to[bend left = 30] (r'');
            
		\end{tikzpicture}

%% file: figrrrxxxlll.tex
		\begin{tikzpicture}[
        every node/.style={draw, circle, inner sep=0, outer sep = 0, minimum size=8pt}, baseline= -0.5mm]
			\fill[opacity=.1] (5,0) rectangle (9,1);
            \draw[dotted] (0,0) rectangle (4,1);
            \draw[dotted] (5,0) rectangle (9,1);
            \draw[dotted] (10,0) rectangle (14,1);
            
            \node (u) at (6, .5) {};
            \node (v) at (7.5, .5) {};
            \node (l) at (11, .5) {};
            \node (r) at (1, .5) {};
            
            \node[opacity=.3] (u') at (5.5, .5) {};
            \node[opacity=.3] (v') at (7, .5) {};
            \node[opacity=.3] (v'') at (8, .5) {};
            \node[opacity=.3] (l') at (10.5, .5) {};
            \node[opacity=.3] (r') at (.5, .5) {};
            \node[opacity=.3] (l'') at (12.5, .5) {};
            \node[opacity=.3] (r'') at (2.5, .5) {};
            
            \path[draw] (u) to[bend left = 25] (l);
            \path[draw] (v) to[bend right = 25] (r);
            \path[draw] (l) to[bend right = 25] (r);
            \path[draw, dashed] (u) to[bend right = 25] (v);
            \path[draw, opacity=.3] (u') to[bend left = 25] (l');
            \path[draw, opacity=.3] (v') to[bend right = 25] (r');
            \path[draw, opacity=.3] (l') to[bend right = 25] (r');
            \path[draw, opacity=.3] (u) to[bend left = 25] (l'');
            \path[draw, opacity=.3] (v'') to[bend right = 25] (r'');
            \path[draw, opacity=.3] (l'') to[bend right = 25] (r'');
            
		\end{tikzpicture}

%% file: figxxxlllrrr.tex
        \begin{tikzpicture}[
        every node/.style={draw, circle, inner sep=0, outer sep = 0, minimum size=8pt}, baseline= -0.5mm]
			\fill[opacity=.1] (0,0) rectangle (4,1);
            \draw[dotted] (0,0) rectangle (4,1);
            \draw[dotted] (5,0) rectangle (9,1);
            \draw[dotted] (10,0) rectangle (14,1);
            
            \node (u) at (1, .5) {};
            \node (v) at (2.5, .5) {};
            \node (l) at (6, .5) {};
            \node (r) at (11, .5) {};
            
            \node[opacity=.3] (u') at (.5, .5) {};
            \node[opacity=.3] (v') at (2, .5) {};
            \node[opacity=.3] (v'') at (3, .5) {};
            \node[opacity=.3] (l') at (5.5, .5) {};
            \node[opacity=.3] (r') at (10.5, .5) {};
            \node[opacity=.3] (l'') at (7.5, .5) {};
            \node[opacity=.3] (r'') at (12.5, .5) {};
            
            \path[draw] (u) to[bend left = 30] (l);
            \path[draw] (v) to[bend left = 30] (r);
            \path[draw] (l) to[bend left = 30] (r);
            \path[draw, dashed] (u) to[bend right = 30] (v);
            \path[draw, opacity=.3] (u') to[bend left = 30] (l');
            \path[draw, opacity=.3] (v') to[bend left = 30] (r');
            \path[draw, opacity=.3] (l') to[bend left = 30] (r');
            \path[draw, opacity=.3] (u) to[bend left = 30] (l'');
            \path[draw, opacity=.3] (v'') to[bend left = 25] (r'');
            \path[draw, opacity=.3] (l'') to[bend left = 30] (r'');
            
		\end{tikzpicture}

%% file: hardness-long-subdivision.tex
\subsection{Long subdivision}\label{sec:longsubdiv}
In this section, we again take advantage of Poljak's subdivision trick, albeit in another way.

\begin{theorem}\label{squstr}
	\MIS is \NP-hard on (\abxxba, \abcbca) -subgraph-free ordered graphs and on \abcabc -subgraph-free ordered graphs.
\end{theorem}
\begin{proof}
We will provide two parallel reductions: one for (\abxxba, \abcbca)-subgraph-free case and another for \abcabc-subgraph-free case. The two reductions are very similar, and the main difference is in the way we order the vertices of the constructed graph.

Let $(G,k)$ be an instance of \MIS, and let us fix an arbitrary ordering on $V(G)$.
We will construct two ordered graphs $G_1$, $G_2$ and integers $k_1,k_2$ such that: (i) $\alpha(G_1)\geq k_1$ if and only if $\alpha(G)\geq k$, (ii) $\alpha(G_2)\geq k_2$ if and only if $\alpha(G)\geq k$, (iii) $G_1$ is (\abxxba, \abcbca)-subgraph-free, and (iv) $G_2$ is \abcabc-subgraph-free.
The construction of $G_1$ and $G_2$ is performed in two steps; most of the parts of the constructions are common to $G_1$ and $G_2$. See also \cref{fig:longsubdiv}.
    
\paragraph{Step 1.} As an intermediate step we construct two ordered graphs $G_1'$ and $G_2'$ -- these graphs have the same underlying graph, but different ordering.
Let us enumerate the edges of $G$ according to their lexicographic order; for an edge $e\in E(G)$, let $\pi(e)$ denote its position, i.e., the number of edges preceding $e$ plus one.
For every edge $e=uv\in E(G)$, where $u\prec v$, let us doubly subdivide it $\pi(e)$ times, replacing $e$ with the path $u-u_{1,e}-u_{2,e}-\ldots-u_{\pi(e), e}-v_{\pi(e), e} - \ldots - v_{2,e} - v_{1,e}-v$.
Next, we partition the expanded set of vertices into sets $V_0,V_1,V_2,\ldots,V_{|E(G)|}$ according to their indices, namely by stipulating that $V_0=V(G)$, and for $e=uv\in E(G)$, $i\in [|E(G)|]$,  we have $u_{i,e}, v_{i,e}\in V_i$.
We call the sets $V_0,V_1,V_2,\ldots,V_{|E(G)|}$ \emph{layers}.
    
Now the order of the vertices of $V_0\cup\ldots\cup V_{|E(G)|}$ is slightly different in $G_1'$ and $G_2'$. In both cases, we order the vertices so that for $i<j$, we have $V_i\prec V_j$.
Furthermore, in $G_1'$, the order among vertices in each $V_i$ is such that: $u_{i,e}\prec v_{i,f}$ whenever $u\prec v$ in $G$, and $u_{i,e}\prec u_{i,f}$ whenever $e\prec f$ lexicographically in $G$.
In $G_2'$, the order among vertices in each $V_i$ depends on the parity of $i$: (i) if $i$ is even, we stipulate that $u_{i,e}\prec v_{i,f}$ whenever $u\prec v$ in $G$ and $u_{i,e}\prec u_{i,f}$ whenever $e\prec f$ in $G$, and (ii) if $i$ is odd, we order the vertices in $V_i$ so that $u_{i,e}\prec v_{i,f}$ whenever $v \prec u$ in $G$ and $u_{i,e}\prec u_{i,f}$ whenever $f \prec e$ in $G$.
This completes the construction of $G_1'$ and $G_2'$.
We proceed to the second step.

\paragraph{Step 2.} Let $p\in \{1,2\}$. For every $e=uv\in E(G)$, where $u\prec v$, we proceed as follows.
Consider the segment of $G_p'$ with endpoints $u_{\pi(e), e}$ and $v_{\pi(e), e}$ -- we point out that by the construction, there is an edge $u_{\pi(e), e}v_{\pi(e), e}$.
We again aim to subdivide this edge an even number of times in order to avoid the forbidden subgraphs.

   Let $\{x_0,\ldots,x_k\}$ be the vertex set of the segment so that $u_{\pi(e), e}=x_0 \prec x_1\prec x_2 \prec \ldots \prec x_k=v_{\pi(e), e}$, where $k \in \mathbb{N}$ is the length of the segment.
    For each $i\in \{0,\ldots, k-1\}$, we create a vertex $\omega_{i,e}$, and we insert it in the order so that $x_i\prec \omega_{i,e}\prec x_{i+1}$.
    Additionally, if $k$ is odd, we create another vertex $\omega_{k,e}$ with $\omega_{k-1,e} \prec \omega_{k,e} \prec v_{\pi(e),e}$.
    Let $\Omega_e$ be the set of vertices created this way, together with $x_0$ and $x_k$.
    Note that we have ensured that $\Omega_e$ has even size.
   We replace the edge $x_0x_k$ with the path $x_0-\omega_{0,e}-\omega_{1,e}-\ldots -\omega_{k-1,e}-(\omega_{k,e}-)x_k$.
   This completes the construction of $G_p$.
    
\paragraph{Correctness.} Observe that $G_1$ and $G_2$ were constructed from $G$ by some numbers $\ell_1$ and $\ell_2$, respectively, of double subdivisions.
Let $k_1=k+\ell_1$ and $k_2=k+\ell_2$.
As we already mentioned in the previous section, it holds that $\alpha(G)\geq k$ if and only if $\alpha(G_1)\geq k_1$, and $\alpha(G)\geq k$ if and only if $\alpha(G_2)\geq k_2$, as desired.
Moreover, the number of double subdivisions performed in the first step is $\Oh(m^2)$, and $\Oh(n\cdot m)$ in the second, where $n,m$ are, respectively, the number of vertices and edges of $G$.
Therefore, the instances $(G_1,k_1)$ and $(G_2,k_2)$ are constructed in polynomial time.
It remains to prove that $G_1$ and $G_2$ do not contain forbidden subgraphs.
 
We first divide the edges of $G_1$ and $G_2$ into two types: the edges that join a vertex from some $V_i$ with a vertex in $V_{i+1}$ -- we will call those \emph{long} edges; the others that join pairs of vertices in some $\Omega_e$ -- those we will call \emph{short} edges.
The following observation about short edges in both $G_1$ and $G_2$ is immediate from the construction.
    
     \begin{claim}\label{squstrshort}
     If $xy$ is a short edge, then there is at most one vertex $z\neq x,y$ in the segment $xy$.
     Moreover, $z$ belongs to some $V_i$, and not to any $\Omega_e$, and thus $z$ is not an endpoint of another short edge.
    \end{claim}
  
  Furthermore, we have the following property of the long edges in $G_1$.  
    \begin{claim}\label{squstrlong}
    Let $xy$ and $x'y'$ be long edges of $G_1$ such that $x\prec y$, $x'\prec y'$, and $x\prec x'$, and all vertices $x,x',y,y'$ are distinct. Then $y\prec y'$.
    \end{claim}
    \begin{claimproof}
    Suppose that there are long edges $xy$, $x'y'$ such that $x\prec y$, $x'\prec y'$, and $x\prec x'$.
   Since $xy$ is long, there is $i\in\{0,\ldots,|E(G)|-1\}$ such that $x\in V_i$ and $y\in V_{i+1}$.
   Moreover, as $x'y'$ is long and we have $x\prec x'\prec y'\preceq y$, it must hold that $x'\in V_i$ and $y'\in V_{i+1}$.
   Recall that by the construction of $V_i$, there are $e=uv\in E(G)$ and $f=u'v'\in E(G)$ such that $x=u_{i,e}$, $y=u_{i+1,e}$, $x'=u'_{i,f}$, and $y'=u'_{i+1,f}$.
   Since $u_{i,e}\prec u'_{i,f}$, we either have $u\prec u'$ in $G$ or $u=u'$ and $e$ precedes $f$ in the lexicographic order.
   In both cases, this implies that $y\prec y'$, as desired.
    \end{claimproof}

Now we are ready to prove that $G_1$ does not contain forbidden subgraphs.
    \begin{claim}
    $G_1$ is \abxxba-subgraph-free.
    \end{claim}
    \begin{claimproof}
        Suppose there is a subgraph \abxxba in $G_1$.
	    By \cref{squstrshort}, both its edges are long.
	    But \cref{squstrlong} precludes long edges from inducing an \abba, a contradiction.
    \end{claimproof}

    \begin{claim}
        $G_1$ is \abcbca-subgraph-free.
    \end{claim}
    \begin{claimproof}
        Suppose there is such a subgraph in $G_1$.
        By \cref{squstrshort}, the edge joining the first vertex and the last is long, and then by \cref{squstrlong} the remaining two are short.
        But this contradicts the second part of \cref{squstrshort}, as short edges cannot induce an \abab.
    \end{claimproof}
    
    It remains to prove that $G_2$ is \abcabc-subgraph-free. We again first prove the following observation on the long edges in $G_2$.
    	\begin{claim}\label{squfliplong}
	 Let $xy$ and $x'y'$ be long edges of $G_2$ such that $x$ and $x'$ belong to the same layer, and $x\prec y$, $x'\prec y'$, and $x\prec x'$, and all vertices $x,x',y,y'$ are distinct. Then $y'\prec y$.
	\end{claim}
	\begin{claimproof}
	Let $i\in\{0,\ldots,|E(G)|-1\}$ be such that $x,x'\in V_i$, and suppose there are long edges $xy$ and $x'y'$ such that $x\prec y$, $x'\prec y'$, and $x\prec x'$, and all vertices $x,x',y,y'$ are distinct.
Since $xy$and  $x'y'$ are long, we have $y,y'\in V_{i+1}$.
	Recall that by the construction, there are $e=uv\in E(G)$ and $f=u'v'\in E(G)$ such that $x=u_{i,e}$, $y=u_{i+1,e}$, $x'=u'_{i,f}$, and $y'=u'_{i+1,f}$.
 Since $u_{i,e}\prec u'_{i,f}$, we either have $u\prec u'$ in $G$ or $u=u'$ and $e \prec f$ lexicographically.
 As $i$ and $i+1$ have opposite parity, in both cases this implies that $y' \prec y$, as desired.	
	\end{claimproof}
	
We are ready to prove that $G_2$ is \abcabc-subgraph-free.

	\begin{claim}
		$G_2$ is \abcabc-subgraph-free.
	\end{claim}
	\begin{claimproof}
		Suppose that $G_2$ contains a copy of \abcabc as a subgraph, and let $a, b, c, d, e, f$ be its consecutive vertices.
		By \cref{squstrshort}, all edges of \abcabc are long.
		Let $a\in V_i$.
		This implies $d\in V_{i+1}$.
		By \cref{squfliplong} applied to edges $ad$ and $be$, we have $b\notin V_i$.
		Furthermore, since $a \prec b \prec d$, it holds that $b\in V_{i+1}$.
		Now $c$, being both preceded and followed by some vertex from $V_{i+1}$, must itself belong to $V_{i+1}$.
		This, however, implies that the edges $be$ and $cf$ violate \cref{squfliplong}, a contradiction that proves the claim.
	\end{claimproof}

	This completes the proof.
\end{proof}
\begin{figure}[t]	
		\centering
		\input{figlong1}
        \input{figlong2}
		\caption{A layer $V_e$, containing the path $\Omega_e$, in $G_1$ (up) and $G_2$ (down) as in the proof of \cref{squstr}.}
		\label{fig:longsubdiv}
\end{figure}
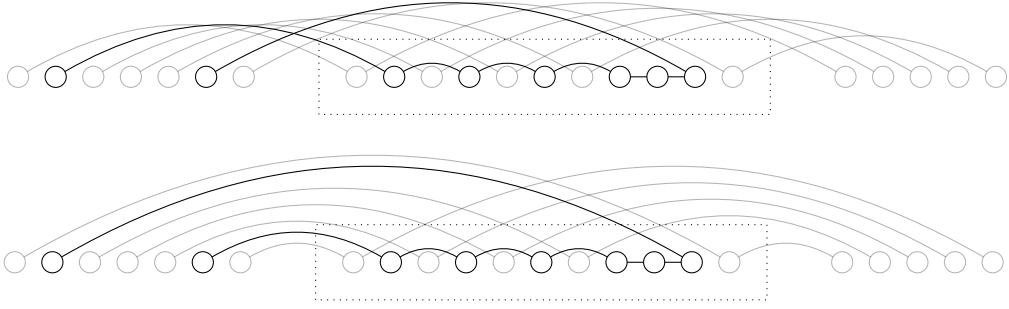

%% file: figlong1.tex
\begin{tikzpicture}[
        every node/.style={draw, circle, inner sep=0, outer sep = 0, minimum size=8pt}, baseline= -0.5mm]
            \draw[dotted] (4.5, 0) rectangle (10.5, 1);
        
            \node[opacity=.3] (a) at (.5,.5) {};
            \node (b) at (1, .5) {};
            \node[opacity=.3] (c) at (1.5, .5) {};
            \node[opacity=.3] (d) at (2,.5) {};
            \node[opacity=.3] (e) at (2.5, .5) {};
            \node (g) at (3,.5) {};
            \node[opacity=.3] (h) at (3.5, .5) {};

            \node[opacity=.3] (a1) at (5,.5) {};
            \node (b1) at (5.5, .5) {};
            \node[opacity=.3] (c1) at (6, .5) {};
            \node (c') at (6.5, .5) {};
            \node[opacity=.3] (d1) at (7,.5) {};
            \node (d') at (7.5, .5) {};
            \node[opacity=.3] (e1) at (8, .5) {};
            \node (e') at (8.5, .5) {};
            \node (t) at (9, .5) {};
            \node (g1) at (9.5,.5) {};
            \node[opacity=.3] (h1) at (10, .5) {};

            \node[opacity=.3] (a2) at (11.5,.5) {};
            \node[opacity=.3] (c2) at (12, .5) {};
            \node[opacity=.3] (d2) at (12.5,.5) {};
            \node[opacity=.3] (e2) at (13, .5) {};
            \node[opacity=.3] (h2) at (13.5, .5) {};

            \path[draw, opacity=.3] (a) to[bend left=30] (a1);
            \path[draw] (b) to[bend left=30] (b1);
            \path[draw, opacity=.3] (c) to[bend left=30] (c1);
            \path[draw, opacity=.3] (d) to[bend left=30] (d1);
            \path[draw, opacity=.3] (e) to[bend left=30] (e1);
            \path[draw] (g) to[bend left=30] (g1);
            \path[draw, opacity=.3] (h) to[bend left=30] (h1);

            \path[draw] (b1) to[bend left=30] (c');
            \path[draw] (c') to[bend left=30] (d');
            \path[draw] (d') to[bend left=30] (e');
            \draw (e') -- (t);
            \draw (t) -- (g1);
            
            \path[draw, opacity=.3] (a1) to[bend left=30] (a2);
            \path[draw, opacity=.3] (c1) to[bend left=30] (c2);
            \path[draw, opacity=.3] (d1) to[bend left=30] (d2);
            \path[draw, opacity=.3] (e1) to[bend left=30] (e2);
            \path[draw, opacity=.3] (h1) to[bend left=30] (h2);
                     
\end{tikzpicture}

%% file: figlong2.tex
\begin{tikzpicture}[
        every node/.style={draw, circle, inner sep=0, outer sep = 0, minimum size=8pt}, baseline= -0.5mm]
            \draw[dotted] (4.5, 0) rectangle (10.5, 1);
        
            \node[opacity=.3] (a) at (.5,.5) {};
            \node (b) at (1, .5) {};
            \node[opacity=.3] (c) at (1.5, .5) {};
            \node[opacity=.3] (d) at (2,.5) {};
            \node[opacity=.3] (e) at (2.5, .5) {};
            \node (g) at (3,.5) {};
            \node[opacity=.3] (h) at (3.5, .5) {};

            \node[opacity=.3] (h1) at (5,.5) {};
            \node (b1) at (5.5, .5) {};
            \node[opacity=.3] (e1) at (6, .5) {};
            \node (c') at (6.5, .5) {};
            \node[opacity=.3] (d1) at (7,.5) {};
            \node (d') at (7.5, .5) {};
            \node[opacity=.3] (c1) at (8, .5) {};
            \node (e') at (8.5, .5) {};
            \node (t) at (9, .5) {};
            \node (g1) at (9.5,.5) {};
            \node[opacity=.3] (a1) at (10, .5) {};

            \node[opacity=.3] (a2) at (11.5,.5) {};
            \node[opacity=.3] (c2) at (12, .5) {};
            \node[opacity=.3] (d2) at (12.5,.5) {};
            \node[opacity=.3] (e2) at (13, .5) {};
            \node[opacity=.3] (h2) at (13.5, .5) {};

            \path[draw, opacity=.3] (a) to[bend left=30] (a1);
            \path[draw] (b) to[bend left=30] (g1);
            \path[draw, opacity=.3] (c) to[bend left=30] (c1);
            \path[draw, opacity=.3] (d) to[bend left=30] (d1);
            \path[draw, opacity=.3] (e) to[bend left=30] (e1);
            \path[draw] (g) to[bend left=30] (b1);
            \path[draw, opacity=.3] (h) to[bend left=30] (h1);

            \path[draw] (b1) to[bend left=30] (c');
            \path[draw] (c') to[bend left=30] (d');
            \path[draw] (d') to[bend left=30] (e');
            \draw (e') -- (t);
            \draw (t) -- (g1);
            
            \path[draw, opacity=.3] (a1) to[bend left=30] (a2);
            \path[draw, opacity=.3] (c1) to[bend left=30] (c2);
            \path[draw, opacity=.3] (d1) to[bend left=30] (d2);
            \path[draw, opacity=.3] (e1) to[bend left=30] (e2);
            \path[draw, opacity=.3] (h1) to[bend left=30] (h2);
                     
\end{tikzpicture}

%% file: hardness-chain-reduction.tex
\subsection{Train reduction}\label{sec:chainreduction}

In this Section, we provide the final hardness results, namely concerning $H\in \{\abxba,\abccab\}$.
The constructions we use are more involved than those we have presented so far; they are inspired by the \NP-hardness proof for \textsc{List 4-Coloring} in \abba-free graphs~\cite{DBLP:conf/stacs/PiecykR26}.

We organize this Section into three parts.
Our aim is to reduce \MIS in general graphs to \MIS in $H$-free ordered graphs.
First, in \cref{sec:traindefs}, we introduce the main building blocks in our reductions -- \emph{trains}, \emph{coupling}, and \emph{permutation gadgets}.
Next, in \cref{sec:trainreduction} we show how to use these building blocks in order to obtain the desired reductions.
Finally, in \cref{sec:traingadgets} we present the construction of permutation gadgets, completing the proof.

\subsubsection{Building blocks}\label{sec:traindefs}

The main component of our constructions is a \emph{train}.

\begin{definition}[Train]\label{deflink}
	For integers $\ell,\ell' \geq 0$,
	a \emph{train} is a tuple
	$ \mathbi{T}=(T, (x_1, x_2, \ldots, x_\ell), (y_1, y_2, \ldots y_{\ell'})),$
	where $T$ is an ordered graph, such that
	\begin{enumerate}[(T1)]
		\item the vertices $x_1, x_2, \ldots, x_\ell,y_1, y_2, \ldots y_{\ell'}$ are pairwise distinct, \label{prop:link-distinct}
		\item $x_1, x_2, \ldots, x_\ell$ are, in that order, a prefix of $T$, \label{prop:link-prefix}
		\item $y_1, y_2, \ldots, y_{\ell'}$ are, in that order, a suffix of $T$, \label{prop:link-suffix}
		\item both $\{x_1, x_2, \ldots, x_\ell\}$ and $\{y_1, y_2, \ldots y_{\ell'}\}$ are independent sets. \label{prop:link-independent}
	\end{enumerate}
	We call $(x_1, x_2, \ldots, x_\ell)$ (resp., $(y_1, y_2, \ldots y_{\ell'})$'s) the \emph{input vertices} (resp., \emph{output vertices})  of the train $\mathbi{T}$.
	By the \emph{underlying graph} of $\mathbi{T}$ we mean that of $T$.
	For an ordered graph $H$, we say that $\mathbi{T}$ is $H$-free if $T$ is $H$-free.
\end{definition}

\paragraph{Coupling and avoidability.} Consider two trains 
\[\mathbi{T}^1=(T^1, (x_1^1, x_2^1, \ldots, x_\ell^1), (y_1^1, y_2^1, \ldots, y_{\ell'}^1)) \text{ \ and \ } \mathbi{T}^2=(T^2, (x_1^2, x_2^2, \ldots, x_{\ell'}^2), (y_1^2, y_2^2, \ldots, y_{\ell''}^2))\]
on disjoint sets of vertices.
Let $T$ be the ordered graph obtained from $T^1$ and $T^2$ by identifying $y_i^1$ with $x_i^2$ for each $i\in [\ell']$.
The train $\mathbi{T} = (T, (x_1^1, x_2^1, \ldots, x_\ell^1), (y_1^2, y_2^2, \ldots y_{\ell''}^2))$ will be called the \emph{coupling} of $\mathbi{T}^1$ and $\mathbi{T}^2$.
The same name will also denote the operation of obtaining $\mathbi{T}$ from its component trains.
Clearly, $\mathbi{T}$ satisfies properties \ref{prop:link-distinct}, \ref{prop:link-prefix}, \ref{prop:link-suffix}, and \ref{prop:link-independent} of \cref{deflink}.

However, it might happen that the coupling of two $H$-free trains is not $H$-free itself.
We say that an ordered graph $H$ is \emph{avoidable} if the coupling of any two $H$-free trains is $H$-free.

Let us introduce a criterion for avoidability.

\begin{lemma}\label{lem:avoidable}
	Suppose $H$ is an ordered graph containing (not necessarily distinct) vertices $a, b, c, d, e, f$, in that order, such that:
	\begin{itemize}
		\item[1.] $a$ is the first and $f$ the last vertex of $H$;
		\item[2.] $H$ contains edges $ae, bf, cd$.
	\end{itemize}
	Then $H$ is avoidable.
\end{lemma}
\begin{proof}
	Let 
	$\mathbi{T}^1=(T^1, (x_1^1, x_2^1, \ldots, x_\ell^1), (y_1^1, y_2^1, \ldots, y_{\ell'}^1))$ and $\mathbi{T}^2=(T^2, (x_1^2, x_2^2, \ldots, x_{\ell'}^2), (y_1^2, y_2^2, \ldots, y_{\ell''}^2))$	be two $H$-free trains and let $\mathbi{T}$ be their coupling.
	For contradiction, suppose that $\mathbi{T}$ contains an induced copy of $H$.
	Since $\mathbi{T}^1$ is $H$-free, $f \in V(T^2) \setminus V(T^1)$.
	Because of the edge $bf$, we have $b\in V(T^2)$.
	Similarly, $a, e \in V(T^1)$.
	Now, as $b\preceq c \preceq d\preceq e$, it must hold that $c, d \in V(T^1)\cap V(T^2)$.
	However, the set $V(T^1)\cap V(T^2)$ is independent, contradicting the existence of the edge $cd$.
\end{proof}

Notably, this implies the avoidability of both \abxba and \abccab.
Indeed, letting $a=b$ and $e=f$ yields \abba, a subgraph of \abxba.
In case of \abccab we apply \cref{lem:avoidable} with all vertices distinct.

\paragraph{Interchangeability.}

Let us start with a definition of an operation that is similar in spirit to Poljak's subdivision trick, i.e., a local replacement of certain induced subgraph in a graph that preserves the independence number.
The definitions below concern unordered graphs.

A \emph{boundaried graph} $(G,B)$ is a graph with distinguished subset of vertices $B$ called the \emph{boundary}.
Two boundaried graphs $(G_1,B_1)$ and $(G_2,B_2)$ are \emph{compatible} if $V(G_1)\cap V(G_2) = B_1 = B_2$ and $G_1[B_1] = G_2[B_2]$. 
Let $G$ be a graph that contains a boundaried graph $(G_1,B)$ as an induced subgraph such that $N(V(G_1)\setminus B) \subseteq B$. 
By \emph{replacing} $(G_1,B)$ with a compatible boundaried graph $(G_2,B)$ we mean removing all vertices from $V(G_1) \setminus B$ and adding the vertices from $V(G_2)\setminus B$ together with all edges of $G_2$ incident to these vertices.

\begin{definition}[Interchangeability]\label{def:braiding}
Let $(G_1,B)$, $(G_2,B)$ be two compatible boundaried graphs.
We say that $(G_1,B)$ and $(G_2,B)$ are \emph{interchangeable} if for every graph $G$, the following holds.
	\begin{description}
	\item[($\star$)] Suppose that $G$ contains an induced copy of $G_1$ such that $N(V(G_1)\setminus B) \subseteq B$, and let $G'$ be obtained from $G$ by replacing $(G_1,B)$ with $(G_2,B)$.
	Then $\alpha(G)=\alpha(G')$.
	\end{description}
\end{definition}

\paragraph{Permutation gadgets.}
We are ready to introduce the crucial blocks (or rather, carriages) of our reductions in this section.

\begin{definition}\label{def:perm}
	Let $\ell \geq 1$ be an integer and let $\sigma:[\ell]\to [\ell]$ be a permutation.
	A \emph{permutation gadget for $\sigma$} is a pair $(\mathbi{P}, \overline{k})$,
	where $\mathbi{P}=(P, (x_1, x_2, \ldots, x_\ell), (y_1, y_2, \ldots, y_\ell))$ is a train
	and $\overline{k} = (k^1,\ldots,k^\ell)$ is a vector of $\ell$ positive integers, all of equal parity, such that the following holds:\\	
	Let $L$ be the (unordered) linear forest consisting of $\ell$ paths $L^1, L^2, \ldots, L^\ell$, where $L^i$ is a path with endpoints $x_i$ and $y_{\sigma(i)}$ and with $k^i$ vertices in total.	
	Then, $(L,\bigcup_{i=1}^\ell \{x_i,y_i\})$ is interchangeable with $(P, \bigcup_{i=1}^\ell \{x_i,y_i\})$,
	where in the latter we treat $P$ as an unordered graph.

	Whenever $\overline{k}$ is clear from the context, we will simply say that $\mathbi{P}$ is a permutation gadget for $\sigma$.
\end{definition}

In the following lemma we show that for $H\in \{\abxba, \abccab \}$, such $H$-free permutation gadgets exist.

\begin{lemma}\label{lem:permutation-gadgets}
	For every $\ell$, and for every permutation $\sigma:[\ell]\rightarrow [\ell]$ there exist:
	\begin{enumerate}[(1.)]
		\item a \abxba-free permutation gadget for $\sigma$;
		\item a \abccab-free permutation gadget for $\sigma$.
	\end{enumerate}
	Moreover, both gadgets can be constructed in time polynomial in $\ell$.
\end{lemma}

\subsubsection{Hardness reduction}\label{sec:trainreduction}

We postpone the proof of \cref{lem:permutation-gadgets} to \cref{sec:traingadgets}, and now we show how it implies the desired hardness results.

\begin{theorem}\label{thm:chainred}
	\MIS is \NP-hard when restricted to \abxba-free ordered graphs and to \abccab-free ordered graphs.
\end{theorem}
\begin{proof}
	Let $(G,k)$ be an instance of \MIS and let $m$ be the number of edges of $G$.

	We aim to construct ordered graphs $G_1$, $G_2$ and integers $k_1,k_2$ such that: (i) $G_1$ is \abxba-free, (ii) $G_2$ is \abccab-free, (iii) $\alpha(G_1)\geq k_1$ if and only if $\alpha(G)\geq k$, and (iv) $\alpha(G_2)\geq k_2$ if and only if $\alpha(G)\geq k$.
	Most parts of the construction of $G_1$ and $G_2$ will be common, so we will describe them together.
	
	\paragraph{Construction of $G_1$ and $G_2$.}
	Each of the graphs $G_1$ and $G_2$ is obtained by the coupling of three trains.

	The first one is the same for both $G_1$ and $G_2$ and it is called the \emph{locomotive} (see  \cref{fig:choo-choo}).
	We start with introducing the set $V(G)$, which we order the arbitrarily.
	Then, we define the set $V_E$ so that for every $v\in V(G)$, and for every $e\in E(G)$ such that $v$ is an endpoint of $e$, we add to $V_E$ the vertex $v_e$, adjacent to $v$.
	Note that the graph constructed so far is similar to the 2-subdivision of $G$ discussed in previous sections, except we do not add any edges between dummy vertices. 
	We place the vertices from $V_E$ after the set $V(G)$, and the order inside $V_E$ is such that $u_e\prec v_f$ whenever $u\prec v$ and the order between $v_e$ and $v_f$ is arbitrary.
	The input of the locomotive is empty, while the output consists of the vertices from $V_E$.
	Note that $|V_E| = 2m$.
	
	The second train is a permutation gadget for an appropriate permutation $\sigma$ of $[2m]$.
	For each $v_e \in V_E$, let $\pi(v_e) \in [2m]$ denote the position of $v_e$ in $V_E$. 
	Let $\sigma: [2m]\to [2m]$ be a permutation such that,
	for every $e=uv\in E(G)$,
	there is $i\in [2m-1]$ such that $\{\sigma(\pi(u_e)),\sigma(\pi(v_e))\}=\{i,i+1\}$.
	We invoke \cref{lem:permutation-gadgets} to construct a \abxba-free (in the case of $G_1$) or a \abccab-free (in the case of $G_2$) permutation gadget $(\cP, \overline{k})$ for $\sigma$,
	where $\cP=(P, (x_1, x_2, \ldots, x_{2m}), (y_1, y_2, \ldots, y_{2m}))$.
		
	Now let us proceed to the third (and final) train, called \emph{caboose} (see \cref{fig:choo-choo}).
	It contains vertices $y_1,\ldots,y_{2m}$ as the input.
	Next, we introduce the vertices $z_1,\ldots,z_{2m}$, which we place in that order after the vertex $y_{2m}$.
	For every $j\in [m]$, we add the edge $z_{2j-1}z_{2j}$.
	In the construction of $G_1$, we add the edge $y_{i}z_{i}$ for every $i\in [2m]$.
	In the construction of $G_2$, we add the edge $y_{i}z_{2m+1-i}$ instead.	
	The output of the caboose is empty.

	Now, $G_1$ and $G_2$ are both obtained by coupling of the locomotive,
	the permutation gadget $\cP$, and the caboose.
	This completes the construction of $G_1$ and $G_2$.
	Clearly, $G_1$ and $G_2$ are constructed in polynomial time.
	
	\paragraph{$H$-freeness: properties (i) and (ii).} 
	Let us verify that $G_1$ is \abxba-free and $G_2$ is \abccab-free.
	As both \abxba and \abccab are avoidable, it is sufficient to verify that each of the three trains we used in the construction of $G_1$ and $G_2$ is, respectively, \abxba-free or \abccab-free.
	
	It is straightforward to verify that the locomotive is even \abba-free,
	so in particular \abxba-free and \abccab-free. 
	The corresponding property of the permutation gadget follows directly from \cref{lem:permutation-gadgets}.

	Thus, let us focus on the caboose.
	For $G_1$, notice that the middle edge of every copy of \abba in this train has to be contained in $\{z_1,\ldots,z_{2m}\}$. This means that the endvertices of this edge must be consecutive so, in particular, this train is \abxba-free.
	In the case of $G_2$, the caboose does not even contain \abab, so it is in particular \abccab-free.
	
	Thus, $G_1$ is \abxba-free and $G_2$ is \abccab-free, as claimed.
	
	\paragraph{Equivalence of instances: properties (iii) and (iv).}
	It remains to define $k_1$ and $k_2$ and prove the equivalence of the instances.
	We show the case of $G_1$; the case of $G_2$ is analogous.
	Consider the permutation gadget $(\cP, \overline{k})$ used in the construction, where $\overline{k} = (k^1, \ldots, k^{2m})$.
	Let us define $B = \bigcup_{i=1}^{2m} \{x_i, y_i\}$.
	
	Let $L$ be the linear forest, as \cref{def:perm}, and denote by $L^1, L^2, \ldots, L^{2m}$ its paths, so that $L^i$ has $k^i$ vertices and joins $x_i$ with $y_{\sigma(i)}$.
	By the definition of $\cP$ and $L$, we know that that $(P, B)$ is interchangeable with $(L, B)$.
	Let $G'_1$ be the unordered graph obtained from the underlying graph of $G_1$ by replacing $(P, B)$ with $(L, B)$.
	Since $(P,B)$ and $(L,B)$ are interchangeable, we have $\alpha(G_1)=\alpha(G_1')$.

	We claim that $G_1'$ can be equivalently obtained from $G$ by subdividing each edge some even number of times.
	Consider an edge $e=uv$ of $G$.
	Let $i = \pi(u_e)$ and $j=\pi(v_e)$, and by symmetry assume that $\sigma(i)<\sigma(j)$, i.e., $\sigma(j) = \sigma(i)+1$.
	In $G'_1$, the edge $e$ is corresponds to the path:
	\[
	 u-\underbrace{u_{e}-\ldots-y_{\sigma(i)}}_{\text{$L^i$, i.e., $k^i$ vertices}}-z_{\sigma(i)}-z_{\sigma(j)}-\underbrace{y_{\sigma(j)}-\ldots-v_{e}}_{_{\text{$L^j$, i.e., $k^j$ vertices}}}-v,
	\]
	In total, this path has $k^i+k^j+4$ vertices, which is an even number as $k^i$ and $k^j$ have the same parity.
	Thus, in the construction of $G_1'$ from $G$, the edge was doubly subdivided $k_e := (k^i+k^j+2)/2$ times.
	
	Since the edge $uv$ was chosen arbitrarily, we conclude that $G'_1$ was obtained by applying Poljak's subdivision trick to $G$ the total of
	$
	\sum_{e\in E(G)} k_e,
	$
	times, and hence  we have $\alpha(G_1)=\alpha(G_1')=\alpha(G) + \sum_{e\in E(G)} k_e$.
	Setting $k_1 = k + \sum_{e\in E(G)} k_e$  proves (iii).
	The arguments for (iv) are analogous, so the proof is complete.
\end{proof}

\begin{figure}[t]	
	\centering
	\input{figopen1}
    
	\input{figclose1}
    
	\input{figclose2}
	\caption{The locomotive and two cabooses used in the proof of \cref{thm:chainred}. In this and subsequent figures, underscoring indicates the input and output vertices.}
	\label{fig:choo-choo}
\end{figure}
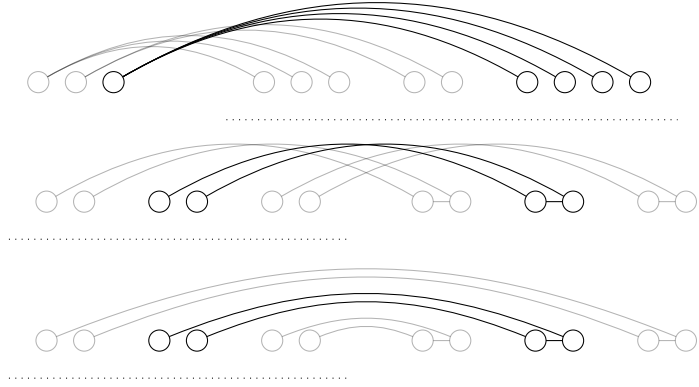

\subsubsection{Constructing permutation gadgets}\label{sec:traingadgets}

It remains to prove \cref{lem:permutation-gadgets}.
We start with the following construction (see \cref{fig:braiding}).

\begin{lemma}\label{lem:braiding}
Let $G_1$ and $G_2$ be graphs defined as follows:
\begin{align*}
& V(G_1)=\{a, b_1, c, d, e_1, f_1, g\}, & E(G_1)=\{ab_1, b_1c, de_1, e_1f_1, f_1g\}, & \\
& V(G_2) = \{a, b_2, b_2', c, d, e_2, f_2, g\}, & E(G_2)=\{ab_2, b_2c, ab_2', b_2'c, b_2b_2', de_2, e_2f_2, f_2g, b_2e_2, b_2'f_2\}, &
\end{align*}
and let $B=V(G_1)\cap V(G_2)=\{a, c, d, g\}$.
Then, $(G_1,B)$ and $(G_2,B)$ are interchangeable.
\end{lemma}
\begin{proof}
	Let $G$ be an arbitrary graph containing an induced copy of $G_1$, with $N(V(G_1)\setminus B) \subseteq B$, and let $G'$ be the graph obtained by replacing this copy with $G_2$.

    First, observe that $G_2$ contains a subgraph isomorphic to $G_1$:
	indeed, removing the vertex $b_2'$ and, additionally, the edge $b_2e_2$, yields a graph isomorphic to $G_1$.
    Therefore, $G'$ also contains a copy of $G$ as a subgraph, and the corresponding mapping $\phi: V(G)\to V(G')$ of vertices is as follows:
	\[
		\phi(v) =
		\begin{cases}				
		v & \text{if } v\in V(G)\setminus \{b_1,e_1,f_1\} \\
		b_2 & \text{if } v=b_1 \\
		e_2 & \text{if } v=e_1 \\
		f_2 & \text{if } v=f_1.
		\end{cases}
	\]

	\begin{claim}
		It holds that $\alpha(G) \geq \alpha(G')$.\label{claim:braiding-geq}	
	\end{claim}
	\begin{claimproof}
		Let $I'$ be an independent set in $G'$, we will find an independent set $I$ in $G$ of the same size.
		If $b_2' \notin I'$, we set $I= \{ \phi^{-1}(v) ~|~ v \in I' \}$.
		Thus, suppose that $b_2' \in I'$. In particular, this means that $a,b_2,c \notin I'$.
		We can we set $I = \{ \phi^{-1}(v) ~|~ v \in I'\setminus \{b_2'\} \} \cup \{b_2\}$.
		It is straightforward to verify that in both cases $I$ is independent in $G$ and that $|I|=|I'|$, as desired.
	\end{claimproof}
    
	\begin{claim}\label{claim:braiding-leq}
	    It holds that $\alpha(G) \leq \alpha(G')$.
	\end{claim}
    \begin{claimproof}
		Let $I$ be an independent set in $G$, we will find an independent set $I'$ in $G'$ of the same size.
		If $I$ contains at most one of $b_1,e_1$, then we set $I' = \{ \phi(v) ~|~ v \in I \}$.
		Thus, suppose that $I$ contains both $b_1$ and $e_1$. In particular, this means that $a,c,d,f_1 \notin I$.
		We can we set $I' = \{ \phi(v) ~|~ v \in I\setminus \{b_1,e_1\} \} \cup \{b_2', e_2\}$.
		Again, it is straightforward to verify that in both cases $I'$ is independent in $G'$ and that $|I'|=|I|$, as desired.
    \end{claimproof}

Combining \cref{claim:braiding-leq} and \cref{claim:braiding-geq}, we obtain the desired equality, which completes the proof.
\end{proof}

\begin{figure}[t]	
	\centering
	\input{figbraid}
	\caption{An interchangeable pair described in \cref{lem:braiding}. $G_2$ can be viewed, up to isomorphism, as $G_1$ augmented by the orange vertex and edges. The boundary vertices are $a,c,d,g$ and only they might have neighbors outside the depicted subgraph.}
	\label{fig:braiding}
\end{figure}
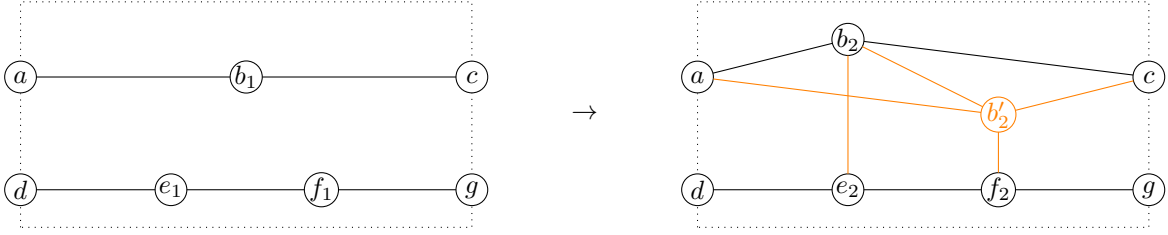

Now we proceed to construct permutation gadgets.
The following observation is straightforward.

\begin{observation}\label{obs:combine-permutations}
	Let $\ell \geq 1$ and let $\sigma, \sigma'$ be permutations of $[\ell]$.	
	Let $(\cP_\sigma, \overline{k}_\sigma)$ and $(\cP_{\sigma'}, \overline{k}_{\sigma'})$ be, respectively, permutation gadgets for $\sigma$ and $\sigma'$.
	Let $\cP$ be the coupling of $\cP_\sigma$ and $\cP_{\sigma'}$
	and let $\overline{k}=(k_{\sigma}^i+k_{\sigma'}^{\sigma(i)}-1)_{i=1}^\ell$. 
	Then, $(\cP,\overline{k})$ is a permutation gadget for $\sigma' \circ \sigma$.
\end{observation}

By \cref{lem:avoidable} and \cref{obs:combine-permutations}, it is sufficient to construct gadgets for some basic permutations and combine them via coupling.
We say that a permutation $\sigma: [\ell] \to [\ell]$ is an \emph{elementary swap} if it exchanges two consecutive elements, i.e., there is $i\in [\ell-1]$ such that for $j\in [\ell]\setminus \{i,i+1\}$, we have $\sigma(j)=j$, and $\sigma(i)=i+1$ and $\sigma(i+1)=i$.

It is well-known that any permutation can be obtained as a composition of $\Oh(\ell^2)$ elementary swaps -- this corresponds to the fact that the worst-case complexity of bubble sort is $\Oh(\ell^2)$.
Thus, for $H\in \{\abxba,\abccab\}$, it is sufficient to construct $H$-free permutation gadgets for elementary swaps.

\begin{proof}[Proof of \cref{lem:permutation-gadgets} for elementary swaps] \textit{Case: $H =$ \abxba.}
Let $\ell \geq 2$ and let $j \in [\ell-1]$ and $\sigma$ be an elementary swap that exchanges elements $j$ with $j+1$.

We will construct a \abxba-free permutation gadget $(\cP, \overline{k})$,
where $\cP = (P, (x_1, x_2, \ldots, x_\ell), (y_1, y_2, \ldots, y_\ell))$ is a train.
The vertex set of $P$ consists of the following four \emph{canonical} segments of vertices, in the described order:

\begin{enumerate}
\item[$V_0$] containing vertices $x_1, x_2, \ldots, x_\ell$ (which are the input of the gadget), 
\item[$V_1$] consisting of vertices $x_1^1, x_2^1, \ldots, x_{j-1}^1, u, v, x_j^1, x^1_{j+1}, \ldots, x_\ell^1$,  
\item[$V_2$] consisting of vertices $x_1^2, x_2^2, \ldots, x_{j-1}^2, w, x_{j+1}^2, x_j^2, x_{j+2}^2, \ldots, x_\ell^2$,
\item[$V_3$] consisting of vertices $x_1^3, x_2^3, \ldots, x_{j-1}^3, x_{j+1}^3, x_j^3, x_{j+2}^3,$ $\ldots, x_\ell^3$ (which are the output of the gadget, we identify them with corresponding vertices $y_i$).
  \end{enumerate}
 We connect the vertices by creating paths $x_i-x_i^1-x_i^2-x_i^3$ for each $i \neq j+1$, the path $x_{j+1}-x_{j+1}^1-u-w-x_{j+1}^2-x_{j+1}^3$, connecting $v$ to $x_j, u, x_j^1, x_j^2$, and adding the edge $x^1_jx^1_{j+1}$.
 This completes the construction of $\cP$.
 We also define $\overline{k}$ as a vectors that has 4 on every coordinate except for the $(j+1)$-th one, which is equal to 6. In particular, all values of $\overline{k}$ are even, as required.
 
 Let us verify that $(\cP, \overline{k})$ is indeed a \abxba-free permutation gadget for $\sigma$.
	
	\begin{claim}
		$(\cP, \overline{k})$ is a permutation gadget for $\sigma$.
	\end{claim}
	\begin{claimproof}
		Notice that $P$ is obtained from a linear forest $L$ by applying \cref{lem:braiding} with:
		\begin{align*}
		& a=x_j && b_1=b_2=x^1_j \\
		& d=x_{j+1}  && b'_2=v \\
		& c=x^2_j && e_1=e_2=x^1_{j+1} \\
		& g=w && f_1=f_2=u.
		\end{align*}
		Conversely, the desired linear forest $L$ can be obtained from $P$, up to isomorphism, by the removal of the vertex $v$ and edge $x_j^1 x^1_{j+1}$.
		It is straightforward to verify that all of its components are paths with endpoints $x_i, x_{\sigma(i)}$ for $i\in [\ell]$.
		Moreover, each such path has 4 vertices, except for the one for $i=j+1$, which has 6 vertices.
		(See also \cref{fig:perms}).
		\end{claimproof}

	\begin{claim}
		$\cP$ is \abxba-free.
	\end{claim}
		\begin{claimproof}
		Suppose, to the contrary, that some vertices $a, b, c, d, e$, in that order, induce \abxba.
		Consider the edge $bd$. It cannot be an edge between two canonical segments, as such edges are all maximal, contradicting the existence of edge $ae$.				
		The only remaining edges are $uv, ux_{j+1}^1, vx_j^1, x_j^1x_{j+1}^1$, and $wx_{j+1}^2$, and one can readily check that none of them participate in any induced \ac.
	\end{claimproof}
	
	This completes the proof of the lemma for $H = \abxba$.
\end{proof}

\begin{figure}[t]	
	\centering
	\input{figperm1}
	\input{figperm2}
	\caption{Permutation gadgets for elementary swaps used to prove \cref{lem:permutation-gadgets}.
		Boxes indicate canonical segments other than the input and output.
		Removal of the orange vertex and edges, by \cref{lem:braiding} yields a linear forest as in \cref{def:perm}. See also \cref{fig:braiding}.}
	\label{fig:perms}
\end{figure}
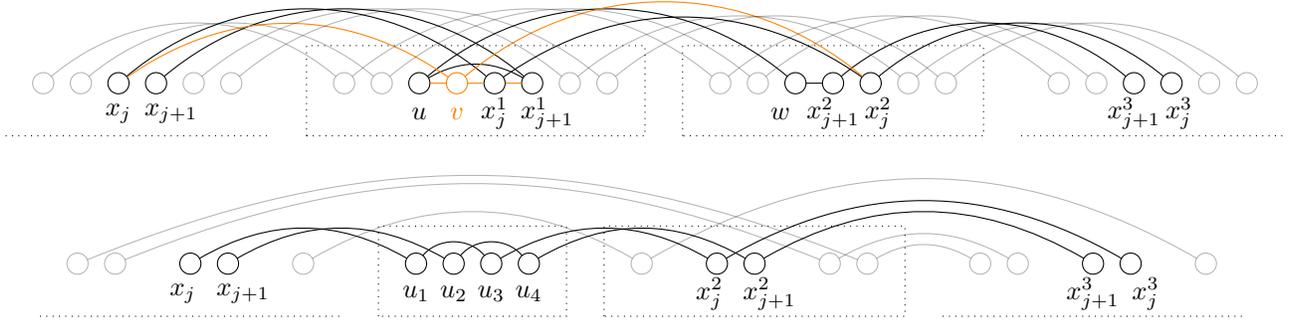

\begin{proof}[Proof of \cref{lem:permutation-gadgets} for elementary swaps] \textit{Case: $H =$ \abccab.}
Let $\ell \geq 2$ and let $j \in [\ell-1]$ and $\sigma$ be an elementary swap that exchanges elements $j$ with $j+1$.

We will construct a \abccab-free permutation gadget $(\cP, \overline{k})$,
where $\cP = (P, (x_1, x_2, \ldots, x_\ell), (y_1, y_2, \ldots, y_\ell))$ is a train.
Again, the vertex set of $P$ consists of the following four \emph{canonical} segments of vertices, in the described order:

\begin{enumerate}
\item[$V_0$] consisting of vertices $x_1, x_2, \ldots, x_\ell$ (which are input of $\cP$),
\item[$V_1$] consisting of vertices $u_1, u_2, u_3, u_4$,  
\item[$V_2$] consisting of vertices $x_\ell^2, x_{\ell-1}^2, \ldots, x_{j+2}^2, x_j^2, x_{j+1}^2, x_{j-1}^2, \ldots, x_1^2$,
\item[$V_3$] consisting of vertices $x_1^3, x_2^3, \ldots, x_{j-1}^3, x_{j+1}^3, x_j^3, x_{j+2}^3, \ldots, x_\ell^3$ (which are the output of $\cP$, we identify them with corresponding vertices $y_i$).
\end{enumerate}
The vertices are connected via paths $x_i-x_i^2-x_i^3$ for $i\neq j, j+1$, path $x_j-u_1-u_3-x_j^2-x_j^3$, and $x_{j+1}-u_2-u_4-x_{j+1}^2-x_{j+1}^3$.
This completes the construction of $\cP$.
We set $\overline{k}$ to the vector that has 3 on every coordinate except for the $j$-th and $(j+1)$-th one, which are equal to 5. In particular, all values of $\overline{k}$ are odd, as required.

 Let us verify that $(\cP, \overline{k})$ is indeed a \abccab-free permutation gadget for $\sigma$.
 It is clear that it is a permutation gadget, as it is already a linear forest of the required form. So, let us we focus on verifying that it is \abccab-free. We will actually prove a stronger statement, namely, that $\cP$ does not contains even \abccab as a subgraph.
	
	\begin{claim}
		$\cP$ is \abccab-subgraph-free.
	\end{claim}
	\begin{claimproof}		
		Suppose that the gadget contains \abccab as a subgraph, and let $a, b, c, d, e, f$ be its consecutive vertices.
		Consider the edge $cd$ and note that it is under a pair of crossing edges $ab$ and $ef$.
		However, any edge containedin $V_0 \cup V_1 \cup V_2$ can only be under some edges between $V_0$ and $V_2$, but such edges are pairwise non-crossing.
		Similarly, an edge between $V_2$ and $V_3$ can only be under some other edges between $V_2$ and $V_3$, which are likewise pairwise non-crossing.
		As these are all edges of $\cP$, we conclude that the edge $cd$ cannot exist, contradicting the existence of \abccab as a subgraph.
	\end{claimproof}
	
This completes the proof of the lemma for $H = \abccab$.
\end{proof}

With this, the proof of \cref{thm:chainred} is complete.

%% file: figopen1.tex
\begin{tikzpicture}[
        every node/.style={draw, circle, inner sep=0, outer sep = 0, minimum size=8pt}, baseline= -0.5mm]
            \draw[dotted] (3,0) -- (9,0);
            
            \node[opacity=.3] (s) at (.5, .5) {};
            \node[opacity=.3] (p) at (1, .5) {};
            \node (q) at (1.5, .5) {};

            \node[opacity=.3] (s1) at (3.5, .5) {};
            \node[opacity=.3] (s2) at (4, .5) {};
            \node[opacity=.3] (s3) at (4.5, .5) {};
            \node[opacity=.3] (p1) at (5.5, .5) {};
            \node[opacity=.3] (p2) at (6, .5) {};
            \node (q1) at (7, .5) {};
            \node (q2) at (7.5, .5) {};
            \node (q3) at (8, .5) {};
            \node (q4) at (8.5, .5) {};
            
            \path[draw, opacity=.3] (s) to[bend left = 30] (s1);
            \path[draw, opacity=.3] (s) to[bend left = 30] (s2);
            \path[draw, opacity=.3] (s) to[bend left = 30] (s3);

            \path[draw, opacity=.3] (p) to[bend left = 30] (p1);
            \path[draw, opacity=.3] (p) to[bend left = 30] (p2);

            \path[draw] (q) to[bend left = 30] (q1);
            \path[draw] (q) to[bend left = 30] (q2);
            \path[draw] (q) to[bend left = 30] (q3);
            \path[draw] (q) to[bend left = 30] (q4);
            
		\end{tikzpicture}

%% file: figclose1.tex
    \begin{tikzpicture}[
        every node/.style={draw, circle, inner sep=0, outer sep = 0, minimum size=8pt}, baseline= -0.5mm]
            \draw[dotted] (0,0) -- (4.5,0);
            
            \node[opacity=.3] (a) at (.5,.5) {};
            \node[opacity=.3] (b) at (1,.5) {};
            \node (c) at (2,.5) {};
            \node (d) at (2.5,.5) {};
            \node[opacity=.3] (e) at (3.5,.5) {};
            \node[opacity=.3] (f) at (4,.5) {};

            \node[opacity=.3] (a') at (5.5,.5) {};
            \node[opacity=.3] (b') at (6,.5) {};
            \node (c') at (7,.5) {};
            \node (d') at (7.5,.5) {};
            \node[opacity=.3] (e') at (8.5,.5) {};
            \node[opacity=.3] (f') at (9,.5) {};
            
            \path[draw, opacity=.3] (a) to[bend left = 30] (a');
            \path[draw, opacity=.3] (b) to[bend left = 30] (b');
            \path[draw] (c) to[bend left = 30] (c');
            \path[draw] (d) to[bend left = 30] (d');
            \path[draw, opacity=.3] (e) to[bend left = 30] (e');
            \path[draw, opacity=.3] (f) to[bend left = 30] (f');

            \draw[opacity=.3] (a')--(b');
            \draw (c')--(d');
            \draw[opacity=.3] (e')--(f');
            
		\end{tikzpicture}

%% file: figclose2.tex
        \begin{tikzpicture}[
        every node/.style={draw, circle, inner sep=0, outer sep = 0, minimum size=8pt}, baseline= -0.5mm]
            \draw[dotted] (0,0) -- (4.5,0);
            
            \node[opacity=.3] (a) at (.5,.5) {};
            \node[opacity=.3] (b) at (1,.5) {};
            \node (c) at (2,.5) {};
            \node (d) at (2.5,.5) {};
            \node[opacity=.3] (e) at (3.5,.5) {};
            \node[opacity=.3] (f) at (4,.5) {};

            \node[opacity=.3] (f') at (5.5,.5) {};
            \node[opacity=.3] (e') at (6,.5) {};
            \node (d') at (7,.5) {};
            \node (c') at (7.5,.5) {};
            \node[opacity=.3] (b') at (8.5,.5) {};
            \node[opacity=.3] (a') at (9,.5) {};
            
            \path[draw, opacity=.3] (a) to[bend left = 22] (a');
            \path[draw, opacity=.3] (b) to[bend left = 22] (b');
            \path[draw] (c) to[bend left = 22] (c');
            \path[draw] (d) to[bend left = 22] (d');
            \path[draw, opacity=.3] (e) to[bend left = 22] (e');
            \path[draw, opacity=.3] (f) to[bend left = 22] (f');

            \draw[opacity=.3] (a')--(b');
            \draw (c')--(d');
            \draw[opacity=.3] (e')--(f');
            
		\end{tikzpicture}

%% file: figbraid.tex
\begin{tikzpicture}[
	every node/.style={draw, circle, inner sep=0, outer sep = 0, minimum size=12pt}, baseline= -0.5mm]

    \draw[dotted] (0, -0.5) rectangle (6, 2.5);
    
	\node[fill, white] (a) at (0, 1.5) {};
	\node (a') at (0, 1.5) {$a$};
	\node (b) at (3, 1.5) {$b_1$};
	\node[fill, white] (c) at (6, 1.5) {};
	\node (c') at (6, 1.5) {$c$};
	\node[fill, white] (d) at (0, 0) {};
	\node (d') at (0, 0) {$d$};
	\node (e) at (2, 0) {$e_1$};
	\node (f) at (4,0) {$f_1$};
	\node[fill, white] (g) at (6,0) {};
	\node (g') at (6,0) {$g$};
	
	\draw (a)--(b)--(c);
	\draw (d)--(e)--(f)--(g);
	
	\node[draw=none] (x) at (7.5, 1) {$\rightarrow$};

    \draw[dotted] (9, -.5) rectangle (15, 2.5);
	
	\node[fill, white] (a1) at (9, 1.5) {};
	\node (a1') at (9, 1.5) {$a$};
	\node (b1) at (11, 2) {$b_2$};
	\node[orange] (b1') at (13, 1) {$b_2'$};
	\node[fill, white] (c1) at (15, 1.5) {};
	\node (c1') at (15, 1.5) {$c$};
	\node[fill, white] (d1) at (9, 0) {};
	\node (d1') at (9, 0) {$d$};
	\node (e1) at (11, 0) {$e_2$};
	\node (f1) at (13,0) {$f_2$};
	\node[fill, white] (g1) at (15,0) {};
	\node (g1') at (15,0) {$g$};
	
	\draw (a1)--(b1)--(c1);
	\draw (d1)--(e1)--(f1)--(g1);
	\draw[orange] (a1)--(b1')--(c1);
	\draw[orange] (f1)--(b1')--(b1)--(e1);
	
\end{tikzpicture}

%% file: figperm1.tex
\begin{tikzpicture}[
        every node/.style={draw, circle, inner sep=0, outer sep = 0, minimum size=8pt}, baseline= -0.5mm]
            \draw[dotted] (0, -.2) -- (3.5, -.2);
            \draw[dotted] (4, -.2) rectangle (8.5,1);
            \draw[dotted] (9, -.2) rectangle (13,1);
            \draw[dotted] (13.5, -.2) -- (17,-.2);
            
            \node[opacity=.3] (a1) at (.5, .5) {};
            \node[opacity=.3] (b1) at (1, .5) {};
            \node (c1) at (1.5, .5) {};
            \node[draw=none,fill=none] (c1') at (1.5, .1) {$x_j$};
            \node (d1) at (2, .5) {};
            \node[draw=none,fill=none] (d1') at (2.2, .1) {$x_{j+1}$};
            \node[opacity=.3] (e1) at (2.5, .5) {};
            \node[opacity=.3] (f1) at (3, .5) {};
            
            \node[opacity=.3] (a2) at (4.5, .5) {};
            \node[opacity=.3] (b2) at (5, .5) {};
            \node (s) at (5.5, .5) {};
            \node[draw=none,fill=none] (s') at (5.5, .1) {$u$};
            \node[orange] (p) at (6, .5) {};
            \node[draw=none,fill=none,orange] (p') at (6, .1) {$v$};
            \node (q) at (6.5, .5) {};
            \node[draw=none,fill=none] (q') at (6.5, .1) {$x_j^1$};
            \node (r) at (7, .5) {};
            \node[draw=none,fill=none] (r') at (7.2, .1) {$x_{j+1}^1$};
            \node[opacity=.3] (e2) at (7.5, .5) {};
            \node[opacity=.3] (f2) at (8, .5) {};
            
            \node[opacity=.3] (a3) at (9.5, .5) {};
            \node[opacity=.3] (b3) at (10, .5) {};
            \node (t) at (10.5, .5) {};
            \node[draw=none,fill=none] (r') at (10.3, .1) {$w$};
            \node (c3) at (11, .5) {};
            \node[draw=none,fill=none] (r') at (11, .1) {$x_{j+1}^2$};
            \node (d3) at (11.5, .5) {};
            \node[draw=none,fill=none] (r') at (11.6, .1) {$x_{j}^2$};
            \node[opacity=.3] (e3) at (12, .5) {};
            \node[opacity=.3] (f3) at (12.5, .5) {};

            \node[opacity=.3] (a4) at (14, .5) {};
            \node[opacity=.3] (b4) at (14.5, .5) {};
            \node (c4) at (15, .5) {};
            \node[draw=none,fill=none] (r') at (15, .1) {$x_{j+1}^3$};
            \node (d4) at (15.5, .5) {};
            \node[draw=none,fill=none] (r') at (15.6, .1) {$x_{j}^3$};
            \node[opacity=.3] (e4) at (16, .5) {};
            \node[opacity=.3] (f4) at (16.5, .5) {};
            
            \draw[orange] (s)--(p)--(q)--(r);
            \draw (t)--(c3);
            
            \path[draw] (c1) to[bend left = 40] (q);
            \path[draw] (q) to[bend left = 35] (d3);
            \path[draw] (d3) to[bend left = 40] (d4);
            \path[draw] (c3) to[bend left = 40] (c4);
            \path[draw] (d1) to[bend left = 40] (r);
            \path[draw] (r) to[bend right = 30] (s);
            \path[draw] (s) to[bend left = 40] (t);
            
            \path[draw, orange] (c1) to[bend left = 35] (p);
            \path[draw, orange] (p) to[bend left = 40] (d3);
            
            \path[draw, opacity=.3] (a1) to[bend left = 40] (a2);
            \path[draw, opacity=.3] (a2) to[bend left = 40] (a3);
            \path[draw, opacity=.3] (a3) to[bend left = 40] (a4);

            \path[draw, opacity=.3] (b1) to[bend left = 40] (b2);
            \path[draw, opacity=.3] (b2) to[bend left = 40] (b3);
            \path[draw, opacity=.3] (b3) to[bend left = 40] (b4);

            \path[draw, opacity=.3] (e1) to[bend left = 40] (e2);
            \path[draw, opacity=.3] (e2) to[bend left = 40] (e3);
            \path[draw, opacity=.3] (e3) to[bend left = 40] (e4);

            \path[draw, opacity=.3] (f1) to[bend left = 40] (f2);
            \path[draw, opacity=.3] (f2) to[bend left = 40] (f3);
            \path[draw, opacity=.3] (f3) to[bend left = 40] (f4);
            
\end{tikzpicture}

%% file: figperm2.tex
\begin{tikzpicture}[
        every node/.style={draw, circle, inner sep=0, outer sep = 0, minimum size=8pt}, baseline= -0.5mm]
            \draw[dotted] (0,-.2) -- (4,-.2);
            \draw[dotted] (12,-.2) -- (16,-.2);
            \draw[dotted] (4.5,-.2) rectangle (7,1);
            \draw[dotted] (7.5,-.2) rectangle (11.5,1);
            
            \node[opacity=.3] (a) at (.5, .5) {};
            \node[opacity=.3] (b) at (1, .5) {};
            \node (c) at (2, .5) {};
            \node[draw=none,fill=none] (c') at (1.9, .1) {$x_j$};
            \node (d) at (2.5, .5) {};
            \node[draw=none,fill=none] (d') at (2.7, .1) {$x_{j+1}$};
            \node[opacity=.3] (e) at (3.5, .5) {};

            \node (s) at (5, .5) {};
            \node[draw=none,fill=none] (s') at (5, .1) {$u_1$};
            \node (p) at (5.5, .5) {};
            \node[draw=none,fill=none] (p') at (5.5, .1) {$u_2$};
            \node (q) at (6, .5) {};
            \node[draw=none,fill=none] (q') at (6, .1) {$u_3$};
            \node (r) at (6.5, .5) {};
            \node[draw=none,fill=none] (r') at (6.5, .1) {$u_4$};

            \node[opacity=.3] (e1) at (8, .5) {};
            \node (c1) at (9, .5) {};
            \node[draw=none,fill=none] (c1') at (8.9, .1) {$x_j^2$};
            \node (d1) at (9.5, .5) {};
            \node[draw=none,fill=none] (d1') at (9.7, .1) {$x_{j+1}^2$};
            \node[opacity=.3] (b1) at (10.5, .5) {};
            \node[opacity=.3] (a1) at (11, .5) {};

            \node[opacity=.3] (a2) at (12.5, .5) {};
            \node[opacity=.3] (b2) at (13, .5) {};
            \node (d2) at (14, .5) {};
            \node[draw=none,fill=none] (d2') at (14, .1) {$x_{j+1}^3$};
            \node (c2) at (14.5, .5) {};
            \node[draw=none,fill=none] (c2') at (14.7, .1) {$x_j^3$};
            \node[opacity=.3] (e2) at (15.5, .5) {};
            
            \path[draw, opacity=.3] (a) to[bend left = 22] (a1);
            \path[draw, opacity=.3] (b) to[bend left = 22] (b1);
            \path[draw, opacity=.3] (e) to[bend left = 30] (e1);

            \path[draw, opacity=.3] (a1) to[bend left = 30] (a2);
            \path[draw, opacity=.3] (b1) to[bend left = 30] (b2);
            \path[draw, opacity=.3] (e1) to[bend left = 30] (e2);

            \path[draw] (c) to[bend left = 30] (s);
            \path[draw] (d) to[bend left = 30] (p);
            \path[draw] (q) to[bend left = 30] (c1);
            \path[draw] (r) to[bend left = 30] (d1);

            \path[draw] (s) to[bend left = 50] (q);
            \path[draw] (p) to[bend left = 50] (r);
            \path[draw] (c1) to[bend left =30] (c2);
            \path[draw] (d1) to[bend left =30] (d2);
            
		\end{tikzpicture}

%% file: hardness-summary.tex
\subsection{Summary: excluding a single induced subgraph}

Our next goal is to summarize the hardness results achieved throughout the chapter.

\begin{theorem}\label{thm:hard}
	\MIS is \NP-hard in $H$-free ordered graphs if $H$ has at least three edges or if it contains \abxba,\ \bad, \adb, \abd, \abnce, or the mirror image of any of these as an induced subgraph.
\end{theorem}
\begin{proof}
	Observe that if $H$ contains \abxba, \bad, \adb, \abd, \abnce or the mirror image of any of these, the conclusion follows, respectively, from \cref{thm:chainred}, \cref{3sat} and \cref{xxxlrlrlr} for the last three.
	What remains is to prove \MIS is \NP-hard in $H$-free ordered graphs whenever $H$ contains at least three edges.
	 
	Let us fix any ordered graph $H$ and suppose that \MIS is not \NP-hard in the class of $H$-free ordered graphs.
	We will prove that $H$ contains at most two edges, thus proving the theorem.
	
	First recall from \cref{sec:intro} that if the underlying graph of $H$ is not in the family $\cS$ (forests of subdivided claws), then \MIS is \NP-hard already in the unordered case and thus also in the ordered case~\cite{alekseev1982effect}.
	Thus, we might assume otherwise. In particular, $H$ is a subcubic forest.

	\begin{claim}\label{claim:linear-forest}
		$H$ is a linear forest, i.e., every component of $H$ is a path.
	\end{claim}
	\begin{claimproof}
    Suppose $H$ has a vertex $v$ of degree 3.
    Since $H$ is a forest, the neighbors of $v$ form an independent set, and thus $N[v]$ induces either  \adbp, or \abdp, or the mirror image of one of these.
 	Then $H$ contains as a subgraph either \adb or \abd (or the mirror image thereof), which, by \cref{xxxlrlrlr}, contradicts the choice of $H$.
	\end{claimproof}

	\begin{claim}\label{claim:summary-components}
		Every component of $H$ has at most 3 vertices.
	\end{claim}
   \begin{claimproof}
		Suppose otherwise. By \cref{claim:linear-forest}, this means that $H$ contains an induced path with 4 vertices.
	   	Such a path can be ordered in 12 ways (up to isomorphism), and disregarding mirror images, there are 8 possible orderings to consider: \badc, \abcd, \acbd, \adcb, \cadb, \adbc, \cabd, and \abdc.
	   	Observe that $H$ cannot contain:
	   	\begin{itemize}
			\item \badc, explicitly by \cref{rrrxxxlll},
			\item \abcd, \acbd, \adcb, explicitly by \cref{xxxlrlrlr},
			\item \cadb nor \adbc by \cref{xxxlrlrlr}, as they both contain \adb as a subgraph,
			\item \cabd nor \abdc by \cref{xxxlrlrlr}, as they both contain \abd as a subgraph.
	   	\end{itemize}
	   	Since we excluded all possible orderings of a $4$-vertex path, this completes the proof of the claim.
   \end{claimproof}
    
    \begin{claim}\label{niema3k2}
    	The underlying graph of $H$ does not contain an induced three-edge matching.
    \end{claim}
    \begin{claimproof}
    	We eliminate all possibilities. The three-edge matching can be ordered in 15 ways (up to isomorphism), and disregarding mirror images, only 11 remain. Observe that $H$ cannot contain:
    	\begin{itemize}
    		\item \abbcac, \abacbc, \aabccb and \aabcbc by \cref{xxxlrlrlr} as they all contain \abnce as a subgraph,
    		\item \abccba, \abcbca nor \abcbac by \cref{thm:chainred} as they all contain an induced \abxba,
    		\item \abbcca nor \aabbcc explicitly by \cref{xxxlllrrr},
    		\item \abccab explicitly by \cref{thm:chainred},
    		\item \abcabc explicitly by \cref{squstr}.
    	\end{itemize}
    	It is easy (yet a bit tedious) to verify that the list above is complete.
		Thus, this completes the proof of the claim.
    \end{claimproof}
    
    \begin{claim}\label{niemap3k2}
    	The underlying graph of $H$ does not contain an induced $P_3+K_2$, i. e., the graph consisting of two components, one being the three-vertex path $P_3$ and the other a $K_2$.
    \end{claim}
    
    \begin{claimproof}
    Suppose that $H$ contains a copy of $P_3+K_2$ as an induced subgraph, and let $u,v,w,x,y$ be its vertices such that there are edges $uv,vw,xy$.
    We consider the cases depending on the ordering of $u,v,w$: (i) $u\prec v\prec w$ and (ii) $u\prec w \prec v$ -- all the remaining cases are symmetric.
    First suppose that $u\prec v\prec w$.
    If any of $x,y$ is between $u$ and $v$ or between $v$ and $w$, then $H$ contains an induced copy of \abd or its mirror image,  which is a contradiction by \cref{xxxlrlrlr}.
   Therefore, the set $\{u,v,w,x,y\}$ induces either \aenbcd, or \abncde, or its mirror image, which is again a contradiction by \cref{xxxlrlrlr}.
   So now suppose that $u\prec w \prec v$.
   If any of $x,y$ is between $u$ and $w$, then $H$ contains an induced \adc, which is a contradiction by \cref{3sat}.
   Furthermore, if any of $x,y$ is between $w$ and $v$, then $H$ contains an induced \adb, which is a contradiction by \cref{xxxlrlrlr}.
   Therefore, the set $\{u,v,w,x,y\}$ induces either \aenbdc or \acbnde or \abnced, which contradicts, respectively, \cref{xxxrlrlrl}, \cref{xxxlllrrr}, and \cref{xxxlrlrlr} as \abnced contains \abnce as a subgraph.
    \end{claimproof} 
    
  So now suppose that $H$ has at least three edges.
  Recall that by \cref{claim:linear-forest} and \cref{claim:summary-components}, $H$ is a linear forest, whose every component has at most $3$ vertices.
  By \cref{niema3k2}, there are at most two components that are not isolated vertices.
 So it is only possible when one component induces a $P_3$, and the other induces a $K_2$.
 However, this contradicts \cref{niemap3k2}, which completes the proof.    
\end{proof}

To conclude, let us argue that \cref{thm:hard} implies the hardness part of \cref{thm:main}.
This is equivalent to the following statement (note that \extoneedgek is an induced subgraph of, say, \extababk).

\begin{theorem}[Hardness part of \cref{thm:main}, restated]
  Let $H$ be an ordered graph.
  If $H$ is not an induced subgraph of any graph in
  \[
  \bigcup_{k \geq 0} \{ \extpthree, \extchord, \extchordrev, \extaakbb, \extababk, \extabbak \},
   \]
  then \MIS is \NP-hard in $H$-free ordered graphs.  
\end{theorem}
\begin{proof}
Let $k = |V(H)|$.
If $H$ has at most one edge, it is an induced subgraph of \extoneedgek and thus of \extababk.

Next, consider the case that $H$ has exactly two edges and let $H'$ be the subgraph of $H$ induced by non-isolated vertices.
Note that the underlying graph of $H'$ is either $P_3$ or $2K_2$.
Consequently, $H'$ is isomorphic to one of the following ordered graphs: \abc, \acb, \bac, \aabb, \abab, \abba.
We consider these cases one by one, all containments are with respect to the induced subgraph relation.
\begin{itemize}
\item If $H$ contains \abc  but is not contained in \extpthree, then it must contain \abd or its mirror image.
\item If $H$ contains \bac, but is not contained in \extchord, it must contain either \cad or \bad.
\item The case that $H$ contains \acb is symmetric to the previous one.
\item If $H$ contains \aabb, but is not contained in \extaakbb, then it must contain \abnce or its mirror image.
\item If $H$ contains \abab, then it is contained in \extababk.
\item If $H$ contains \abba, but is not contained in \extabbak, then it must contain \abxba.
\end{itemize}
In each of these cases \cref{thm:hard} implies that \MIS is \NP-hard in $H$-free ordered graphs.

Lastly, whenever $H$ has at least three edges, then hardness in $H$-free graphs  follows directly from \cref{thm:hard}.
This completes the proof.
\end{proof}

%% file: outro.tex
Let us conclude the paper by pointing out some directions for future research.
Note that the hardness part of \cref{thm:main} is close to a dichotomy into cases solvable in quasipolynomial time and those that are \NP-hard.
Indeed, the (quasi)polynomial-time algorithms given by \cref{thm:main}~(1) and (2) are a strong indication that the corresponding cases of \MWIS in $H$-free ordered graphs are not \NP-hard, since otherwise every problem in \NP would admit a (quasi)polynomial-time algorithm.
By contrast, the subexponential-time algorithm for $H = \extabbak$ (for $k \geq 0$) from \cref{thm:main}~(3) does not provide comparable evidence against \NP-hardness.
Thus, it would be very interesting to decide whether \MWIS can be solved in quasipolynomial time, or even in polynomial time, in \abbak-free ordered graphs, or whether the problem is \NP-hard in this case.
We conjecture that the algorithmic side holds.

Another natural direction is to strengthen the hardness side to ETH-lower bounds that exclude subexponential-time algorithms.
At present, the reductions in \cref{sec:longsubdiv,sec:chainreduction} (for $H \in \{ \abxba, \abcabc, \abccab,  \abcbca\}$) introduce a super-linear blow-up in the size of the constructed instance.
It would be interesting either to replace them with tighter reductions or to design subexponential-time algorithms for those cases.
We conjecture that the former is possible.

Finally, we believe that the ordered setting is promising also for problems beyond independent set and coloring.
To the best of our knowledge, apart from the current paper and recent work on (\textsc{List}) $k$-\textsc{Coloring}~\cite{DBLP:conf/stacs/PiecykR26,DBLP:journals/siamdm/HajebiLS24} such algorithmic questions in hereditary classes of ordered graphs have not been considered so far.
In particular, it would be interesting to investigate \textsc{Feedback Vertex Set} and \textsc{Odd Cycle Transversal} in hereditary classes of ordered graphs.
Both problems can be seen as generalizations of \MWIS.
\textsc{Feedback Vertex Set} can be equivalently stated as finding a maximum size (or weight) induced forest, i.e., an induced subgraph of treewidth at most 1 (while an independent set is an induced subgraph of treewidth 0).
\textsc{Odd Cycle Transversal} is equivalent to finding a maximum size (or weight) induced bipartite subgraph, i.e., an induced subgraph of chromatic number at most 2 (while an independent set is an induced subgraph of chromatic number 1).
Both problems are among the best studied in hereditary classes of unordered graphs~\cite{Tara,DBLP:conf/stoc/GartlandLPPR21,DBLP:conf/soda/ChudnovskyMPPR24,DBLP:journals/siamdm/PaesaniPR22,DBLP:journals/algorithmica/DabrowskiFJPPR20,DBLP:journals/talg/0001LLR0S25,DBLP:journals/siamdm/ChudnovskyKPRS21,DBLP:conf/esa/GalbyLMN25}, so they seem like natural next targets in the ordered setting.


%% file: main.bib
@article{DBLP:journals/jcss/ImpagliazzoPZ01,
  author       = {Russell Impagliazzo and
                  Ramamohan Paturi and
                  Francis Zane},
  title        = {Which Problems Have Strongly Exponential Complexity?},
  journal      = {J. Comput. Syst. Sci.},
  volume       = {63},
  number       = {4},
  pages        = {512--530},
  year         = {2001},
  url          = {https://doi.org/10.1006/jcss.2001.1774},
  doi          = {10.1006/JCSS.2001.1774},
  bibsource    = {dblp computer science bibliography, https://dblp.org}
}

@book{Golumbic2004,
  author    = {Golumbic, Martin C.},
  title     = {Algorithmic Graph Theory and Perfect Graphs},
  publisher = {Elsevier/North-Holland},
  year      = {2004},
  edition   = {2nd},
  isbn      = {978-0-444-51530-8}, 
  series    = {Annals of Discrete Mathematics},
  volume    = {57},
}

@article{DBLP:journals/siamcomp/Gavril72,
  author       = {Fanica Gavril},
  title        = {Algorithms for Minimum Coloring, Maximum Clique, Minimum Covering
                  by Cliques, and Maximum Independent Set of a Chordal Graph},
  journal      = {{SIAM} J. Comput.},
  volume       = {1},
  number       = {2},
  pages        = {180--187},
  year         = {1972},
  url          = {https://doi.org/10.1137/0201013},
  doi          = {10.1137/0201013},
  bibsource    = {dblp computer science bibliography, https://dblp.org}
}

@article{DBLP:journals/siamdm/HajebiLS24,
  author       = {Sepehr Hajebi and
                  Yanjia Li and
                  Sophie Spirkl},
  title        = {{List-3-Coloring} Ordered Graphs with a Forbidden Induced Subgraph},
  journal      = {{SIAM} J. Discret. Math.},
  volume       = {38},
  number       = {1},
  pages        = {1158--1190},
  year         = {2024},
  url          = {https://doi.org/10.1137/22m1515768},
  doi          = {10.1137/22M1515768},
  bibsource    = {dblp computer science bibliography, https://dblp.org}
}

@inproceedings{DBLP:conf/sosa/PilipczukPR21,
  author       = {Marcin Pilipczuk and
                  Micha\l{} Pilipczuk and
                  Pawe\l{} Rz\k{a}\.ewski},
  editor       = {Hung Viet Le and
                  Valerie King},
  title        = {Quasi-polynomial-time algorithm for Independent Set in {$P_t$}-free
                  graphs via shrinking the space of induced paths},
  booktitle    = {4th Symposium on Simplicity in Algorithms, {SOSA} 2021, Virtual Conference,
                  January 11-12, 2021},
  pages        = {204--209},
  publisher    = {{SIAM}},
  year         = {2021},
  url          = {https://doi.org/10.1137/1.9781611976496.23},
  doi          = {10.1137/1.9781611976496.23},
  bibsource    = {dblp computer science bibliography, https://dblp.org}
}

@book{LovaszPlummer1986_MatchingTheory,
  author    = {László Lovász and Michael D. Plummer},
  title     = {Matching Theory},
  publisher = {North–Holland},
  year      = {1986},
  series    = {North–Holland Mathematics Studies},
  volume    = {121},
  pages     = {544},
  address   = {Amsterdam – New York – Oxford – Tokyo},
  isbn      = {9780444879165},
  language  = {English}
}

@article{DBLP:journals/siamdm/DebskiPR22,
  author       = {Micha\l{} D\k{e}bski and
                  Marta Piecyk and
                  Pawe\l{} Rz\k{a}\.zewski},
  title        = {Faster {3-Coloring} of Small-Diameter Graphs},
  journal      = {{SIAM} J. Discret. Math.},
  volume       = {36},
  number       = {3},
  pages        = {2205--2224},
  year         = {2022},
  url          = {https://doi.org/10.1137/21m1447714},
  doi          = {10.1137/21M1447714},
  bibsource    = {dblp computer science bibliography, https://dblp.org}
}

@article{Po74,
author={Svatopluk Poljak},
title={A note on stable sets and colorings of graphs},
journal={Commentationes Mathematicae Universitatis Carolinae},
volume={15},
issue={2},
pages={307--309},
year={1974}
}

@article{GAREY1976237,
title = "Some simplified {NP}-complete graph problems",
journal = "Theoretical Computer Science",
volume = "1",
number = "3",
pages = "237--267",
year = "1976",
issn = "0304-3975",
doi = "https://doi.org/10.1016/0304-3975(76)90059-1",
url = "http://www.sciencedirect.com/science/article/pii/0304397576900591",
author = "Michael R. Garey and David S. Johnson and Larry J. Stockmeyer",
}

@misc{Walczakordered,
        title = {Coloring ordered graphs with excluded induced ordered matchings},
        year = {2022},
        howpublished = {Banff International Research Station. 22w5009: Extremal Combinatorics and Geometry. \url{https://www.birs.ca/events/2022/5-day-workshops/22w5009/videos/watch/202208161436-Walczak.html}},
        author = {Marcin Briański and James Davies and Bartosz Walczak},
    }

@article{DBLP:journals/combinatorica/AxenovichRU18,
  author       = {Maria Axenovich and
                  Jonathan Rollin and
                  Torsten Ueckerdt},
  title        = {Chromatic Number of Ordered Graphs with Forbidden Ordered Subgraphs},
  journal      = {Comb.},
  volume       = {38},
  number       = {5},
  pages        = {1021--1043},
  year         = {2018},
  url          = {https://doi.org/10.1007/s00493-017-3593-0},
  doi          = {10.1007/S00493-017-3593-0},
  bibsource    = {dblp computer science bibliography, https://dblp.org}
}

@article{DBLP:journals/jctb/PachT21,
  author       = {J{\'{a}}nos Pach and
                  Istv{\'{a}}n Tomon},
  title        = {{Erd{\H{o}}s-Hajnal}-type results for monotone paths},
  journal      = {J. Comb. Theory {B}},
  volume       = {151},
  pages        = {21--37},
  year         = {2021},
  url          = {https://doi.org/10.1016/j.jctb.2021.05.004},
  doi          = {10.1016/J.JCTB.2021.05.004},
  bibsource    = {dblp computer science bibliography, https://dblp.org}
}

@inproceedings{DBLP:conf/stacs/PiecykR26,
  author       = {Marta Piecyk and
                  Pawe\l{} Rz\k{a}\.zewski},
  editor       = {Meena Mahajan and
                  Florin Manea and
                  Annabelle McIver and
                  Kim Thang Nguyen},
  title        = {List Coloring Ordered Graphs with Forbidden Induced Subgraphs},
  booktitle    = {43rd International Symposium on Theoretical Aspects of Computer Science,
                  {STACS} 2026, Grenoble, France, March 9-13, 2026},
  series       = {LIPIcs},
  pages        = {74:1--74:17},
  publisher    = {Schloss Dagstuhl - Leibniz-Zentrum f{\"{u}}r Informatik},
  year         = {2026},
  url          = {https://doi.org/10.4230/LIPIcs.STACS.2026.74},
  doi          = {10.4230/LIPICS.STACS.2026.74},
  bibsource    = {dblp computer science bibliography, https://dblp.org}
}

@article{DBLP:journals/siamcomp/HeathR92,
  author       = {Lenwood S. Heath and
                  Arnold L. Rosenberg},
  title        = {Laying out Graphs Using Queues},
  journal      = {{SIAM} J. Comput.},
  volume       = {21},
  number       = {5},
  pages        = {927--958},
  year         = {1992},
  url          = {https://doi.org/10.1137/0221055},
  doi          = {10.1137/0221055},
  bibsource    = {dblp computer science bibliography, https://dblp.org}
}

@article{ChartrandHarary1967,
  author    = {Gary Chartrand and Frank Harary},
  title     = {Planar permutation graphs},
  journal   = {Annales de l'institut Henri Poincar{\'e}},
  volume    = {6},
  number    = {4},
  pages     = {433--438},
  year      = {1967}
}

@article{BernhartKainen1979,
  author    = {Frank R. Bernhart and Paul C. Kainen},
  title     = {The book thickness of a graph},
  journal   = {Journal of Combinatorial Theory, Series B},
  volume    = {27},
  number    = {3},
  pages     = {320--331},
  year      = {1979},
  doi       = {10.1016/0095-8956(79)90021-2}
}

@article{DBLP:journals/algorithmica/BacsoLMPTL19,
  author       = {G{\'{a}}bor Bacs{\'{o}} and
                  Daniel Lokshtanov and
                  D{\'{a}}niel Marx and
                  Marcin Pilipczuk and
                  Zsolt Tuza and
                  Erik Jan van Leeuwen},
  title        = {Subexponential-Time Algorithms for Maximum Independent Set in {$P_t$}-Free and Broom-Free Graphs},
  journal      = {Algorithmica},
  volume       = {81},
  number       = {2},
  pages        = {421--438},
  year         = {2019},
  url          = {https://doi.org/10.1007/s00453-018-0479-5},
  doi          = {10.1007/S00453-018-0479-5},
  bibsource    = {dblp computer science bibliography, https://dblp.org}
}

@book{platypus,
  author       = {Marek Cygan and
                  Fedor V. Fomin and
                  Lukasz Kowalik and
                  Daniel Lokshtanov and
                  D{\'{a}}niel Marx and
                  Marcin Pilipczuk and
                  Michal Pilipczuk and
                  Saket Saurabh},
  title        = {Parameterized Algorithms},
  publisher    = {Springer},
  year         = {2015},
  url          = {https://doi.org/10.1007/978-3-319-21275-3},
  doi          = {10.1007/978-3-319-21275-3},
  isbn         = {978-3-319-21274-6},
  bibsource    = {dblp computer science bibliography, https://dblp.org}
}

@article{DBLP:journals/talg/GrzesikKPP22,
  author       = {Andrzej Grzesik and
                  Tereza Klimo\v{s}ov{\'{a}} and
                  Marcin Pilipczuk and
                  Micha\l{} Pilipczuk},
  title        = {Polynomial-time Algorithm for Maximum Weight Independent Set on {$P_6$}-free
                  Graphs},
  journal      = {{ACM} Trans. Algorithms},
  volume       = {18},
  number       = {1},
  pages        = {4:1--4:57},
  year         = {2022},
  url          = {https://doi.org/10.1145/3414473},
  doi          = {10.1145/3414473},
  bibsource    = {dblp computer science bibliography, https://dblp.org}
}

@article{DBLP:journals/jda/LozinM08,
  author       = {Vadim V. Lozin and
                  Martin Milani\v{c}},
  title        = {A polynomial algorithm to find an independent set of maximum weight
                  in a fork-free graph},
  journal      = {J. Discrete Algorithms},
  volume       = {6},
  number       = {4},
  pages        = {595--604},
  year         = {2008},
  url          = {https://doi.org/10.1016/j.jda.2008.04.001},
  doi          = {10.1016/J.JDA.2008.04.001},
  bibsource    = {dblp computer science bibliography, https://dblp.org}
}

@inproceedings{DBLP:conf/coco/Karp72,
  author       = {Richard M. Karp},
  editor       = {Raymond E. Miller and
                  James W. Thatcher},
  title        = {Reducibility Among Combinatorial Problems},
  booktitle    = {Proceedings of a symposium on the Complexity of Computer Computations,
                  held March 20-22, 1972, at the {IBM} Thomas J. Watson Research Center,
                  Yorktown Heights, New York, {USA}},
  series       = {The {IBM} Research Symposia Series},
  pages        = {85--103},
  publisher    = {Plenum Press, New York},
  year         = {1972},
  url          = {https://doi.org/10.1007/978-1-4684-2001-2\_9},
  doi          = {10.1007/978-1-4684-2001-2\_9},
  bibsource    = {dblp computer science bibliography, https://dblp.org}
}

@article{Hastad96cliqueis,
  author       = {Johan H{\aa}stad},
  title        = {Clique is hard to approximate within $n^{{(1-\epsilon)}}$},
  journal      = {Electron. Colloquium Comput. Complex.},
  volume       = {{TR97}},
  eid          = {{TR97-038}},
  year         = {1997},
  url          = {https://eccc.weizmann.ac.il/eccc-reports/1997/TR97-038/index.html},
  eprinttype   = {ECCC},
  eprint       = {TR97-038},
  bibsource    = {dblp computer science bibliography, https://dblp.org}
}

@article{DBLP:journals/siamcomp/ChalermsookCKLM20,
  author       = {Parinya Chalermsook and
                  Marek Cygan and
                  Guy Kortsarz and
                  Bundit Laekhanukit and
                  Pasin Manurangsi and
                  Danupon Nanongkai and
                  Luca Trevisan},
  title        = {From {Gap-Exponential Time Hypothesis} to Fixed Parameter Tractable
                  Inapproximability: Clique, Dominating Set, and More},
  journal      = {{SIAM} J. Comput.},
  volume       = {49},
  number       = {4},
  pages        = {772--810},
  year         = {2020},
  url          = {https://doi.org/10.1137/18M1166869},
  doi          = {10.1137/18M1166869},
  bibsource    = {dblp computer science bibliography, https://dblp.org}
}

@inproceedings{DBLP:conf/focs/LinRSW23,
  author       = {Bingkai Lin and
                  Xuandi Ren and
                  Yican Sun and
                  Xiuhan Wang},
  title        = {Improved Hardness of Approximating $k$-{Clique} under {ETH}},
  booktitle    = {64th {IEEE} Annual Symposium on Foundations of Computer Science, {FOCS}
                  2023, Santa Cruz, CA, USA, November 6-9, 2023},
  pages        = {285--306},
  publisher    = {{IEEE}},
  year         = {2023},
  url          = {https://doi.org/10.1109/FOCS57990.2023.00025},
  doi          = {10.1109/FOCS57990.2023.00025},
  bibsource    = {dblp computer science bibliography, https://dblp.org}
}

@article{alekseev1982effect,
  title={The effect of local constraints on the complexity of determination of the graph independence number},
  author={Alekseev, Vladimir E.},
  journal={Combinatorial-algebraic methods in applied mathematics},
  pages={3--13},
  year={1982}
}

@article{DBLP:journals/dam/BrandstadtM18a,
  author       = {Andreas Brandst{\"{a}}dt and
                  Raffaele Mosca},
  title        = {Maximum weight independent set for {$\ell$}claw-free graphs
                  in polynomial time},
  journal      = {Discret. Appl. Math.},
  volume       = {237},
  pages        = {57--64},
  year         = {2018},
  url          = {https://doi.org/10.1016/j.dam.2017.11.029},
  doi          = {10.1016/J.DAM.2017.11.029},
  bibsource    = {dblp computer science bibliography, https://dblp.org}
}

@article{Pach2006,
  title = {Forbidden paths and cycles in ordered graphs and matrices},
  volume = {155},
  ISSN = {1565-8511},
  url = {http://dx.doi.org/10.1007/BF02773960},
  DOI = {10.1007/bf02773960},
  number = {1},
  journal = {Israel Journal of Mathematics},
  publisher = {Springer Science and Business Media LLC},
  author = {Pach,  János and Tardos,  Gábor},
  year = {2006},
  month = dec,
  pages = {359–380}
}

@article{DBLP:journals/combinatorics/BalkoC0K20,
  author       = {Martin Balko and
                  Josef Cibulka and
                  Karel Kr{\'{a}}l and
                  Jan Kyncl},
  title        = {Ramsey Numbers of Ordered Graphs},
  journal      = {Electron. J. Comb.},
  volume       = {27},
  number       = {1},
  pages        = {1},
  year         = {2020},
  url          = {https://doi.org/10.37236/7816},
  doi          = {10.37236/7816},
  bibsource    = {dblp computer science bibliography, https://dblp.org}
}

@article{DBLP:journals/jct/ConlonFLS17,
  author       = {David Conlon and
                  Jacob Fox and
                  Choongbum Lee and
                  Benny Sudakov},
  title        = {Ordered {Ramsey} numbers},
  journal      = {J. Comb. Theory {B}},
  volume       = {122},
  pages        = {353--383},
  year         = {2017},
  url          = {https://doi.org/10.1016/j.jctb.2016.06.007},
  doi          = {10.1016/J.JCTB.2016.06.007},
  bibsource    = {dblp computer science bibliography, https://dblp.org}
}

@inproceedings{DBLP:conf/bcc/Tardos19,
  author       = {G{\'{a}}bor Tardos},
  editor       = {Allan Lo and
                  Richard Mycroft and
                  Guillem Perarnau and
                  Andrew Treglown},
  title        = {Extremal theory of vertex or edge ordered graphs},
  booktitle    = {Surveys in Combinatorics, 2019: Invited lectures from the 27th British
                  Combinatorial Conference, Birmingham, UK, July 29 - August 2, 2019},
  pages        = {221--236},
  publisher    = {Cambridge University Press},
  year         = {2019},
  url          = {https://doi.org/10.1017/9781108649094.008},
  doi          = {10.1017/9781108649094.008},
  bibsource    = {dblp computer science bibliography, https://dblp.org}
}

@article{DBLP:journals/jct/KorandiTTW19,
  author       = {D{\'{a}}niel Kor{\'{a}}ndi and
                  G{\'{a}}bor Tardos and
                  Istv{\'{a}}n Tomon and
                  Craig Weidert},
  title        = {On the {Tur{\'{a}}n} number of ordered forests},
  journal      = {J. Comb. Theory {A}},
  volume       = {165},
  pages        = {32--43},
  year         = {2019},
  url          = {https://doi.org/10.1016/j.jcta.2019.01.006},
  doi          = {10.1016/J.JCTA.2019.01.006},
  bibsource    = {dblp computer science bibliography, https://dblp.org}
}

@article{DBLP:journals/combinatorics/DudekGR24,
  author       = {Andrzej Dudek and
                  Jaros\l{}aw Grytczuk and
                  Andrzej Ruci\'nski},
  title        = {Ordered Unavoidable Sub-Structures in Matchings and Random Matchings},
  journal      = {Electron. J. Comb.},
  volume       = {31},
  number       = {2},
  year         = {2024},
  url          = {https://doi.org/10.37236/11932},
  doi          = {10.37236/11932},
  bibsource    = {dblp computer science bibliography, https://dblp.org}
}

@article{DBLP:journals/siamdm/BoskovicK23,
  author       = {Vladimir Bo\v{s}kovi\'c and
                  Bal{\'{a}}zs Keszegh},
  title        = {Saturation of Ordered Graphs},
  journal      = {{SIAM} J. Discret. Math.},
  volume       = {37},
  number       = {2},
  pages        = {1118--1141},
  year         = {2023},
  url          = {https://doi.org/10.1137/22m1485735},
  doi          = {10.1137/22M1485735},
  bibsource    = {dblp computer science bibliography, https://dblp.org}
}

@article{DBLP:journals/siamdm/ScottSS22,
  author       = {Alex Scott and
                  Paul D. Seymour and
                  Sophie Spirkl},
  title        = {Pure Pairs {VI: Excluding} an Ordered Tree},
  journal      = {{SIAM} J. Discret. Math.},
  volume       = {36},
  number       = {1},
  pages        = {170--187},
  year         = {2022},
  url          = {https://doi.org/10.1137/20m1368331},
  doi          = {10.1137/20M1368331},
  bibsource    = {dblp computer science bibliography, https://dblp.org}
}

@article{Tara,
  author       = {Tara Abrishami and
                  Maria Chudnovsky and
                  Marcin Pilipczuk and
                  Paweł Rzążewski and
                  Paul D. Seymour},
  title        = {Induced Subgraphs of Bounded Treewidth and the Container Method},
  journal      = {{SIAM} J. Comput.},
  volume       = {53},
  number       = {3},
  pages        = {624--647},
  year         = {2024},
  url          = {https://doi.org/10.1137/20m1383732},
  doi          = {10.1137/20M1383732},
  bibsource    = {dblp computer science bibliography, https://dblp.org}
}

@inproceedings{DBLP:conf/stoc/GartlandLPPR21,
  author       = {Peter Gartland and
                  Daniel Lokshtanov and
                  Marcin Pilipczuk and
                  Michał Pilipczuk and
                  Paweł Rzążewski},
  editor       = {Samir Khuller and
                  Virginia Vassilevska Williams},
  title        = {Finding large induced sparse subgraphs in {$C_{>t}$}-free graphs in quasipolynomial time},
  booktitle    = {{STOC} '21: 53rd Annual {ACM} {SIGACT} Symposium on Theory of Computing,
                  Virtual Event, Italy, June 21-25, 2021},
  pages        = {330--341},
  publisher    = {{ACM}},
  year         = {2021},
  url          = {https://doi.org/10.1145/3406325.3451034},
  doi          = {10.1145/3406325.3451034},
  bibsource    = {dblp computer science bibliography, https://dblp.org}
}

@inproceedings{DBLP:conf/soda/ChudnovskyMPPR24,
  author       = {Maria Chudnovsky and
                  Rose McCarty and
                  Marcin Pilipczuk and
                  Michał Pilipczuk and
                  Paweł Rzążewski},
  editor       = {David P. Woodruff},
  title        = {Sparse induced subgraphs in {$P_6$}-free graphs},
  booktitle    = {Proceedings of the 2024 {ACM-SIAM} Symposium on Discrete Algorithms,
                  {SODA} 2024, Alexandria, VA, USA, January 7-10, 2024},
  pages        = {5291--5299},
  publisher    = {{SIAM}},
  year         = {2024},
  url          = {https://doi.org/10.1137/1.9781611977912.190},
  doi          = {10.1137/1.9781611977912.190},
  bibsource    = {dblp computer science bibliography, https://dblp.org}
}

@article{DBLP:journals/siamdm/PaesaniPR22,
  author       = {Giacomo Paesani and
                  Dani{\"{e}}l Paulusma and
                  Pawe\l{} Rz\k{a}\.zewski},
  title        = {{Feedback Vertex Set} and {Even Cycle Transversal} for {$H$}-Free
                  Graphs: Finding Large Block Graphs},
  journal      = {{SIAM} J. Discret. Math.},
  volume       = {36},
  number       = {4},
  pages        = {2453--2472},
  year         = {2022},
  url          = {https://doi.org/10.1137/22m1468864},
  doi          = {10.1137/22m1468864},
  bibsource    = {dblp computer science bibliography, https://dblp.org}
}

@article{DBLP:journals/algorithmica/DabrowskiFJPPR20,
  author       = {Konrad K. Dabrowski and
                  Carl Feghali and
                  Matthew Johnson and
                  Giacomo Paesani and
                  Dani{\"{e}}l Paulusma and
                  Pawe\l{} Rz\k{a}\.zewski},
  title        = {On Cycle Transversals and Their Connected Variants in the Absence
                  of a Small Linear Forest},
  journal      = {Algorithmica},
  volume       = {82},
  number       = {10},
  pages        = {2841--2866},
  year         = {2020},
  url          = {https://doi.org/10.1007/s00453-020-00706-6},
  doi          = {10.1007/S00453-020-00706-6},
  bibsource    = {dblp computer science bibliography, https://dblp.org}
}

@article{DBLP:journals/talg/0001LLR0S25,
  author       = {Akanksha Agrawal and
                  Paloma T. Lima and
                  Daniel Lokshtanov and
                  Pawe\l{} Rz\k{a}\.zewski and
                  Saket Saurabh and
                  Roohani Sharma},
  title        = {{Odd Cycle Transversal} on {$P_5$}-free Graphs in
                  Polynomial Time},
  journal      = {{ACM} Trans. Algorithms},
  volume       = {21},
  number       = {2},
  pages        = {16:1--16:14},
  year         = {2025},
  url          = {https://doi.org/10.1145/3708544},
  doi          = {10.1145/3708544},
  bibsource    = {dblp computer science bibliography, https://dblp.org}
}

@article{DBLP:journals/siamdm/ChudnovskyKPRS21,
  author       = {Maria Chudnovsky and
                  Jason King and
                  Michal Pilipczuk and
                  Pawe\l{} Rz\k{a}\.zewski and
                  Sophie Spirkl},
  title        = {Finding Large {$H$}-Colorable Subgraphs in Hereditary Graph Classes},
  journal      = {{SIAM} J. Discret. Math.},
  volume       = {35},
  number       = {4},
  pages        = {2357--2386},
  year         = {2021},
  url          = {https://doi.org/10.1137/20M1367660},
  doi          = {10.1137/20M1367660},
  bibsource    = {dblp computer science bibliography, https://dblp.org}
}

@inproceedings{DBLP:conf/esa/GalbyLMN25,
  author       = {Esther Galby and
                  Paloma T. Lima and
                  Andrea Munaro and
                  Amir Nikabadi},
  editor       = {Anne Benoit and
                  Haim Kaplan and
                  Sebastian Wild and
                  Grzegorz Herman},
  title        = {{Maximum List $r$-Colorable Induced Subgraphs} in {$kP_3$}-Free
                  Graphs},
  booktitle    = {33rd Annual European Symposium on Algorithms, {ESA} 2025, Warsaw,
                  Poland, September 15-17, 2025},
  series       = {LIPIcs},
  pages        = {40:1--40:13},
  publisher    = {Schloss Dagstuhl - Leibniz-Zentrum f{\"{u}}r Informatik},
  year         = {2025},
  url          = {https://doi.org/10.4230/LIPIcs.ESA.2025.40},
  doi          = {10.4230/LIPICS.ESA.2025.40},
  bibsource    = {dblp computer science bibliography, https://dblp.org}
}

@article{HeathLeightonRosenberg1992,
  author    = {Lenwood S. Heath and F. Thomson Leighton and Arnold L. Rosenberg},
  title     = {Comparing queues and stacks as mechanisms for laying out graphs},
  journal   = {SIAM Journal on Discrete Mathematics},
  volume    = {5},
  number    = {3},
  pages     = {398--412},
  year      = {1992},
  doi       = {10.1137/0405031}
}

@inproceedings{DBLP:conf/focs/GartlandL20,
  author       = {Peter Gartland and
                  Daniel Lokshtanov},
  editor       = {Sandy Irani},
  title        = {Independent Set on {$P_k$}-Free
                  Graphs in Quasi-Polynomial Time},
  booktitle    = {61st {IEEE} Annual Symposium on Foundations of Computer Science, {FOCS}
                  2020, Durham, NC, USA, November 16-19, 2020},
  pages        = {613--624},
  publisher    = {{IEEE}},
  year         = {2020},
  url          = {https://doi.org/10.1109/FOCS46700.2020.00063},
  doi          = {10.1109/FOCS46700.2020.00063},
  bibsource    = {dblp computer science bibliography, https://dblp.org}
}

@inproceedings{DBLP:conf/stoc/GartlandLMPPR24,
  author       = {Peter Gartland and
                  Daniel Lokshtanov and
                  Tom{\'{a}}s Masa\v{r}{\'{\i}}k and
                  Marcin Pilipczuk and
                  Michał Pilipczuk and
                  Paweł Rzążewski},
  editor       = {Bojan Mohar and
                  Igor Shinkar and
                  Ryan O'Donnell},
  title        = {Maximum Weight Independent Set in Graphs with no Long Claws in Quasi-Polynomial
                  Time},
  booktitle    = {Proceedings of the 56th Annual {ACM} Symposium on Theory of Computing,
                  {STOC} 2024, Vancouver, BC, Canada, June 24-28, 2024},
  pages        = {683--691},
  publisher    = {{ACM}},
  year         = {2024},
  url          = {https://doi.org/10.1145/3618260.3649791},
  doi          = {10.1145/3618260.3649791},
  bibsource    = {dblp computer science bibliography, https://dblp.org}
}
